\documentclass[a4paper,11pt,]{article}
\usepackage{graphicx}
\usepackage{tabularx}
 \usepackage{amssymb}
 \usepackage{amsmath}
\usepackage{multirow}
\usepackage{supertabular}

\newtheorem{THM}{Theorem}[section]
\newtheorem{LMA}[THM]{Lemma}

\global \count10 = 0

\global \count11 = 1

\global \count12 = 1

\global \count13 = 1

\newcommand{\com}[1]{\ifnum\count13<1 #1 \fi}

\global \count14 = 0

\def\squarebox#1{\hbox to #1{\hfill\vbox to #1{\vfill}}}
\def\qed{\hspace*{\fill}%
        \vbox{\hrule\hbox{\vrule\squarebox{.667em}\vrule}\hrule}\smallskip}
\newenvironment{proof}{\begin{trivlist}
\item[\hspace{\labelsep}{\em\noindent Proof.~}]}{\qed\end{trivlist}}

\def\squarebox#1{\hbox to #1{\hfill\vbox to #1{\vfill}}}
\def\qed{\hspace*{\fill}%
        \vbox{\hrule\hbox{\vrule\squarebox{.667em}\vrule}\hrule}\smallskip}

\begin{document}

\title{Better Bounds for Online $k$-Frame Throughput Maximization\\in Network Switches}

\author{
	Jun Kawahara$^{1}$, 
	Koji M. Kobayashi$^{2}$, and 
	Shuichi Miyazaki$^{3}$
	\\   
	{\footnotesize 
		$^{1}$Graduate School of Information Science, Nara Institute of Science and Technology 
	}
	\\
	{\footnotesize 
		$^{2}$
		Corresponding author, 
	}
	{\footnotesize 
		National Institute of Informatics, 
	}
	\\
	{\footnotesize 
		2-1-2, Hitotsubashi, Chiyoda-ku, 
		Tokyo, 1018430, Japan, 
	}
	{\footnotesize 
		kobaya@nii.ac.jp
	}
	\\   
	{\footnotesize 
		$^{3}$Academic Center for Computing and Media Studies, Kyoto University
	}
}

\date{}

\setlength{\baselineskip}{4.85mm}
\maketitle

\begin{abstract}
\ifnum \count10 > 0
ネットワーク上におけるスイッチのキュー管理は、オンラインバッファ管理問題として定式化されており、多くのモデルが考案されている。我々は、そのモデルの1つである、$k(\geq 2)$個のパケットからなるフレーム転送量最大化問題（$k$-FTM）を扱う。この問題では、データ転送の単位である各フレームをパケットに分割し、転送完了後にフレームを再構成する際の手順に焦点を当てている。各フレームは$k$個のパケットから構成されており、そのパケットがスイッチ内のFIFOバッファに到着する。1つのフレームを構成する$k$個のパケットをバッファから転送した場合、そのフレームを再構成することが可能となる。しかし、バッファのサイズは制限されているため、到着したパケットは全てを転送できない場合がある。本問題の目的は、再構成するフレームの数を最大化することである。
Kesselmanらは本問題に対して、$k=2$の場合でさえ$k$-FTMに対する任意のアルゴリズムの競合比は発散することを示している。
そこで、入力に制限を加えた$k$-FTM（$k$-OFTM）を考え、任意のバッファの大きさ$B (\geq k)$に対して、その競合比が高々$\frac{ 2kB }{ \lfloor B/k \rfloor + k}$となるアルゴリズムを考案し、また、任意の$2B (\geq k)$に対して、任意のアルゴリズムの競合比の下限が少なくとも$\frac{ B }{ \lfloor 2B/k \rfloor}$であることを示した。
本稿において、$k$-OFTMに対して我々は以下の結果得ている。
(i) 任意の$k \geq 2$と任意の$B$の場合に、決定性アルゴリズムに対する$2k-1$となる下限を示す。
	これは、先の$\frac{ B }{ \lfloor 2B/k \rfloor }$という下限を改良している。
(ii) $B \geq 2k$の場合に、競合比が$\frac{5B + \lfloor B/k \rfloor - 4}{\lfloor \lfloor B/k \rfloor / 2 \rfloor} = O(k)$
	となる決定性アルゴリズムを設計している。これはKesselmanらの下限の値と我々の(i)の下限に漸近的に一致する。
(iii) 任意の$k \geq 3$と任意の$B$に対して、
	$k-1$となる確率アルゴリズムに対する下限を示す。
	これは、(ii)の上限と漸近的に一致している。
\fi
\ifnum \count11 > 0
\com{（■英語）}
	We consider a variant of the online buffer management problem in network switches, called the $k$-frame throughput maximization problem ($k$-FTM).  This problem models the situation where a large frame is fragmented into $k$ packets and transmitted through the Internet, and the receiver can reconstruct the frame only if he/she accepts all the $k$ packets. 
	Kesselman et~al.\ introduced this problem and showed that its competitive ratio is unbounded even when $k=2$.
	They also introduced an ``order-respecting'' variant of $k$-FTM, called $k$-OFTM, where inputs are restricted in some natural way. They proposed an online algorithm and showed that its competitive ratio is at most $\frac{ 2kB }{ \lfloor B/k \rfloor } + k$ for any $B \ge k$, where $B$ is the size of the buffer. They also gave a lower bound of $\frac{ B }{ \lfloor 2B/k \rfloor }$ for deterministic online algorithms when $2B \geq k$ and $k$ is a power of $2$. 
	
	In this paper, 
	we improve upper and lower bounds on the competitive ratio of $k$-OFTM. 
	Our main result is to improve an upper bound of $O(k^{2})$ by Kesselman et~al.\ to $\frac{5B + \lfloor B/k \rfloor - 4}{\lfloor B/2k \rfloor} = O(k)$ for $B\geq 2k$. 
	Note that this upper bound is tight up to a multiplicative constant factor since the lower bound given by
	Kesselman et~al.\ is $\Omega(k)$. 
	We also give two lower bounds. 
	First we give a lower bound of $\frac{2B}{\lfloor {B/(k-1)} \rfloor} + 1$
	on the competitive ratio of deterministic online algorithms for any $k \geq 2$ and any $B \geq k-1$, 
	which improves the previous lower bound of $\frac{B}{ \lfloor 2B/k \rfloor }$ by a factor of almost four. 
	Next, 
	we present the first nontrivial lower bound on the competitive ratio of randomized algorithms. 
	Specifically, 
	we give a lower bound of $k-1$ against an oblivious adversary for any $k \geq 3$ and any $B$. 
	Since a deterministic algorithm, as mentioned above, 
	achieves an upper bound of about $10k$, 
	this indicates that randomization does not help too much. 
\fi
\end{abstract}

\newpage

\section{Introduction} \label{Intro}
\ifnum \count10 > 0
\com{（■日本語）\\}
（■未修正）
ネットワークを介して動画のような巨大なデータを送信する需要が高まっている．
このような巨大データは多くのデータ断片から構成され，それらは{\em フレーム}と呼ばれる．
小さな転送単位を用いているネットワークでフレームを転送するために，
フレームをさらにいくつかのパケットに分割する必要がある．
受信側では，1つのフレームの全パケットを受信するならば，そのフレームを再構成できる．
このように，上記の状況ではパケット間に依存関係が存在する．
ネットワークが混雑している場合，いくつかのパケットはバッファからあふれるため，
破棄しなければならない．
上記のシナリオの下で，出来る限り多くのフレームを
再構築できるようにパケットを転送したい．
近年，
これらの問題においてバッファ管理に焦点をあててオンライン問題として定式化し，
それに対するオンラインアルゴリズムの研究が盛んに行われている．
最も基本的なモデルは次の通りである\cite{WA00}．
各パケットは何らかの価値を持っており，スイッチは大きさ$B$のバッファを持っている．
問題の入力はイベントの列から構成される．
各イベントは，到着イベントか送信イベントである．
到着イベントでは，
スイッチの入力ポートにパケットが1つ到着する．
パケットは単位長であり，その価値はパケットの優先度を示す．
スイッチのバッファに同時に$B$個以下のパケットを受理して蓄えることが出来る．
蓄えられたパケットはFIFOでスイッチから送出される．
到着イベントの際，バッファが満杯であるなら，新しく到着したパケットは非受理となる．
バッファに空きがあるならば，アルゴリズムは受理するかどうかをオンラインに判断することになる．
各送信イベントには，
FIFOキューの先頭からパケットを1つ送信する．
本問題の目的は，送信されるパケットの価値の総和の最大化である．
アルゴリズムの性能評価には競合比解析\cite{AB98,DS85}が用いられる．競合
比解析においては，オンラインアルゴリズムの得た価値を，入力が全て分かっ
ている場合の，最適なオフラインアルゴリズム（これを以下では$OPT$と書く）
の得る価値と比較する．任意の入力に対して，オンラインアルゴリズムが得る
価値が，$OPT$の得る価値の$1/c$倍以上である場合に，そのオンラインアルゴ
リズムの競合比は$c$であると言う．
Kesselmanらはフレームの再構築を考慮したオンラインバッファ管理問題を，
$k$フレーム転送量最大化問題（$k$-FTM）として
以下のように定式化した\cite{AK09}．ここで，$k$は$2$以上の整数である．
フレーム$f$内の全パケットを転送した場合，
フレーム$f$は{\em 完全である}と定義する．我々の目標はアルゴリズムの利得，すなわち完全なフレーム総数の
最大化である．
\fi
\ifnum \count11 > 0
\com{（■英語）\\}
When transmitting data through the Internet, 
each data is fragmented into smaller pieces, and such pieces are encapsulated into data packets.
Packets are transmitted to the receiver via several switches and routers over a network, 
and are reconstructed into the original data at the receiver's side. 
One of the bottlenecks in achieving high throughput is processing ability of switches and routers. 
If the arrival rate of packets exceeds the processing rate of a switch, 
some packets must be dropped. 
To ease this inconvenience, 
switches are usually equipped with FIFO buffers that temporarily store packets which will be processed later. 
In this case, 
the efficiency of buffer management policies is important since it affects the performance of the overall network. 
Aiello et~al.\ \cite{WA00} initiated the analysis of buffer management problem using the {\em competitive analysis}~\cite{AB98,DS85}: 
An input of the problem is a sequence of events where each event is an arrival event or a send event. 
At an arrival event, one packet arrives at an input port of the buffer (FIFO queue). 
Each packet is of unit size and has a positive value that represents its priority. 
A buffer can store at most $B$ packets simultaneously. 
At an arrival event, 
if the buffer is full, the new packet is rejected. 
If there is room for the new packet, 
an online algorithm determines whether to accept it or not without knowing the future events. 
At each send event, the packet at the head of the queue is transmitted. 
The gain of an algorithm is the sum of the values of the transmitted packets, 
and the goal of the problem is to maximize it. 
If, for any input $\sigma$, 
the gain of an online algorithm $ALG$ is at least $1/c$ of the gain of an optimal offline algorithm for $\sigma$, 
then we say that $ALG$ is $c$-competitive.

Following the work of Aiello et~al.\ \cite{WA00}, 
there has been a great amount of work related to the competitive analysis of buffer management. 
For example, Andelman et~al.\ \cite{NACQ03} generalized the two-value model of \cite{WA00} into the multi-value model 
in which the priority of packets can take arbitrary values. 
Another generalization is to allow {\em preemption}, 
i.e., an online algorithm can discard packets existing in the buffer. 
Results of the competitiveness on these models are given in \cite{AKB01,MS01,AKI02,NACM03,NA05,ME06}. 
Also, management policies not only for a single queue but also for the whole switch are
extensively studied, which includes multi-queue switches \cite{YA03,NACQ03,SA04,YAM04,KK09,MB10}, 
shared-memory switches \cite{EH01,AKH01,KK07}, 
CIOQ switches \cite{AKS03,YA06,AKI08,AKC08}, 
and crossbar switches \cite{AKP08,AKB08}. 
See \cite{MG10} for a comprehensive survey.

Kesselman et~al.\ \cite{AK09} proposed another natural extension, called the {\em $k$-frame throughput maximization} problem ($k$-FTM), 
motivated by a scenario of reconstructing the original data from data packets at the receiver's side. 
In this model, a unit of data, called a {\em frame}, is fragmented into $k$ packets 
(where the $j$th packet of the frame is called a $j$-packet for $j\in [1,k]$) 
and transmitted through the Internet. 
At the receiver's side, 
if all the $k$ packets (i.e., the $j$-packet of the frame for all $j$) are received, 
the frame can be reconstructed (in such a case, we say that the frame is {\em completed}); 
otherwise, 
even if one of them is missing, the receiver can obtain nothing.
The goal is to maximize the number of completed frames. 
Kesselman et~al.\ \cite{AK09} considered this scenario on a single FIFO queue. 
They first showed that the competitive ratio of any deterministic algorithm for $k$-FTM is unbounded even when $k=2$ 
(which can also be applied to randomized algorithms with a slight modification). 
However, 
their lower bound construction somehow deviates from the real-world situation, 
that is, 
although each packet generally arrives in order of departure in a network such as a TCP/IP network, 
in their adversarial input sequence 
the 1-packet of the frame $f_{i}$ arrives prior to that of the frame $f_{i'}$, 
while the 2-packet of $f_{i'}$ arrives before that of $f_{i}$. 
Motivated by this, 
they introduced a natural setting for the input sequence, 
called the {\em order-respecting} adversary, 
in which, roughly speaking, 
the arrival order of the $j$-packets of $f_{i}$ and $f_{i'}$ must obey the arrival order of the $j'$-packets of $f_{i}$ and $f_{i'}$ ($j' < j$) 
(a formal definition will be given in Sec.~\ref{model}). 
We call this restricted problem the {\em order-respecting $k$-frame throughput maximization} problem ($k$-OFTM). 
For $k$-OFTM, 
they showed that the competitive ratio of any deterministic algorithm is at least $B/\lfloor 2B/k \rfloor$ 
when $2B \geq k$ and $k$ is a power of $2$. 
As for an upper bound, 
they designed a non-preemptive algorithm called {\sc StaticPartitioning} ($SP$), 
and showed that its competitive ratio is at most $\frac{ 2kB }{ \lfloor B/k \rfloor } + k$ for any $B \ge k$. 
\fi
\subsection{Our Results}
In this paper, we present the following results:

(i) 
We design a deterministic algorithm {\sc Middle-Drop and Flush}
($MF$) for $B \geq 2k$, and show that its competitive ratio is at most
$\frac{5B + \lfloor B/k \rfloor - 4}{\lfloor B/2k \rfloor}$. 
Note that this ratio is $O(k)$, which improves $O(k^{2})$ of Kesselman et~al.\
\cite{AK09} and matches the lower bound of $\Omega(k)$ up to a constant
factor.

(ii) 
For any deterministic algorithm, 
we give a lower bound of $\frac{2B}{\lfloor {B/(k-1)} \rfloor} + 1$ on the competitive ratio
for any $k \geq 2$ and any $B \geq k-1$.  This improves the previous
lower bound of $\frac{B}{\lfloor 2B/k \rfloor}$ by a factor of almost four. 
Moreover, we show that the competitive ratio of any deterministic
online algorithm is unbounded if $B \leq k-2$.

(iii) 
In the randomized setting, 
we establish the first nontrivial lower bound of $k-1$ against an oblivious adversary for any $k \geq 3$ and any $B$.
This bound matches our deterministic upper bound mentioned in (i) up to a constant factor,
which implies that randomization does not help for this problem.
\subsection{Used Techniques} \label{sec:idea}
\ifnum \count10 > 0
\fi
\ifnum \count11 > 0
\com{（■英語）}
Let us briefly explain an idea behind our algorithm $MF$. 
The algorithm $SP$ by Kesselman et~al.\ \cite{AK09} works as follows: 
(1) 
It virtually divides its buffer evenly into $k$ subbuffers, each with size $A=\lfloor
\frac{B}{k} \rfloor$, and each subbuffer (called {\em $j$-subbuffer} for
$j \in [1,k]$) is used for storing only $j$-packets. 
(2) 
If the $j$-subbuffer overflows, i.e., 
if a new $j$-packet arrives when $A$
$j$-packets are already stored in the $j$-subbuffer, it rejects the newly arriving $j$-packet (the ``tail-drop'' policy). 
It can be shown that $SP$ behaves poorly when a lot of $j$-packets arrive at a burst, 
which increases $SP$'s competitive ratio as bad as $\Omega(k^{2})$
(such a bad example for $SP$ is given in Appendix~\ref{sec:ap.1}). 
In this paper, 
we modify the tail-drop policy and employ the ``middle-drop'' policy, which preempts the $(\lfloor A/2 \rfloor +1)$st packet in the $j$-subbuffer and accepts the newly arriving $j$-packet, 
which is crucial in improving the competitive ratio to $O(k)$, 
as explained in the following.
$MF$ partitions the whole set of given frames into {\em blocks} $BL_{1}, BL_{2},
\ldots$, each with about $3B$ frames, using the rule concerning the
arrival order of 1-packets. 
(This rule is explained in Sec.~\ref{algorithm_k} at the definition of $MF$, 
where the block $BL_{i}$ corresponds to the set of frames with the {\em block number} $i$.) 
Each block is categorized into {\em good} or {\em bad}: 
At the beginning of the input, 
all the blocks are good. 
At some moment during the execution of $MF$, 
if there is no more possibility of completing at least $\lfloor A/2 \rfloor$ frames of a block $BL_{i}$ 
(as a result of preemptions and/or rejections of packets in $BL_{i}$), 
then $BL_{i}$ turns bad. 
In such a case, 
$MF$ completely gives up $BL_{i}$ and preempts all the packets belonging to $BL_{i}$ in its buffer if any
(which is called the ``flush'' operation). 
Note that at the end of input, 
$MF$ completes at least $\lfloor A/2 \rfloor$ frames of a good block.

Consider the moment when the block $BL_{i}$ turns bad from good, 
which can happen only when preempting a $j$-packet $p$ (for some $j$) of $BL_{i}$ from the $j$-subbuffer. 
Due to the property of the middle-drop policy, 
we can show that there exist two integers $i_{1}$ and $i_{2}$ ($i_{1} < i < i_{2}$) such that  
(i) just after this flush operation, $BL_{i_{1}}$ and $BL_{i_{2}}$ are
good and all the blocks $BL_{i_{1}+1}, BL_{i_{1}+2}, \ldots,
BL_{i_{2}-1}$ are bad, and
(ii) just before this flush operation, all the $j$-packets of $BL_{i}$
(including $p$) each of which belongs to a frame 
that still has a chance of being completed are located between $p_{1}$ and $p_{2}$, 
where $p_{1}$ and $p_{2}$ are $j$-packets in the buffer belonging to $BL_{i_{1}}$ and $BL_{i_{2}}$, respectively. 
The above (ii) implies that 
even though $i_{2}$ may be much larger than $i_{1}$ 
(and hence there may be many blocks between $BL_{i_{1}}$ and $BL_{i_{2}}$), 
the arrival times of $p_{1}$ and $p_{2}$ are close 
(since $p_{1}$ is still in the buffer when $p_{2}$ arrived). 
This means that $j$-packets of $BL_{i_{1}}$ through $BL_{i_{2}}$ arrived at a burst within a very short span, 
and hence any algorithm (even an optimal offline algorithm $OPT$) cannot accept
many of them. 
In this way, 
we can bound the number of packets accepted by $OPT$ 
(and hence the number of frames completed by $OPT$) 
between two consecutive good blocks. 
More precisely, 
if $BL_{i_1}$ and $BL_{i_2}$ are consecutive good blocks at the end of the input, 
we can show that the number of frames in $BL_{i_1}, BL_{i_1+1}, \ldots, BL_{i_2-1}$ completed by $OPT$ is at most $5B+A-4=O(B)$ using (i). 
Recall that $MF$ completes at least $\lfloor A/2 \rfloor=\Omega(B/k)$ frames of $BL_{i_1}$ since $BL_{i_1}$ is good, 
which leads to the competitive ratio of $O(k)$. 
\fi

\subsection{Related Results}
In addition to the above mentioned results,
Kesselman et~al.\ \cite{AK09} proved that for any $B$, the competitive
ratio of a preemptive greedy algorithm for $k$-OFTM is unbounded when $k \ge 3$.
They also considered offline version of $k$-FTM and proved the approximation hardness.
Recently, 
Kawahara and Kobayashi~\cite{JK13} proved that the optimal competitive ratio of 2-OFTM is 3, which is achieved by a greedy algorithm.

Scalosub et~al.\ \cite{GS10} proposed a generalization of $k$-FTM, called
the {\em max frame goodput} problem.  In this problem, a set of frames
constitute a {\em stream}, and a constraint is imposed on the arrival
order of packets within the same stream.  They established an
$O((kMB+M)^{k+1})$-competitive deterministic algorithm, where $M$
denotes the number of streams.  Furthermore, they showed that the
competitive ratio of any deterministic algorithm is $\Omega(kM/B)$.

Emek et~al.\ \cite{YE10} introduced the {\em online set packing} problem.
This problem is different from $k$-FTM in that each frame may consist of
different number (at most $k_{\mbox{\scriptsize max}}$) of packets.
Also, a frame $f$ consisting of $s(f)$ packets can be reconstructed if
$s(f)(1 - \beta)$ packets are transmitted, where $\beta$ ($0 \le \beta <
1$) is a given parameter.  There is another parameter $c$ representing
the  capacity of a switch.  At an arrival event, several packets arrive at
an input port of the queue. 
A switch can transmit $c$ of them instantly, 
and operates a buffer management algorithm for the rest of the packets, 
that is, decides whether to accept them 
(if any). 
Emek et~al.\ designed a randomized algorithm {\sc Priority},
and showed that it is $k_{\mbox{\scriptsize max}}
\sqrt{\sigma_{\mbox{\scriptsize max}}}$-competitive when $\beta = 0$ and
$B=0$, where $\sigma_{\mbox{\scriptsize max}}$ is the maximum number of
packets arriving simultaneously.  They also derived a lower bound of
$k_{\mbox{\scriptsize max}} \sqrt{\sigma_{\mbox{\scriptsize max}}} (\log
\log k / \log k)^2$ for any randomized algorithm.  If the number of
packets in any frame is exactly $k$, Mansour et~al.\ \cite{YMO11} showed
that for any $\beta$ the competitive ratio of {\sc Priority} is $8k
\sqrt{\sigma_{\mbox{\scriptsize max}} (1 - \beta) / c} $.  Moreover,
some variants of this problem have been studied \cite{MH11,YMC11}. 
%

\section{Model Description and Notation} \label{model}

%

%
\ifnum \count10 > 0
（■未修正）
本節では、$k$-OFTMの正式な定義を与える。
到着するパケットはバッファに蓄えられる。
各パケットの大きさは1であり、
そのバッファはFIFOキューで構成され、同時に高々$B$個のパケットを蓄えることが出来る。
入力はフェイズの列である。
各フェイズは、2つのサブフェイズからなる。
{\em 到着サブフェイズ}と{\em 送信サブフェイズ}からなる。 
到着サブフェイズでは、
幾つかのパケットがバッファに到着する。
アルゴリズムの役割は、任意の到着するパケット$p$に対して、
$p$を{\em 受理する}、もしくは$p$を{\em 非受理する}か決定することである。
アルゴリズムは、バッファ内に存在する任意のパケット$p'$を破棄し、
到着するパケットを受理するために、バッファに空きを作ることが出来る。
（$p'$を{\em プリエンプトする}と言う。）
バッファに空きがなければ、到着したパケットは非受理される。
受理されたパケットは、キューの最後尾に蓄えられる。
同じサブフェイズに到着して受理されたパケットは、任意の順序でキューに入れることができる。
到着サブフェイズの後に、送信フェイズが起こる。
送信サブフェイズには、バッファが空でなければ、
キューの先頭のパケットが送信される。
一般性を失わないので、各フェイズは整数時間に起こるとする。
各フレーム$f$は、$k$個のパケット$p_{1}, \ldots, p_{k}$からなる。
（$k$は2以上の整数である。）
構成する$k$個のパケットについて$\mbox{arr}(p_1) \le \cdots \le \mbox{arr}(p_k)$が成立する様な、
任意のフレーム$f = \{ p_{1}, \ldots, p_{k} \}$に対して、
パケット$p_{i}$を{\em $i$-パケット}（$i \in [1, k]$）と呼ぶ。 
任意の2つのフレーム$f_{i} = \{ p_{i,1}, \ldots, p_{i,k} \}, f_{i'} = \{ p_{i',1}, \ldots, p_{i',k} \}$
（ただし、任意の$\ell \in [1, k]$に対して、$p_{i,\ell}$と$p_{i',\ell}$は$\ell$-パケットである。）、
任意の整数$j,\, j' = 1,\ldots,k$に対して、
$\mbox{arr}(p_{i,j}) \le \mbox{arr}(p_{i',j}) \Leftrightarrow \mbox{arr}(p_{i,j'}) \le \mbox{arr}(p_{i',j'})$
が成立する。
この性質を{\em 順序遵守}と呼ぶ。
任意のアルゴリズム$ALG$、$ALG$の任意のフレーム$f = \{ p_{1}, \ldots, p_{k} \}$に対して、 
$f$を構成する全てのパケット、すなわち、$p_{1}, \ldots, p_{k}$が$ALG$に転送される場合、
$ALG$の$f$は{\em 完全である}という。
アルゴリズムの{\em 価値}は、完全なフレームの数である。 
それ故に、本問題の目的は、
完全なフレームの数の最大化である。
$V_{A}(\sigma)$ (${\mathbb E}[V_{ALG}(\sigma)]$, respectively)は、入力$\sigma$に対して決定性（確率, respectively）アルゴリズム$A$の得る価値(価値の期待値, respectively)を表す。
$OPT$を最適なオフライン・アルゴリズムとする。
任意の入力$\sigma$に対して、
$V_{A}(\sigma) \geq V_{OPT}(\sigma)/c$
(${\mathbb E}[V_{ALG}(\sigma)] \geq V_{OPT}(\sigma)/c$, respectively)
が成立するとき、
$A$の{\em 競合比}は$c$であるという。
一般性を失わないので、
$OPT$は決してプリエンプトを行わないと仮定する。
また、
$OPT$は決して不完全なフレームを構成するパケットを受理しないと仮定する。
表記を簡単にするため、
同じフレームに属する任意の$i$-パケット$p$と$j(\ne i)$-パケット$p'$に対して、
$p$の$j$-パケット
（$p'$の$i$-パケット）
と言うこととする。
また、
任意の$ALG (\in \{ MF, OPT \})$に対して、
$ALG$のバッファに到着するパケットを、
$ALG$のパケットと呼ぶ。
更に、必ずパケットは1つ以上到着する入力のみを考える。
以下の解析において用いる記法を導入する。
任意の整数時間$t$に対して、
$t$における到着サブフェイズに到着した$n$個のパケット$p_{1}, p_{2}, \ldots, p_{n}$を
順にバッファに入れるかどうか$MF$が決定するとする。
このとき、
配送フェイズにパケットが転送される瞬間を$t_{d}$で表し、
{\em delivery event time}と呼ぶ。
$p_{i} \hspace{1mm} (1 \leq i \leq n)$の受理するかどうかを判断する瞬間を$t_{p_{i}}$で表す。
更に、
$t_{p_{i}} \hspace{1mm} (1 \leq i \leq n)$を{\em arrival event time}と呼ぶ。
delivery event timeもしくは、arrival event timeのことを{\em event time}と呼ぶ。
event timeでない時刻をnon-event timeと呼ぶ。
時刻$t$に起こる全てのevent timeに対して、
$t_{p_{1}} < t_{p_{2}} < \cdots < t_{p_{n}} < t_{d}$
が成立すると定義する。
また、
任意の整数時間$t, t'(>t)$に対して、
$t$の任意のevent time $e$
$t'$の任意のevent time $e'$
に対して、
$e < e'$が成立すると定義する。
任意のevent time $e$、non-event time $d$ に対して、
$d$の直前のevent timeを$e'$とすると、
$e \le e'$ ならば $e<d$
が成立すると定義する。
また、任意のevent time $e$、non-event time $d$ に対して、
$d$の直後のevent timeを$e''$とすると、
$e'' \le e$ ならば $d<e$
が成立すると定義する。
また、
任意のevent time $e$に対して、
$e$より後であり、次のevent timeより前の任意の瞬間を$e+$、
$e$より前であり、1つ前のevent timeより後の任意の瞬間を$e-$で表す。
また、簡単のため、
任意の整数時間$t'$に対して、
$t'+ = t'_{d}+$と定義する。
また、解析のために、任意の到着するパケット$p$に対して、
$OPT$は$MF$と同時に受理するかどうかの判断を行うものとする。
これらを定義する目的は、
パケットが受理されたり、転送されたりする瞬間の前後におけるバッファ内に存在するパケットの位置を、
解析のために明確にする必要があるからである。
\fi
\ifnum \count11 > 0
\com{（■英語）}
In this section, we give a formal description of the {\em
order-respecting $k$-frame throughput maximization} problem ($k$-OFTM).
A {\em frame} $f$ consists of $k$ packets $p_{1}, \ldots, p_{k}$.  We
say that two packets $p$ and $q$ belonging to the same frame are {\em
corresponding}, or $p$ {\em corresponds to} $q$.  There is one buffer
(FIFO queue), which can store at most $B$ packets simultaneously.  An
input is a sequence of {\em phases} starting from the 0th phase.  The
$i$th phase consists of the $i$th {\em arrival subphase} followed by the
$i$th {\em delivery subphase}.  At an arrival subphase, some packets
arrive at the buffer, and the task of an algorithm is to decide for each
arriving packet $p$, whether to {\em accept} $p$ or {\em reject} $p$.
An algorithm can also discard a packet $p'$ existing in the current
buffer in order to make space (in which case we say that the algorithm
{\em preempts} $p'$).  If a packet $p$ is rejected or preempted, we say
that $p$ is {\em dropped}.  If a packet is accepted, it is stored at the
tail of the queue.  Packets accepted at the same arrival subphase can be
inserted into the queue in an arbitrary order.  At a delivery subphase,
the first packet of the queue is transmitted if the buffer is nonempty.
For a technical reason, 
we consider only the inputs in which at least one packet arrives.

If a packet $p$ arrives at the $i$th arrival subphase, 
we write $\mbox{arr}(p)=i$. 
%
%
For any frame $f = \{ p_{1}, \ldots, p_{k} \}$ 
such that $\mbox{arr}(p_1) \le \cdots \le \mbox{arr}(p_k)$, 
we call $p_{i}$ the {\em $i$-packet} of $f$. 
Consider two frames $f_{i} = \{ p_{i,1}, \ldots, p_{i,k} \}$ and
$f_{i'} = \{ p_{i',1}, \ldots, p_{i',k} \}$ such that
$\mbox{arr}(p_{i,1}) \le \cdots \le \mbox{arr}(p_{i,k})$ and
$\mbox{arr}(p_{i',1}) \le \cdots \le \mbox{arr}(p_{i',k})$. 
If for any $j$ and $j'$, $\mbox{arr}(p_{i,j}) \le \mbox{arr}(p_{i',j})$ if and only
if $\mbox{arr}(p_{i,j'}) \le \mbox{arr}(p_{i',j'})$, 
then we say that
$f_{i}$ and $f_{i'}$ are {\em order-respecting}. 
If any two frames in an input sequence $\sigma$ are order-respecting, 
we say that $\sigma$ is {\em order-respecting}. 
If all the packets constituting a frame $f$ are transmitted, 
we say that $f$ is {\em completed}, otherwise, 
$f$ is {\em incompleted}.  The goal of $k$-FTM is to maximize the number of
completed frames.  $k$-OFTM is $k$-FTM where inputs are restricted to
order-respecting sequences.

For an input $\sigma$, the {\em gain} of an algorithm $ALG$ is the
number of frames completed by $ALG$ and is denoted by $V_{ALG}(\sigma)$.
If $ALG$ is a randomized algorithm, the gain of $ALG$ is defined as an
expectation ${\mathbb E}[V_{ALG}(\sigma)]$, where the expectation is
taken over the randomness inside $ALG$. 
If $V_{ALG}(\sigma) \geq V_{OPT}(\sigma)/c$ 
(${\mathbb E}[V_{ALG}(\sigma)] \geq V_{OPT}(\sigma)/c$) for an arbitrary input $\sigma$, 
we say that $ALG$ is $c$-{\em competitive}, where $OPT$ is an optimal offline algorithm for $\sigma$. 
Without loss of generality, we can assume that $OPT$
never preempts packets and never accepts a packet of an
incompleted frame.
\fi
%

\section{Upper Bound} \label{UBk}
\ifnum \count10 > 0
（■未修正）
\fi
\ifnum \count11 > 0
\com{（■英語）}
In this section, we present our algorithm {\sc Middle-Drop and Flush}
($MF$) and analyze its competitive ratio.
\fi
%

\subsection{Algorithm} \label{algorithm_k}
\ifnum \count10 > 0
（■未修正）
本節では、我々が考案したアルゴリズムMiddle-Drop and Flush Algorithm($MF$)の定義を与える。
$MF$は1-packetに対するGreedy Algorithm($GR_{1}$)の動作をシミュレートして、
サブルーチンとして利用している。
まず、
$GR_{1}$の定義を与える。
\begin{tabularx}{155mm}{|X|}
	\hline
		{\bf Greedy Algorithm for 1-packets:}
			1-packet$p$が到着した場合、バッファに空きがあるならば、$p$を受理する。
			そうでないならば、$p$をrejectする。
			$j$-packet \hspace{1mm} $(j = 2, \ldots, k)$ が到着した場合、それをrejectする。
\\
	\hline 
\end{tabularx}
$MF$は、
$GR_{1}$が$3B$個の1-packetを受理するごとに、
$MF$は$A$個の1-packetを受理する。
（補題~\ref{LMA:ap.1}参照。）
また、任意の$j (\in [1, k])$に対して、
$MF$はバッファを各$j$-packetのために分割して、
各$j$-packetを少なくとも$A$個ずつ同時に保持できる様に設計されている。
次に、$MF$の定義のためにいくつか定義を与える。
任意のアルゴリズム$ALG (\in \{MF, OPT\})$、
任意のnon-event time $d$、$ALG$の任意の$j$-packet $p$、
$p$の属するフレーム$f$に対して、
$f$の任意のpacketが$d$において$ALG$にdropされていない場合、
$ALG$の$p$は$d$に{\em valid}であると言う。
また、
$ALG$の$f$もまた$d$に{\em valid}であると言う。
$p$の属するフレームは$t$の時点で、completeになり得るフレームである。
（■subbuffer）
更に、
任意の$j$-packet $p'$に対して、
$p'$は時刻$d$に$MF$のバッファ内に存在し、
$MF$のバッファ内の$[1, \ell_{MF}(d, p')]$の位置に$y$個の$j$-packetが存在すると仮定する。
このとき、
${\ell}_{j}(d, p') = y$
と定義する。
$MF$は内部変数として、
{\tt Counter}と{\tt Block}を用いる。
{\tt Counter}は、
$GR_{1}$が受理するpacketの数を数える為の変数である。
{\tt Block}は、各frameを1-packetの到着順に従って分類するための変数である。
任意のnon-event time $d$に対して、
$b(d)$は$d$における$MF$の内部変数{\tt Block}の値を表す。
任意の1-packet $p_{1}$と$p_{1}$のarrival event time $e$に対して、
$g(p_{1}) = b(e-)$と定義する。
また、
$p_{1}$の任意の$j$-packet $p_{j}$に対して、
$g(p_{j}) = b(e-)$と定義する。
更に、
任意のpacket $p$と$p$を含むframe $f$に対して、
$g(f) = g(p)$と定義する。
また、$g(p'')$を$f$(もしくは$p''$)の{\em 通し番号}と呼ぶ。
任意のnon-event time $d$、
任意のアルゴリズム$ALG (\in \{MF, OPT \})$、
任意の通し番号$u$に対して、
$h_{ALG, u}(d)$は、
$d$における$g(f) = u$が成立する$ALG$のvalidなフレーム$f$の数を表す。
それでは、
$MF$の定義を与える。
以下の手続きは、
arrival subphaseにおいて、
各到着したパケット毎に行われる。
ただし、
同時刻に到着した、
任意の$i$-packetと任意の$j(>i)$-packetに対して、
$i$-packetに対する手続きを先に行う。
また、
同時刻に到着した、
任意の$i(\geq 2)$-packet $p, p'(\ne p)$に対して、
通し番号の昇順に手続きを行う。
$MF$を理解するために、
Appendix~\ref{sec:ap.2}に、
$MF$の実行例を与える。
\fi
\ifnum \count11 > 0
\com{（■英語）}
We first give notation needed to describe $MF$. 
Suppose that $n$ packets $p_{1}, p_{2}, \ldots, p_{n}$ arrive at $MF$'s buffer at the $i$th arrival subphase.
For each packet, 
$MF$ decides whether to accept it or not one by one (in some order defined later). 
Let $t_{p_{j}}$ denote the time when $MF$ deals with the packet $p_{j}$, and 
let us call $t_{p_{j}}$ the {\em decision time} of $p_{j}$. 
Hence if $p_{1}, p_{2}, \ldots, p_{n}$ are processed in this order, 
we have that $t_{p_{1}} < t_{p_{2}} < \cdots < t_{p_{n}}$. 
(For convenience, in the later analysis, 
we assume that $OPT$ also deals with $p_{j}$ at the same time $t_{p_{j}}$.) 
\com{（■アルゴリズムの通し番号のところとか、L3.6で使ってる）}
Also, 
let us call the time when $MF$ transmits a packet from the head of its buffer at the $i$th delivery subphase the {\em delivery time} of the $i$th delivery subphase. 
A decision time or a delivery time is called an {\em event time}, 
and any other moment is called a {\em non-event time}. 
Note that during the non-event time, 
the configuration of the buffer is unchanged. 
For any event time $t$,
$t+$ denotes any non-event time between $t$ and the next event time.
Similarly, $t-$ denotes any non-event time between $t$ and the previous
event time.

Let $ALG$ be either $MF$ or $OPT$. 
For a non-event time $t$ and a packet $p$ of a frame $f$, 
we say that $p$ is {\em valid} for $ALG$ at $t$ 
if $ALG$ has not dropped any packet of $f$ before $t$, 
i.e., $f$ still has a chance of being completed. 
In this case 
we also say that the frame $f$ is {\em valid} for $ALG$ at $t$. 
Note that a completed frame is valid at the end of the input. 
For a $j$-packet $p$ and a non-event time $t$,
if $p$ is stored in $MF$'s buffer at $t$,
we define ${\ell}(t, p)$ as ``1+(the number of $j$-packets located in front of $p$)'',
that is, $p$ is the ${\ell}(t, p)$th $j$-packet in $MF$'s queue.
If $p$ has not yet arrived at $t$, we define ${\ell}(t, p) = \infty$. 
During the execution, 
$MF$ virtually runs the following greedy algorithm $GR_{1}$ on the same input sequence. 
Roughly speaking, 
$GR_{1}$ is greedy for only 1-packets and ignores all $j (\geq 2)$-packets. 
Formally, 
$GR_{1}$ uses a FIFO queue of the same size $B$. 
At an arrival of a packet $p$, $GR_{1}$ rejects it 
if it is a $j$-packet for $j \geq 2$. 
If $p$ is a 1-packet, $GR_{1}$ accepts it whenever there is a space in the queue. 
At a delivery subphase,
$GR_{1}$ transmits the first packet of the queue as usual.
$MF$ uses two internal variables {\tt Counter} and {\tt Block}. 
{\tt Counter} is used to count the number of packets accepted by $GR_{1}$ modulo $3B$. 
{\tt Block} takes a positive integer value; it is initially one and is increased by one each time {\tt Counter} is reset to zero. 
Define $A = \lfloor B/k \rfloor$. 
$MF$ stores at most $A$ $j$-packets for any $j$. 
For $j=1$, 
$MF$ refers to the behavior of $GR_{1}$ in the following way:
Using two variables {\tt Counter} and {\tt Block}, 
$MF$ divides 1-packets accepted by $GR_{1}$ into blocks according to their arrival order,
each with $3B$ 1-packets.
$MF$ accepts the first $A$ packets of each block and rejects the rest.
For $j \geq 2$, $MF$ ignores $j$-packets that are not valid.
When processing a valid $j$-packet $p$, 
if $MF$ already has $A$ $j$-packets in its queue, 
then $MF$ preempts the one in the ``middle'' among those $j$-packets and accepts $p$. 
For a non-event time $t$, 
let $b(t)$ denote the value of {\tt Block} at $t$. 
For a packet $p$, we define the {\em block number} $g(p)$ of $p$ as follows. 
For a 1-packet $p$, 
$g(p) = b(t-)$ where $t$ is the decision time of $p$, 
and for some $j (\geq 2)$ and a $j$-packet $p$ , 
$g(p) = g(p')$ where $p'$ is the 1-packet corresponding to $p$. 
Hence, all the packets of the same frame have the same block number. 
We also define the block number of frames in a natural way, namely, 
the block number $g(f)$ of a frame $f$ is the (unique) block number of the packets constituting $f$.  For a non-event
time $t$ and a positive integer $u$, let $h_{ALG, u}(t)$ denote the
number of frames $f$ valid for $ALG$ at $t$ such that $g(f)=u$.

Recall that at an arrival subphase, 
more than one packet may arrive at a queue. 
$MF$ processes the packets ordered non-increasingly first by their frame indices and then by block numbers. 
If both are equal, they are processed in arbitrary order. 
That is, 
$MF$ processes these packets by the following rule: 
Consider an $i$-packet $p$ and an $i'$-packet $p'$. 
If $i<i'$, $p$ is processed before $p'$ and if $i'<i$, $p'$ is processed before $p$. 
If $i=i'$, then $p$ is processed before $p'$ if $g(p) < g(p')$ and $p'$ is processed before $p$ if $g(p') < g(p)$. 
If $i=i'$ and $g(p) = g(p')$, 
the processing order is arbitrary.
The formal description of $MF$ is as follows.  
%
%
To illustrate an execution of $MF$, 
we give an example in Appendix~\ref{sec:ap.2}. 
\fi
\ifnum \count10 > 0
\noindent\vspace{-1mm}\rule{\textwidth}{0.5mm} 
\vspace{-3mm}
{\bf Middle-Drop and Flush Algorithm}\\
\rule{\textwidth}{0.1mm}
{\bf Initialize:} {\tt Counter} $:= 0$, {\tt Block} $:= 1$. \\
整数時間$t$に$j$-packet $p$がarriveしたとする。\\
{\bf Case 1: $j = 1$、すなわち、$p$が1-packetの場合: }\\
\hspace{3mm}
	{\bf Case 1.1: $GR_{1}$が$p$をrejectする場合:}\\
	\hspace{6mm}
		$MF$は$p$をrejectする。\\
\hspace{3mm}
	{\bf Case 1.2: $GR_{1}$が$p$をacceptする場合:}\\
	\hspace{6mm}
		{\tt Counter} $:=$ {\tt Counter} $+ 1$. \\
	\hspace{6mm}
		{\bf Case 1.2.1: {\tt Counter} $\leq A$の場合:}\\
		\hspace{9mm}
			$MF$は$p$をacceptする。（$MF$は必ず$p$をaccept出来る。Lemma~\ref{LMA:ap.1}参照。）\\
	\hspace{6mm}
		{\bf Case 1.2.2: $A < $ {\tt Counter} $< 3B$の場合:}\\
		\hspace{9mm}
			$MF$は$p$をrejectする。\\
	\hspace{6mm}
		{\bf Case 1.2.3: {\tt Counter} $= 3B$の場合:}\\
		\hspace{9mm}
			$MF$は$p$をrejectする。{\tt Counter} $:= 0$とし、{\tt Block} $:= {\tt Block} + 1$とする。\\
{\bf Case 2: $j \geq 2$、すなわち、$p$が1-packetでない場合:} \\
\hspace{3mm}
	{\bf Case 2.0${}^\dagger$: \com{（■g削除）}
		$g(p) < {\tt Block}$が成立し、
		かつ時刻$t_{p}-$に$g(p) = g(f)$が成立するvalidなフレーム$f$の数が高々$\lfloor A/2 \rfloor - 1$である場合
	:} \\
	\hspace{6mm}
		$MF$は$p$をrejectする。\\
\hspace{3mm}
	{\bf Case 2.1: 時刻$t_{p}-$に$p$がvalidでない場合:} \\
	\hspace{6mm}
		$MF$は$p$をrejectする。\\
\hspace{3mm}
	{\bf Case 2.2: 時刻$t_{p}-$に$p$がvalidである場合:} \\
	\hspace{6mm}
		{\bf Case 2.2.1: 時刻$t_{p}-$にバッファ内の$j$-packetの数が高々$A-1$である場合:} \\
		\hspace{9mm}
			$MF$は$p$をacceptする。\\
	\hspace{6mm}
		{\bf Case 2.2.2: そうでない場合、すなわち、
					時刻$t_{p}-$にバッファ内の$j$-packetの数が少なくとも$A$である場合:} \\
			\hspace{9mm}
				$t_{p}-$の$MF$のバッファ内の先頭から$\lfloor A/2 \rfloor + 1$番目の$j$-packet、
				すなわち、$\ell(t_{p}-, p') = \lfloor A/2 \rfloor + 1$が成立する$p'$、をpreemptし、
				$p$をacceptする。
				（$p'$を$t_{p}$における{\em causing packet}と呼ぶ。）
				また、$t_{p}-$におけるバッファ内に$p'$と同じフレームのpacketが存在するなら、
				それらを全てpreemptする。
				\\
			\hspace{9mm}
				{\bf Case 2.2.2.1: $h_{MF, g(p')}(t_{p}-) \leq \lfloor A/2 \rfloor$である場合:} \\
				\hspace{12mm}\com{（■g削除）}
					バッファ内の$g(p'') = g(p')$が成立するpacket$p''$を全てpreemptする。\\
			\hspace{9mm}
				{\bf Case 2.2.2.2: $h_{MF, g(p')}(t_{p}-) \geq \lfloor A/2 \rfloor + 1$である場合:} \\
				\hspace{12mm}
					何もしない。\\
\rule{\textwidth}{0.1mm}
\fi
\ifnum \count11 > 0
\com{\noindent（■英語）\\}
\noindent\vspace{-1mm}\rule{\textwidth}{0.5mm} 
\vspace{-3mm}
{\bf Middle-Drop and Flush}\\
\rule{\textwidth}{0.1mm}
{\bf Initialize:} {\tt Counter} $:= 0$, {\tt Block} $:= 1$. \\
Let $p$ be a $j$-packet to be processed. \\
{\bf Case 1: $\boldsymbol{j = 1}$:} \\
\hspace*{3mm}
	{\bf Case 1.1:} If $GR_{1}$ rejects $p$, reject $p$. \\
\hspace*{3mm}
	{\bf Case 1.2:} If $GR_{1}$ accepts $p$, 
					set {\tt Counter} $:=$ {\tt Counter} $+ 1$ and do the following. \\
	\hspace*{6mm}
		{\bf Case 1.2.1:} If ${\tt Counter} \leq A$, accept $p$. 
		(We prove in Lemma~\ref{LMA:ap.1}(c) that $MF$'s buffer has a\\
	\hspace*{12mm}  space whenever ${\tt Counter} \leq A$.)\\ 
	\hspace*{6mm}
		{\bf Case 1.2.2:} If $A < {\tt Counter} < 3B$, reject $p$. \\
	\hspace*{6mm}
		{\bf Case 1.2.3:} If ${\tt Counter} = 3B$, 
							reject $p$ and set ${\tt Counter} := 0$ and ${\tt Block} := {\tt Block} + 1$. \\ 
{\bf Case 2: $\boldsymbol{j \geq 2}$:} \\
\hspace*{3mm}
	{\bf Case 2.1:} If $p$ is not valid for $MF$ at $t_{p}-$, reject $p$. \\
\hspace*{3mm}
	{\bf Case 2.2:} If $p$ is valid for $MF$ at $t_{p}-$, do the following.\\
	\hspace*{6mm}
		{\bf Case 2.2.1:} If the number of $j$-packets in $MF$'s buffer at $t_{p}-$ is at most $A-1$, accept $p$.\\
	\hspace*{6mm}
		{\bf Case 2.2.2:} If the number of $j$-packets in $MF$'s buffer at $t_{p}-$ is at least $A$, then preempt the\\
			\hspace*{12mm}  $j$-packet $p'$ such that ${\ell}(t_{p}-, p') =
							\lfloor A/2 \rfloor + 1$, and accept $p$.  Preempt all the packets \\
			\hspace*{12mm}  corresponding to $p'$ (if any).\\
		\hspace*{9mm}
			{\bf Case 2.2.2.1:} If $h_{MF, g(p')}(t_{p}-) \leq \lfloor A/2 \rfloor$,
								preempt all the packets $p''$ in $MF$'s buffer such\\
			\hspace*{15mm}  that $g(p'')=g(p')$.  (Call this operation ``flush''.)\\
		\hspace*{9mm}
			{\bf Case 2.2.2.2:} If $h_{MF, g(p')}(t_{p}-) \geq \lfloor A/2 \rfloor + 1$, 
								do nothing. \\
\rule{\textwidth}{0.1mm}
\fi
%

\subsection{Overview of the Analysis} \label{overview_k}
\ifnum \count10 > 0
（■）
いくつか定義を与える。
任意のnon-event time $d$に対して、
$c(d)$は$d$における$MF$の内部変数{\tt Counter}の値を表す。
入力が終了した後の任意の時刻を$\tau$とする。
$c(\tau) = 0$ならば、$M = b(\tau) - 1$と定義し、
$c(\tau) > 0$ならば、$M = b(\tau)$と定義する。
\com{（■g削除）}
入力が終了した時点において、
$g(f) = u$が成立し、
$MF$がcompleteしたframe $f$の数が少なくとも$\lfloor A/2 \rfloor$であれば、
$u$は{\em 良い}通し番号であるという。
\com{（■g削除）}
ここで、
${G} = \{ u \mid $u$ \mbox{ は良い通し番号} \} \cup \{ M \}$
と定義する。
${G}$に含まれる通し番号の数を$m$個とする。
$a_{j} \hspace{1mm} (j \in [1, m])$
は、通し番号を表す。
ただし、$a_{1} < a_{2} < \ldots < a_{m} = M$が成立する。
また、
$M \geq 2$の場合、
$MF$は通し番号が1のframeを少なくとも$\lfloor A/2 \rfloor$はcompleteするので、
$a_1 = 1$が成立する。
このとき、
$V_{OPT}(\sigma) = \sum_{i = 1}^{M} h_{OPT, i}(\tau)$
かつ
$V_{MF}(\sigma) = \sum_{i = 1}^{M} h_{MF, i}(\tau)
		\geq \sum_{i = 1}^{m} h_{MF, a_{i}}(\tau)$
が成立する。
$a_{i}$の定義より、
任意の$i \in [1, m-1]$に対して、
$h_{MF, a_{i}}(\tau) \geq \lfloor A/2 \rfloor$
が成立する。
補題~\ref{LMA:k.1}より、
$c(\tau) = 0$
もしくは
$c(\tau) \in [\lfloor A/2 \rfloor, 3B]$ならば、
$h_{MF, M}(\tau) \geq \lfloor A/2 \rfloor$
が成立する。
$c(\tau) \in [1, \lfloor A/2 \rfloor-1]$ならば、
$h_{MF, M}(\tau) = c(\tau)$が成立し、
$M \geq 2$ならば、
$h_{MF, M}(\tau) + B - 1 \geq h_{OPT, M}(\tau)$
が成立し、
$M = 1$ならば、
$h_{MF, M}(\tau) \geq h_{OPT, M}(\tau)$
が成立する。
また、
補題~\ref{LMA:k.2}より、
$\sum_{j = a_{1}}^{a_{2} - 1} h_{OPT, j}(\tau) \leq 4B + A - 3$
と
任意の$i \in [2, m-1]$に対して、
$\sum_{j = a_{i}}^{a_{i+1} - 1} h_{OPT, j}(\tau) \leq 5B + A - 4$
が成立する。
更に、
$h_{OPT, M}(\tau) \leq 4B-1$
が成立する。
以上の式より、
$M = 1$ならば、
$\frac{V_{OPT}(\sigma)}{V_{MF}(\sigma)} 
	\leq \max \{ \frac{ c(\tau) }{ c(\tau) }, \frac{  4B-1 }{ \lfloor A/2 \rfloor } \}
	= \frac{ 4B+A-3 }{ \lfloor A/2 \rfloor }
$
が成立する。
また、
$M \geq 2$ならば、
$V_{OPT}(\sigma) = \sum_{u = 1}^{M} h_{OPT, u}(\tau) 
		 = \sum_{i = 1}^{m-1} \sum_{j = a_{i}}^{a_{i+1} - 1} h_{OPT, j}(\tau) + h_{OPT, a_{m}}(\tau)
		 \leq (m-1)(5B + A - 4) - B + 1 + h_{OPT, M}(\tau)
		 \leq (m-1)(5B + A - 4) + h_{MF, M}(\tau)
		 $, 
かつ
$V_{MF}(\sigma) = \sum_{i = 1}^{m} h_{MF, a_{i}}(\tau)
		\geq (m-1) \lfloor A/2 \rfloor + h_{MF, M}(\tau)$
が成立する。
よって、
$\frac{V_{OPT}(\sigma)}{V_{MF}(\sigma)} 
	\leq \frac{ (m - 1)(5B + A - 4) + h_{MF, M}(\tau) }{ (m-1) \lfloor A/2 \rfloor + h_{MF, M}(\tau) }
	\leq \frac{ (m - 1)(5B + A - 4) + 1 }{ (m-1) \lfloor A/2 \rfloor + 1}
$
が成立する。
以上より、
次の定理が得られる。
\fi
\ifnum \count11 > 0
\com{（■英語）}
Let $\tau$ be any fixed time after the end of the input, and 
let $c$ denote the value of {\tt Counter} at $\tau$. 
Also, we define $M = b(\tau) - 1$ if $c = 0$ and $M = b(\tau)$ otherwise. 
Then, 
note that for any frame $f$, 
$1\leq g(f) \leq M$. 
Define the set $G$ of integers as 
$G = \{ M \} \cup \{ i \mid $ there are at least $\lfloor A/2 \rfloor$ frames $f$ completed by $MF$ such that $g(f)=i \}$ 
and let $m = |G|$. 
For each $j \in [1, m]$, 
let $a_{j}$ be the $j$th smallest integer in $G$. 
All the block numbers in $G$ are called {\em good}, 
which almost corresponds to good blocks in Sec.~\ref{sec:idea}. 
The other block numbers are called {\em bad}. 
Note that $a_{j}$ denotes the $j$th good block number. 
Note that $a_{m} = M$. 
Also, we prove $a_{1} = 1$ in Lemma~\ref{LMA:k.035}. 
Now since at the end of the input any valid frame is completed, 
we have $V_{OPT}(\sigma) = \sum_{i = 1}^{M} h_{OPT, i}(\tau)$ 
and 
$V_{MF}(\sigma) = \sum_{i = 1}^{M} h_{MF, i}(\tau) \geq \sum_{i = 1}^{m} h_{MF, a_{i}}(\tau)$. 

By the definition of $G$, 

\hspace{6mm}
$h_{MF, a_{i}}(\tau) \geq \lfloor A/2 \rfloor$ for any $i \in [1, m-1]$.\hspace{3mm} (1)

\noindent 
Lemma~\ref{LMA:k.1} focuses on the $m$th good block number, i.e., $M$. 
Since it has some exceptional properties, 
we discuss the number of completed frames with block number $M$ independently of the other good block numbers as follows: 

\hspace{6mm} 
If $c \in \{0\} \cup [\lfloor A/2 \rfloor, 3B-1]$,
then $h_{MF, M}(\tau) \geq \lfloor A/2 \rfloor$.\hspace{3mm} (2) 

\hspace{6mm} 
If $c \in [1, \lfloor A/2 \rfloor-1]$ and $M \geq 2$, 
then $h_{MF, M}(\tau) + B - 1 \geq h_{OPT, M}(\tau)$.\hspace{3mm} (3) 

\hspace{6mm} 
If $c \in [1, \lfloor A/2 \rfloor-1]$ and $M = 1$, 
then $h_{MF, M}(\tau) \geq h_{OPT, M}(\tau)$.\hspace{3mm} (4)

\noindent 
Also, 
in Lemma~\ref{LMA:k.2}, 
we evaluate the number of $OPT$'s completed frames from a viewpoint of good block numbers: 

\hspace{6mm} 
$h_{OPT, M}(\tau) \leq 4B-1$,\hspace{3mm} (5) 

\hspace{6mm} 
$\sum_{j = a_{1}}^{a_{2} - 1} h_{OPT, j}(\tau) \leq 4B + A - 3$,\hspace{3mm} (6) 

and 

\hspace{6mm} 
$\sum_{j = a_{i}}^{a_{i+1} - 1} h_{OPT, j}(\tau) \leq 5B + A - 4$ 
for any $i \in [2, m-1]$.\hspace{3mm} (7) 
Using the above inequalities, 
we can get the competitive ratio of $MF$ according to the values of $M$ and $c$. 
First, 
note that if $M = 1$ then $c \geq 1$ because at least one packet arrives. 
Thus $V_{OPT}(\sigma) > 0$.
Now if $M=1$ and $c \in [1, \lfloor A/2 \rfloor-1]$, 
then $\frac{V_{OPT}(\sigma)}{V_{MF}(\sigma)} =
\frac{h_{OPT, 1}(\tau)}{h_{MF, 1}(\tau)} \leq 1$ by (4). 
If $M=1$ and $c \in [\lfloor A/2 \rfloor, 3B-1]$, 
then $\frac{V_{OPT}(\sigma)}{V_{MF}(\sigma)} = \frac{h_{OPT, 1}(\tau)}{h_{MF,
1}(\tau)} \leq \frac{4B-1}{\lfloor A/2 \rfloor} < \frac{5B+A-4}{\lfloor A/2 \rfloor}$ by (2) and (5).
If $M \ge 2$ and $c \in \{ 0 \} \cup [\lfloor A/2 \rfloor, 3B-1]$,
\begin{eqnarray*}
	V_{OPT}(\sigma) 
		& = & \sum_{i = 1}^{M} h_{OPT, i}(\tau) 
		  =  \sum_{i = 1}^{m-1} \sum_{j = a_{i}}^{a_{i+1} - 1} h_{OPT, j}(\tau) + h_{OPT, a_{m}}(\tau) \\ 
		& \leq & (m-1)(5B + A - 4) - B + 1 + (4B - 1) 
		  <  m(5B + A - 4)
\end{eqnarray*}
by (5), (6), and (7) (note that $a_{1} = 1$ and $a_{m} = M$). 
Also,
$V_{MF}(\sigma) \geq \sum_{i = 1}^{m} h_{MF, a_{i}}(\tau) \geq m\lfloor A/2 \rfloor$ by (1) and (2). 
Therefore,
$\frac{V_{OPT}(\sigma)}{V_{MF}(\sigma)} < \frac{5B + A - 4}{ \lfloor A/2 \rfloor }$. 
Finally, 
if $M \ge 2$ and $c \in [1, \lfloor A/2 \rfloor-1]$,
\begin{eqnarray*}
	V_{OPT}(\sigma) 
		& = & \sum_{i=1}^{M} h_{OPT, i}(\tau) 
		 =  \sum_{i=1}^{m-1} \sum_{j = a_{i}}^{a_{i+1}-1} h_{OPT, j}(\tau) + h_{OPT, a_{m}}(\tau)  \\
		& \leq & (m-1)(5B+A-4)-B+1+h_{OPT, M}(\tau) 
		 \leq  (m-1)(5B+A-4) + h_{MF, M}(\tau)
\end{eqnarray*}
by (3), (6) and (7). 
Also, 
$V_{MF}(\sigma) = \sum_{i = 1}^{m} h_{MF, a_{i}}(\tau) \geq (m-1) \lfloor A/2 \rfloor + h_{MF, M}(\tau)$ by (1).
Therefore,
\begin{equation*}
	\frac{V_{OPT}(\sigma)}{V_{MF}(\sigma)} 
		 \leq  \frac{ (m-1)(5B+A-4) + h_{MF, M}(\tau) }{ (m-1) \lfloor A/2 \rfloor + h_{MF, M}(\tau) } 
		 <  \frac{5B+A-4}{\lfloor A/2 \rfloor}.
\end{equation*}
We have proved that in all the cases
$\frac{V_{OPT}(\sigma)}{V_{MF}(\sigma)} < \frac{5B+A-4}{\lfloor A/2 \rfloor}$. 
By noting that $\frac{5B+A-4}{\lfloor A/2 \rfloor} =
\frac{5B + \lfloor B/k \rfloor - 4}{\lfloor B/2k \rfloor}$, we have the
following theorem:
\fi
\begin{THM}\label{thm:1}
	\ifnum \count10 > 0
	$B / k \geq 2$のとき、
	$MF$の競合比は高々$\frac{5B + \lfloor B/k \rfloor - 4}{ \lfloor B/2k \rfloor }$である。
	\fi
	\ifnum \count11 > 0
	\com{（■英語）}
	When $B / k \geq 2$, 
	the competitive ratio of $MF$ is at most $\frac{5B + \lfloor B/k \rfloor - 4}{ \lfloor B/2k \rfloor }$. 
	\fi
\end{THM}

The rest of Sec.~\ref{UBk} is devoted to the proofs of Lemmas~\ref{LMA:k.035}, \ref{LMA:k.1}, and \ref{LMA:k.2}.

\subsection{Analysis of $MF$} \label{sec:analysis}
\ifnum \count10 > 0
まず、
$MF$の実行可能性を示そう。
Specifically, we show the executability of Case 1.2.1. 
それと同時に
以下の補題では、$GR_{1}$がacceptする1-packetと
$MF$のそれとの関係を明らかにする。
定義を与える。
$ALG$がpacket $p$を$i$番目のdelivery subphaseにtransmitするとき、
$\mbox{del}_{ALG}(p) = i$と表記する。 
\fi
\ifnum \count11 > 0
\com{（■英語）}
First, 
we guarantee the feasibility of $MF$ in the following lemma. 
Specifically, we show the executability of Case 1.2.1. 
At the same time, 
we also show the connection between 1-packets  accepted by $GR_{1}$ and those by $MF$. 
To prove the following lemma, 
we give the notation. 
If an algorithm $ALG$ transmits a packet $p$ at the $i$th delivery subphase, 
we write $\mbox{del}_{ALG}(p) = i$. 
\fi
%
\ifnum \count14 > 0
%
The proofs of some of the following lemmas are included in Appendix~\ref{sec:ap.3}.
\fi
%

%
\begin{LMA}\label{LMA:ap.1}
	\ifnum \count10 > 0
	任意の正の整数$z$に対して、
	$GR_{1}$が$z$個の1-packetを受理するとし、
	$GR_{1}$が$i$番目にacceptするpacketを$p_{i} \hspace{1mm} (i \in [1, z])$で表す。
	このとき、以下が成立する。
	(a) 
	$MF$が$p_{j}$をacceptする様な
	任意の$j (\in [1, z - 2B + 1])$に対して、
	$\mbox{del}_{MF}(p_{j}) < \mbox{arr}(p_{j + 2B - 1})$
	が成立する。
	(b) 
	任意の$u (\geq 0)$に対して、
	$MF$はパケット集合$\{ p_{3Bu + 1}, \ldots, p_{3Bu + A} \}$
	\com{（■$\{ p_{\gamma u + 1}, \ldots, p_{\gamma u + A} \}$）}
	のパケットを受理することが出来る。
	また、
	それらの通し番号は、$u+1$である。
	(c) 
	1-packetの受理を判断する直前に、
	$0 \leq {\tt Counter} \leq A-1$が成立している場合、
	$MF$はその1-packetを受理するための、バッファの空きが存在する。
	\fi
	\ifnum \count11 > 0
	\com{（■英語）}
	Suppose that $GR_{1}$ accepts $z (\geq 2B)$ 1-packets. 
	Let $p_{i} \hspace{1mm} (i \in [1, z])$ denote the $i$th 1-packet accepted by $GR_{1}$. 
	Then, the following holds. 
	(a) 
	For any $j (\in [1, z - 2B + 1])$ such that $MF$ accepts $p_{j}$, 
	$\mbox{del}_{MF}(p_{j}) < \mbox{arr}(p_{j + 2B - 1})$. 
	(b) 
	For any $u (\geq 0)$, $MF$ accepts all the packets in $\{ p_{3Bu + 1}, \ldots, p_{3Bu + A} \}$, 
	and their block number is $u+1$. 
	(c) 
	If $0 \leq {\tt Counter} \leq A-1$ just before the decision time of a 1-packet $p$, 
	$MF$'s buffer has a space to accept $p$. 
	\fi
\end{LMA}
%
%
\begin{proof}
	\ifnum \count10 > 0
	(a) 
	補題の条件をみたす1-packet $p_{j}$と
	non-event time $t_{p_{j}}+$を考えよ。
	$t_{p_{j}}+$は、
	$GR_{1}$と$MF$が$p_{j}$をacceptする瞬間である
	ことに注意。
	$MF$はその定義より、
	acceptした1-packetを必ずtransmitする。
	また、
	バッファのサイズは$B$であり、
	各サブフェイズに転送可能なpacketは1つだけであるので、
	$MF$は$B$フェイズ内に$p_{j}$を転送する。
	すなわち、
	$t_{p_{j}}+$から$p_{j}$をtransmitする前の間の、
	delivery subphaseの数は、高々$B-1$である。
	結果として、
	$GR_{1}$は、この期間に、
	高々$B-1$個のpacketsを転送できる。
	一方、
	$GR_{1}$も$p_{j}$をacceptするので、
	$t_{p_{j}}+$における$GR_{1}$のバッファ内のpacket数は少なくとも1個である。
	すなわち、
	$t_{p_{j}}+$における$GR_{1}$のバッファ内の空きは、
	高々$B-1$である。
	以上より、
	$t_{p_{j}+}$から$MF$が$p_{j}$を転送するnon-event timeの間に、
	$GR_{1}$がaccept可能なpacketの数は、
	高々$(B-1)+(B-1)=2B-2$である。
	つまり、
	$GR_{1}$が$p_{j + 2B - 1}$をバッファに入れる前に、
	$MF$は$p_{j}$をtransmitする。
	よって、
	$\mbox{del}_{MF}(p_{j}) < \mbox{arr}(p_{j + 2B - 1})$
	が成立する。
	(b) 
	まず、
	$MF$は各$x (\in [1, k])$-packetを$A$個までaccept可能であることに注意せよ。
	Cases 1.2.1, 1.2.2, 1.2.3より、
	$MF$はpacket $p_{1}, p_{2}, \cdots, p_{A}$をacceptし、
	$p_{A+1}, p_{A+2}, \cdots, p_{3B}$をrejectする。
	この間、
	${\tt Block}$は1であり、
	$MF$が$p_{3B}$をrejectした直後、 
	${\tt Block}$は$1$になり、
	${\tt Counter}$は$0$となる. 
	(a)の証明より、
	$\mbox{del}_{MF}(p_{A}) < \mbox{arr}(p_{A + 2B - 1})$
	が成立するので、
	$\mbox{del}_{MF}(p_{A}) < \mbox{arr}(p_{3B+1})$
	が成立する。
	すなわち、
	$t_{p_{3B+1}}-$の時点で$MF$のバッファ内には1-packetは1つも存在しない。
	ゆえに、
	$MF$は、パケット$p_{3B+1}, p_{3B+2}, \cdots, p_{3B+A}$を受理する。
	よって、帰納的に、
	(b)を示すことが出来る。
	(c)
	各$u = 0, 1, \ldots$に対して、
	$p_{3Bu + 1}, \ldots, p_{3Bu + A}$のdecision timeの直前に、
	$0 \leq {\tt Counter} \leq A-1$が成立している。
	よって、
	(b)の証明より、(c)が成立する。
	\fi
	\ifnum \count11 > 0
	\com{（■英語）}
	(a) 
	Consider a 1-packet $p_{j}$ satisfying the condition of the lemma, 
	and consider the non-event time $t_{p_{j}}+$, 
	i.e., the moment just after $GR_{1}$ and $MF$ accept $p_{j}$. 
	By definition, 
	$MF$ certainly transmits any 1-packet inserted into its buffer. 
	In addition, 
	$MF$ will transmit $p_{j}$ within $B$ phases, 
	since the buffer size is $B$, and 
	only one packet can be transmitted in one phase. 
	That is, 
	the number of delivery subphases between $t_{p_{j}}+$ and the moment before $MF$ transmits $p_{j}$ 
	is at most $B-1$, 
	which means that 
	$GR_{1}$ can also transmit at most $B-1$ packets during this period. 
	On the other hand, 
	$GR_{1}$ accepts $p_{j}$ as well, and 
	there exists at least one packet in $GR_{1}$'s buffer at $t_{p_{j}}+$. 
 	Hence there are at most $B-1$ vacancy in $GR_{1}$'s buffer at this moment. 
	Therefore, 
	the number of packets $GR_{1}$ can accept 
	between $t_{p_{j}}+$ and the moment before $MF$ transmits $p_{j}$ 
	is at most $(B-1)+(B-1)=2B-2$. 
	In other words,
 	$MF$ transmits $p_{j}$ before $GR_{1}$ accepts $p_{j + 2B - 1}$. 
	This proves $\mbox{del}_{MF}(p_{j}) < \mbox{arr}(p_{j + 2B - 1})$.
	(b) 
	First, recall that $MF$ can always store up to $A$ $x$-packets for $x \in [1, k]$. 
	Due to Cases 1.2.1, 1.2.2 and 1.2.3, 
	$MF$ accepts packets $p_{1}, p_{2}, \cdots, p_{A}$, 
	and rejects $p_{A+1}, p_{A+2}, \cdots, p_{3B}$. 
	During this period, 
	${\tt Block}$ stays $1$. 
	Just after $MF$ rejects $p_{3B}$, 
	${\tt Block}$ is incremented to $2$, 
	and ${\tt Counter}$ is reset to $0$. 
	Since $\mbox{del}_{MF}(p_{A}) < \mbox{arr}(p_{A + 2B - 1})$ by the proof of the part (a) of this lemma, 
	$\mbox{del}_{MF}(p_{A}) < \mbox{arr}(p_{3B+1})$, 
	which means that there exists no 1-packet in $MF$'s queue at the non-event time $t_{p_{3B+1}}-$.
	Hence, 
	$MF$ starts accepting $p_{3B+1}, p_{3B+2}, \cdots, p_{3B+A}$. 
	By continuing this argument, we can prove part (b).
	(c) 
	$0 \leq {\tt Counter} \leq A-1$ holds 
	just before the decision times of $p_{3Bu + 1}, \ldots, p_{3Bu + A}$ 
	for each $u = 0, 1, \ldots$. 
	Thus,
	(c) is immediate from the proof of part (b). 
	\fi
\end{proof}
%
%

%
\ifnum \count10 > 0
さて、
次に$MF$が敢えて定期的に1-packetを受理しようとする
理由を次の2つの補題において示そう。
これらの補題では、
簡単に言えば、
validな任意の2つの$j(\geq 2)$-パケットに対して、
先に到着するpacketの到着番号は、
後から到着するpacketのそれ以下であるというものである。
この補題は単純だが、$MF$の解析において最も核の部分である。
\fi
\ifnum \count11 > 0
\com{（■英語）}
Now we clarify the reason $MF$ dares to accept 1-packets in moderation in the following two lemmas. 
Briefly speaking, 
these lemmas show that for any two valid $j(\geq 2)$-packets, the block number of the earlier one is equal to or less than that of the later one. 
The lemmas are quite simple but the core of our analysis of $MF$. 
\fi
%

%
\begin{LMA}\label{LMA:k.00}
	\ifnum \count10 > 0
	(■)
	任意の$x \in [1, k]$に対して、
	$p$を$MF$がacceptする$1$-packetとし、
	$p'$を$p$の$x$-packetとする。
	$q$を任意の$1$-packetとする。
	$q'$を$q$の$x$-packetとする。
	ただし、
	$g(p) < g(q)$が成立する。
	このとき、
	(a) 
	$\mbox{arr}(p) < \mbox{arr}(q)$
	が成立し、
	(b)
	$\mbox{arr}(p') \leq \mbox{arr}(q')$
	が成立する。
	\fi
	\ifnum \count11 > 0
	\com{（■英語）}
	(a) Let $p$ be any $1$-packet accepted by $MF$ and $q$ be any $1$-packet such that $g(p) < g(q)$.
	Then, $\mbox{arr}(p) < \mbox{arr}(q)$.
	(b) 
	For $x \in [2, k]$, let $p'$ and $q'$ be any $x$-packets such that $g(p') < g(q')$ and 
	suppose that $MF$ accepts the 1-packet corresponding to $p'$.
	Then, $\mbox{arr}(p') \leq \mbox{arr}(q')$.
	\fi
\end{LMA}
%
%
\begin{proof}
	\ifnum \count10 > 0
	(a) 
	各1-packetの通し番号は、単調増加である。
	すなわち、
	$g(p) < g(q)$ならば、
	$\mbox{arr}(p) \leq \mbox{arr}(q)$
	が成立する。
	よって、
	$\mbox{arr}(p) \neq \mbox{arr}(q)$. 
	の場合のみ示せばよい。
	$p$は$GR_{1}$もacceptする。
	また、
	blockが$g(q)$である最初の1-packet $\hat{q}$も、
	$GR_{1}$はacceptし、
	$\mbox{arr}(\hat{q}) \leq \mbox{arr}(q)$
	が成立する。
	そこで、
	$p$と$\hat{q}$が$GR_{1}$によって、
	$i$番目と$j$番目にacceptされるpacketと仮定せよ。
	補題~\ref{LMA:ap.1}(b)と
	仮定より、$g(p) < g(q) (= g(\hat{q}))$, 
	が成立するので、
	$j-i \geq (3B(g(\hat{q})-1)+1)-(3B(g(p)-1)+A) =3B(g(\hat{q})-g(p))+1-A \geq 3B+1-A > 2B$
	が成立する。
	このとき、
	補題~\ref{LMA:ap.1}(a)より、
	$p$は$\hat{q}$の到着時刻より前に、転送されている。
	よって、
	$\mbox{arr}(p) < \mbox{arr}(\hat{q})$が成立する。
	すなわち、
	$\mbox{arr}(p) < \mbox{arr}(q)$
	が成立する。
	(b) 
	(a)より、$\mbox{arr}(p) < \mbox{arr}(q)$が成立する。
	到着するパケットは、order-respectingなので、
	$\mbox{arr}(p') \leq \mbox{arr}(q')$
	が成立する。
	\fi
	\ifnum \count11 > 0
	\com{（■英語）}
	{\bf (a)} 
	Note that $p$ is accepted by also $GR_{1}$. 
	Also, 
	let $\hat{q}$ be the first 1-packet with block number $g(q)$. 
	Clearly $GR_{1}$ accepts $\hat{q}$ and 
	$\mbox{arr}(\hat{q}) \leq \mbox{arr}(q)$. 
	Suppose that $p$ and $\hat{q}$ are the $i$th and the $j$th packets, respectively, accepted by $GR_{1}$. 
	By Lemma~\ref{LMA:ap.1}(b) and the assumption that $g(p) < g(q) (= g(\hat{q}))$, 
	$j-i \geq (3B(g(\hat{q})-1)+1)-(3B(g(p)-1)+A) =3B(g(\hat{q})-g(p))+1-A \geq 3B+1-A > 2B$. 
	Then 
	by Lemma~\ref{LMA:ap.1}(a), 
	$p$ is transmitted by $MF$ before $\hat{q}$ arrives. 
	Therefore, 
	$\mbox{arr}(p) < \mbox{arr}(\hat{q})$, 
	which means that 
	$\mbox{arr}(p) < \mbox{arr}(q)$.
	This completes the proof.
	{\bf (b)} 
	Let $p_1$ and $q_1$ be the 1-packets corresponding to $p'$ and $q'$, respectively.
	Since $g(p') < g(q')$, $g(p_1) < g(q_1)$.
	Therefore, $\mbox{arr}(p_1) < \mbox{arr}(q_1)$ by (a).
	Since the input is order-respecting, 
	$\mbox{arr}(p') \leq \mbox{arr}(q')$. 
	\fi
\end{proof}
%
%

%
\begin{LMA}\label{LMA:k.01}
	\ifnum \count10 > 0
	任意の任意のnon-event time $t$、
	$x \in [1, k]$、
	$MF$の任意の$x$-packet $p, q(\ne p)$に対して、
	$d$に$p$が$MF$のバッファ内に存在しており、
	$d$に$q$が$MF$のバッファ内に存在するか、まだ到着していないとする。
	このとき、
	${\ell}(t, p) < {\ell}(t, q)$ならば、
	$g(p) \leq g(q)$が成立する。
	\fi
	\ifnum \count11 > 0
	\com{（■英語）}
	Suppose that $t$ is a non-event time. 
	For any $x \in [1, k]$, 
	let $p$ be an $x$-packet stored in $MF$'s buffer at $t$, and 
	let $q$ be an $x$-packet which is stored in $MF$'s buffer at $t$ or has not arrived yet at $t$. 
	If ${\ell}(t, p) < {\ell}(t, q)$, 
	then $g(p) \leq g(q)$. 
	\fi
\end{LMA}
%
%
\begin{proof}
	\ifnum \count10 > 0
	${\ell}(t, p) < {\ell}(t, q)$
	が成立するので、
	$MF$は$q$より$p$の方が先に処理している。
	すなわち、
	$\mbox{arr}(p) \leq \mbox{arr}(q)$
	が成立する。
	$x = 1$の場合、
	補題~\ref{LMA:k.00}の対偶より、
	$\mbox{arr}(p) \leq \mbox{arr}(q)$
	が成立する場合、
	$g(p) \leq g(q)$
	が成立する。
	$x \ne 1$かつ$\mbox{arr}(p) < \mbox{arr}(q)$の場合、
	同様に、
	補題~\ref{LMA:k.00}の対偶より、
	$g(p) \leq g(q)$
	が成立する。
	$x \ne 1$かつ$\mbox{arr}(p) = \mbox{arr}(q)$の場合、
	定義より、
	$MF$は小さい通し番号のパケットを先に処理するので、
	$g(p) \leq g(q)$
	が成立する。
	\fi
	\ifnum \count11 > 0
	\com{（■英語）}
	Since ${\ell}(t, p) < {\ell}(t, q)$, 
	$MF$ processes $p$ earlier than $q$, which means that $\mbox{arr}(p) \leq \mbox{arr}(q)$. 
	Thus, in the case of $x = 1$, 
	if $\mbox{arr}(p) \leq \mbox{arr}(q)$, 
	then $g(p) \leq g(q)$ by the contrapositive of Lemma~\ref{LMA:k.00}. 
	In the same way, 
	using the contrapositive of Lemma~\ref{LMA:k.00}, 
	$g(p) \leq g(q)$
	if $x \ne 1$ and $\mbox{arr}(p) < \mbox{arr}(q)$. 
	In the case where both $x \ne 1$ and $\mbox{arr}(p) = \mbox{arr}(q)$, 
	$MF$ processes a packet with a smaller block number earlier by definition, and hence $g(p) \leq g(q)$. 
	\fi
\end{proof}
%
%

%
\ifnum \count10 > 0
次の補題では、
通し番号1がかならずgoodであることを示す。
\ref{sec:idea}節でも軽く触れたが、
この補題によって、
各badの通し番号は、
入力が終了した時点においてある2つのgood通し番号の間の値であるということを示すことが出来る。
（$M$も$G$の定義より、goodである）
\fi
\ifnum \count11 > 0
\com{（■英語）}
In the next lemma, 
we prove that the block number 1 is always good. 
As we described briefly in Sec.~\ref{sec:idea}, 
we can easily show that each bad block number lies between some two good block numbers using the lemma. 
(Recall that $M$ is good by the definition of $G$.)
\fi
%

%
\begin{LMA}\label{LMA:k.035}
	\ifnum \count10 > 0
	$a_1 = 1$が成立する。
	\fi
	\ifnum \count11 > 0
	\com{（■英語）}
	$a_1 = 1$. 
	\fi
\end{LMA}
%
%
\begin{proof}
	\ifnum \count10 > 0
	$M = 1$が成立する場合、
	定義より、$M \in G$が成立するので、
	$a_{1} = 1$が成立する。
	$M \geq 2$の場合を考える。
	$MF$は、
	補題~\ref{LMA:ap.1}(b)より、
	通し番号が$1$である1-packetを$A$個acceptする。
	もし、
	$MF$が一度も通し番号が$1$のpacketをpreemptしないならば、
	明らかに題意がみたされる。
	そこで、
	event time $t$に、
	$MF$が通し番号が$1$である$x (\in [2, k])$-packetをpreemptすると仮定せよ。
	Case 2.2.2の定義より、
	$t-$において$MF$のバッファ内には$A$個の$x$-packetが存在する。
	また、
	補題~\ref{LMA:k.01}より、
	$MF$のバッファ内の$x$-packetの通し番号は単調増加である。
	よって、
	${\ell}(t+, p) \in [1, \lfloor A/2 \rfloor]$
	が成立する
	各$x$-packet $p$に対して、
	$g(p) = 1$が成立している。
	よって、
	$h_{MF, 1}(t+) \geq \lfloor A/2 \rfloor$
	が成立する。
	これは題意をみたす。
	\fi
	\ifnum \count11 > 0
	\com{（■英語）}
	When $M=1$, 
	clearly $a_1=1$ because $M\in G$ by definition. 
	When $M \ge 2$, 
	we show that at any moment there are at least $\lfloor A/2 \rfloor$ frames $f$ with block number $1$ such that $f$ is valid for $MF$. 
	$MF$ accepts at least $\lfloor A/2 \rfloor$ 1-packets with block number $1$ according to Lemma~\ref{LMA:ap.1}(b). 
	If $MF$ does not preempt any packet with block number $1$, 
	the statement is clearly true. 
	Then, 
	suppose that at an event time $t$, 
	$MF$ preempts an $x (\in [2, k])$-packet with block number $1$. 
	By Case 2.2.2 in $MF$, 
	$MF$ stores $A$ $x$-packets in its buffer at $t-$. 
	Moreover, 
	all the $x$-packets in $MF$'s buffer are queued in ascending order of block number by Lemma~\ref{LMA:k.01}. 
	Thus, 
	for each $x$-packet $p$ such that ${\ell}(t+, p) \in [1, \lfloor A/2 \rfloor]$, 
	$g(p) = 1$. 
	As a result, 
	$h_{MF, 1}(t+) \geq \lfloor A/2 \rfloor$, 
	which proves the lemma. 
	\fi
\end{proof}
%
%

%
\ifnum \count10 > 0
以下では、
$OPT$がcompleteするframeの数を評価する。
その為の有力なツールを示そう。
以下の補題では、
ある一定時間の間に$OPT$が受理し得る1-packetの数を、
$GR_{1}$の受理する1-packetの数で評価する。
%
%

%
\fi
\ifnum \count11 > 0
\com{（■英語）}
In the following lemmas, 
we evaluate the number of frames completed by $OPT$. 
Then we show the next lemma that is a useful tool to do so. 
Specifically, 
we bound the number of 1-packets accepted by $OPT$ during a time interval from above by the number of 1-packets accepted by $GR_{1}$ in the lemma. 
\fi
%

%
\begin{LMA}\label{LMA:k.20}
	\ifnum \count10 > 0
	$t_{1}, t_{2}(>t_{1})$を任意のnon-event timeとする。
	時間$[t_{1}, t_{2}]$の間に、
	$GR_{1}$が1-packetを$w (\geq 1)$個受理すると仮定せよ。
	また、
	時間$[t_{1}, t_{2}]$の間に、
	$GR_{1}$が最初にacceptするpacketを$p$とせよ。
	このとき、
	時間$[t_{p}-, t_{2}]$の間に
	$OPT$がacceptする1-packetの数は高々$w+B-1$個である。
	とくに、
	$t_{1}$が入力の始まる前の時刻である場合、
	時間$[t_{p}-, t_{2}]$の間に
	$OPT$がacceptする1-packetの数は高々$w$個である。
	\fi
	\ifnum \count11 > 0
	\com{（■英語）}
	Let $t_{1}, t_{2}(>t_{1})$ be any non-event times. 
	Suppose that $GR_{1}$ accepts $w (\geq 1)$ 1-packets during time $[t_{1}, t_{2}]$. 
	Also, 
	let $p$ be the first 1-packet accepted by $GR_{1}$ during time $[t_{1}, t_{2}]$. 
	Then, 
	the number of 1-packets accepted by $OPT$ during time $[t_{p}-, t_{2}]$ is at most $w+B-1$. 
	In particular, 
	when $t_{1}$ is a time before the beginning of the input, 
	the number of 1-packets accepted by $OPT$ during time $[t_{p}-, t_{2}]$ is at most $w$. 
	\fi
\end{LMA}
%
%
\begin{proof}
	\ifnum \count10 > 0
	$OPT_{1}$を$OPT$のacceptする1-packetのみacceptするオフラインアルゴリズムとする。
	時間$[t_{p}-, t_{2}]$の間に
	$GR_{1}$（$OPT_{1}$）がacceptし、$OPT_{1}$($GR_{1}$)がrejectする1-packetの数を$x$($x'$)とし、
	$GR_{1}$と$OPT_{1}$がacceptする1-packetの数を$x''$とする。
	定義より、
	$x + x'' = w$が成立する。
	以下では、$x' + x''$を上から抑えることを目指す。
	任意のnon-event time $t$、
	任意のアルゴリズム$ALG' (\in \{ OPT_{1}, GR_{1} \})$に対して、
	$t$において$ALG'$がバッファに保持するpacketの数を$f_{ALG'}(t)$で表すとする。
	$GR_{1}$は1-packetをgreedyにacceptし、
	$OPT_{1}$は1-packetのみをacceptするので、
	任意のnon-event time $t$に対して、
	$f_{GR_{1}}(t) - f_{OPT_{1}}(t) \geq 0$
	が成立する。
	また、
	時間$[t_{p}-, t_{2}]$の間に
	$GR_{1}$($OPT_{1}$)がtransmitするpacketの数を$y$($y'$)とする。
	$GR_{1}$と$OPT_{1}$の定義より、
	$y \geq y'$が成立する。
	また、
	$f_{GR_{1}}(t_{2}) = f_{GR_{1}}(t_{p}-) + x + x'' - y$
	と
	$f_{OPT_{1}}(t_{2}) = f_{OPT_{1}}(t_{p}-) + x' + x'' - y'$
	が成立する。
	以上の式より、
	$0 \leq f_{GR_{1}}(t_{2}) - f_{OPT_{1}}(t_{2}) 
		= f_{GR_{1}}(t_{p}-) + x + x'' - y - (f_{OPT_{1}}(t_{p}-) + x' + x'' - y')
		= f_{GR_{1}}(t_{p}-) - f_{OPT_{1}}(t_{p}-) + x - x' - y + y'
		\leq f_{GR_{1}}(t_{p}-) - f_{OPT_{1}}(t_{p}-) + x - x'
	$
	が成立する。
	すなわち、
	$x' \leq f_{GR_{1}}(t_{p}-) - f_{OPT_{1}}(t_{p}-) + x$
	が成立する。
	よって、
	$x' + x'' \leq f_{GR_{1}}(t_{p}-) - f_{OPT_{1}}(t_{p}-) + x + x''
		= f_{GR_{1}}(t_{p}-) - f_{OPT_{1}}(t_{p}-) + w
	$
	が成立する。
	更に、
	$p$をacceptする$GR_{1}$のバッファは$t_{p}-$においてfullではないので、
	$f_{GR_{1}}(t_{p}-) - f_{OPT_{1}}(t_{p}-) \leq B - 1$
	が成立する。
	\com{（■$[t_{1}, t_{2}]$の間で考えると、ここでつまる）}
	よって、
	$x' + x'' \leq B - 1 + w$
	が成立する。
	$t_{1}$が入力が始まる前の時間であるとしよう。
	このとき、
	$p$が一番最初の1-packetであるので、
	$f_{GR_{1}}(t_{p}-) = f_{OPT_{1}}(t_{p}-) = 0$
	が成立する。
	よって、
	$x' + x'' \leq x + x'' = w$
	が成立する。
	\fi
	\ifnum \count11 > 0
	\com{（■英語）}
	Define $OPT_{1}$ as the offline algorithm that accepts only all the 1-packets accepted by $OPT$. 
	Let $x$ (respectively $x'$) 
	be the number of 1-packets accepted by $GR_{1}$ but not accepted by $OPT_{1}$ 
	(respectively accepted by $OPT_{1}$ but not accepted by $GR_{1}$) 
	during time $[t_{p}-, t_{2}]$. 
	Also, 
	let $x''$ be the number of 1-packets accepted by both $GR_{1}$ and $OPT_{1}$ 
	during time $[t_{p}-, t_{2}]$. 
	Since $GR_{1}$ accepts $w$ packets during time $[t_{p}-, t_{2}]$, 
	$x + x'' = w$. 
	In what follows, we bound $x' +x''$ from above.
	For a non-event time $t$ and an algorithm $ALG' (\in \{ OPT_{1}, GR_{1} \})$, 
	let $f_{ALG'}(t)$ denote the number of 1-packets in $ALG'$'s buffer at $t$. 
	Since $GR_{1}$ accepts 1-packets greedily and $OPT_{1}$ accepts only 1-packets, 
	$f_{GR_{1}}(t) - f_{OPT_{1}}(t) \geq 0$ holds for any $t$. 
	Let $y$ (respectively $y'$) 
	denote the number of 1-packets transmitted by $GR_{1}$ (respectively $OPT_{1}$) 
	during time $[t_{p}-, t_{2}]$. 
	Since $f_{GR_{1}}(t) - f_{OPT_{1}}(t) \geq 0$ for any $t$, 
	$GR_{1}$ transmits a 1-packet whenever $OPT_{1}$ does so, 
	and hence $y \geq y'$. 
	By an easy calculation,  
	$f_{GR_{1}}(t_{2}) = f_{GR_{1}}(t_{p}-) + x + x'' - y$ 
	and 
	$f_{OPT_{1}}(t_{2}) = f_{OPT_{1}}(t_{p}-) + x' + x'' - y'$. 
	By the above equality and inequalities, 
	$0 \leq f_{GR_{1}}(t_{2}) - f_{OPT_{1}}(t_{2}) 
		= f_{GR_{1}}(t_{p}-) + x + x'' - y - (f_{OPT_{1}}(t_{p}-) + x' + x'' - y')
		= f_{GR_{1}}(t_{p}-) - f_{OPT_{1}}(t_{p}-) + x - x' - y + y'
		\leq f_{GR_{1}}(t_{p}-) - f_{OPT_{1}}(t_{p}-) + x - x'
	$. 
	That is, 
	$x' \leq f_{GR_{1}}(t_{p}-) - f_{OPT_{1}}(t_{p}-) + x$. 
	Hence, 
	$x' + x'' \leq f_{GR_{1}}(t_{p}-) - f_{OPT_{1}}(t_{p}-) + x + x''
		= f_{GR_{1}}(t_{p}-) - f_{OPT_{1}}(t_{p}-) + w
	$ 
	holds. 
	Furthermore, 
	$f_{GR_{1}}(t_{p}-) - f_{OPT_{1}}(t_{p}-) \leq B - 1$ 
	since $GR_{1}$ accepts $p$, 
	namely, 
	$GR_{1}$'s buffer is not full just before the decision time of $p$.
	Thus, 
	$x' + x'' \leq B - 1 + w$. 
	Finally 
	we consider the case where $t_{1}$ is a time before the beginning of the input. 
	Since $f_{GR_{1}}(t_{p}-) = f_{OPT_{1}}(t_{p}-) = 0$, 
	$x' + x'' \leq f_{GR_{1}}(t_{p}-) - f_{OPT_{1}}(t_{p}-) + w = w
	$ 
	holds. 
	\fi
\end{proof}
%
%

%
\ifnum \count10 > 0
We are ready to give one of the two key lemmas, 
which evaluates the number of $OPT$'s packets whose block number is $M$.
どれほどバースト的にパケットが到着しても、
$MF$は決して通し番号$M$のpacketをflushしないという定義になっている。
このことから、以下の例外的な性質を示すことが出来る。
\fi
\ifnum \count11 > 0
\com{（■英語）}
We are ready to give one of the two key lemmas, 
which evaluates the number of $OPT$'s packets whose block number is $M$.
Even if many packets arrive at a burst, 
$MF$ never flushes any packets with block number $M$ by the definition of $MF$. 
Thus, 
we can get the following exceptional properties. 
\fi
%

%
\begin{LMA}\label{LMA:k.1}
	\ifnum \count10 > 0
	(a) 
	$c = 0$
	もしくは
	$c \in [\lfloor A/2 \rfloor, 3B - 1]$の場合、
	$h_{MF, M}(\tau) \geq \lfloor A/2 \rfloor$
	が成立する。
	(b)
	$c \in [1, \lfloor A/2 \rfloor-1]$の場合、
	$M \geq 2$ならば、
	$h_{MF, M}(\tau) + B - 1 \geq h_{OPT, M}(\tau)$
	が成立する。
	(c) 
	$c \in [1, \lfloor A/2 \rfloor-1]$の場合、
	$M = 1$ならば、
	$h_{MF, M}(\tau) \geq h_{OPT, M}(\tau)$
	が成立する。
	\fi
	\ifnum \count11 > 0
	\com{（■英語）}
	(a) 
	If either $c = 0$ or $c \in [\lfloor A/2 \rfloor, 3B-1]$,
	$h_{MF, M}(\tau) \geq \lfloor A/2 \rfloor$. 
	(b) 
	If $c \in [1, \lfloor A/2 \rfloor-1]$ and $M \geq 2$, 
	$h_{MF, M}(\tau) + B - 1 \geq h_{OPT, M}(\tau)$. 
	(c) 
	If $c \in [1, \lfloor A/2 \rfloor-1]$ and $M = 1$, 
	$h_{MF, M}(\tau) \geq h_{OPT, M}(\tau)$. 
	\fi
\end{LMA}
%
%
\begin{proof}
	\ifnum \count10 > 0
	(a) 
	この証明は補題~\ref{LMA:k.035}と同じである。
	$c = 0$、
	もしくは$c \in [\lfloor A/2 \rfloor, 3B-1]$ならば、
	$MF$は、
	補題~\ref{LMA:ap.1}(b)より、
	通し番号が$M$である1-packetを少なくとも$\lfloor A/2 \rfloor$個はacceptする。
	もし、
	$MF$が一度も通し番号が$M$のpacketをpreemptしないならば、
	明らかに題意がみたされる。
	そこで、
	event time $t$に、
	$MF$が通し番号が$M$である$x (\in [2, k])$-packetをpreemptすると仮定せよ。
	Case 2.2.2の定義より、
	$t-$において$MF$のバッファ内には$A$個の$x$-packetが存在する。
	また、
	補題~\ref{LMA:k.01}より、
	$MF$のバッファ内の$x$-packetの通し番号は単調増加である。
	よって、
	${\ell}(t+, p) \in [\lfloor A/2 \rfloor + 1, A]$
	が成立する
	各$x$-packet $p$に対して、
	$g(p) = M$が成立している。
	よって、
	$h_{MF, M}(t+) \geq \lfloor A/2 \rfloor$
	が成立する。
	これは題意をみたす。
	(b) 
	$c \in [1, \lfloor A/2 \rfloor-1]$が成立するので、
	補題~\ref{LMA:ap.1}(b)より、
	$MF$が受理する通し番号が$M$の1-packetと$GR_{1}$が受理する1-packetは等しい。
	よって、
	$MF$がacceptする通し番号が$M$の最初の1-packetを$p'$とし、
	$h_{MF, M}(\tau)$番目の1-packetを$p''$とする。
	ここで、
	補題~\ref{LMA:k.20}において、
	$t_{1} = t_{p'}-$とし、
	$t_{2} = t_{p''}+$とする。
	このとき、
	$w = h_{MF, M}(\tau)$が成立し、
	$h_{MF, M}(\tau) + B - 1 \geq h_{OPT, M}(\tau)$
	が成立する。
	(c) 同様に補題~\ref{LMA:k.20}より、
	$h_{MF, M}(\tau) \geq h_{OPT, M}(\tau)$が成立する。
	\fi
	\ifnum \count11 > 0
	\com{（■英語）}
	(a) 
	The proof of (a) is almost the same as that of Lemma~\ref{LMA:k.035}. 
	By the assumption that $c = 0$ or $c \in [\lfloor A/2 \rfloor, 3B-1]$, 
	$MF$ accepts at least $\lfloor A/2 \rfloor$ 1-packets with block number $M$ according to Lemma~\ref{LMA:ap.1}(b). 
	If $MF$ does not preempt any packet with block number $M$, the statement is clearly true. 
	Then, 
	suppose that at an event time $t$, $MF$ preempts an $x (\in [2, k])$-packet with block number $M$. 
	By Case 2.2.2 in $MF$, 
	$MF$ stores $A$ $x$-packets in its buffer at $t-$. 
	Moreover, 
	all the $x$-packets in $MF$'s buffer are queued in ascending order by their block numbers by Lemma~\ref{LMA:k.01}. 
	Thus, 
	for each $x$-packet $p$ such that ${\ell}(t+, p) \in [\lfloor A/2 \rfloor + 1, A]$, 
	$g(p) = M$. 
	As a result, 
	$h_{MF, M}(t+) \geq \lfloor A/2 \rfloor$, 
	which proves the lemma. 
	(b) 
	Since $c \in [1, \lfloor A/2 \rfloor-1]$ by the assumption of (b), 
	all the 1-packets with block number $M$ which are accepted by $MF$ are the same as those of $GR_{1}$. 
	Then, 
	let $p'$ ($p''$) be the first ($h_{MF, M}(\tau)$th) 1-packet accepted by $MF$ such that the block number of $p'$ ($p''$) is $M$. 
	Thus, 
	by applying Lemma~\ref{LMA:k.20} with 
	$t_{1} = t_{p'}-$ and $t_{2} = t_{p''}+$, i.e., $w = h_{MF, M}(\tau)$, 
	we have that 
	$h_{MF, M}(\tau) + B - 1 \geq h_{OPT, M}(\tau)$. 
	(c) 
	In the same way as (b), 
	we obtain 
	$h_{MF, M}(\tau) \geq h_{OPT, M}(\tau)$ by Lemma~\ref{LMA:k.20}. 
	\fi
\end{proof}
%
%

%
\ifnum \count10 > 0
以下では、
$M$だけでなく他の通し番号の$OPT$にcompleteされるframeの数を評価する。
そのために、
別の便利なツールを示しておこう。
\fi
\ifnum \count11 > 0
\com{（■英語）}
In what follows, 
we discuss the number of $OPT$'s completed frames. 
Now we get other useful tools in the next two lemmas. 
\fi
%

%
\begin{LMA}\label{LMA:k.02}
	\ifnum \count10 > 0
	任意のnon-event time $t$、
	任意の$x \in [2, k]$、
	$t$において、validな$MF$の任意の$x$-packet $p$に対して、
	$\mbox{arr}(p) < \mbox{arr}(q)$が成立する様な、
	$OPT$の受理する1から$(g(p)-1)$の通し番号の$x$-packet $q$の数は、
	高々$B$個である。
	\fi
	\ifnum \count11 > 0
	\com{（■英語）}
	For any non-event time $t$ and $x \in [2, k]$, 
	let $p$ be an $x$-packet valid for $MF$ at $t$. 
	Then the number of $x$-packets $q$ 
	such that
	$OPT$ accepts $q$, $\mbox{arr}(p) < \mbox{arr}(q)$, 
	and $g(q) \in [1, g(p)-1]$ is at most $B$. 
	\fi
\end{LMA}
%
%
\begin{proof}
	\ifnum \count10 > 0
	$p$の1-packetを$p_{1}$とする。
	$OPT$のacceptする任意の$x$-packetを$q'$とし、
	その1-packetを$q'_{1}$とする。
	$OPT$はcompleteしないframeのpacketをaccpetしないので、
	$q'_{1}$は$OPT$にacceptされる。
	到着するpacketはorder-respectingなので、
	$\mbox{arr}(p_{1}) > \mbox{arr}(q'_{1})$
	が成立する場合、
	$\mbox{arr}(p) \geq \mbox{arr}(q')$
	が成立する。
	すなわち、このとき、
	$q'$は$q$の補題の2番目の条件をみたさない。
	また、
	1-packetの通し番号は、その到着順に関して非減少であるので、
	$\mbox{arr}(p_{1}) < \mbox{arr}(q'_{1})$
	が成立する場合、
	$g(p_{1}) \leq g(q'_{1})$
	が成立する。
	すなわち、
	その様な$q'$は$q$の3番目の条件をみたさない。
	以上より、
	$\mbox{arr}(p_{1}) = \mbox{arr}(q'_{1})$
	が成立する様な$q'$のみが$q$の条件を全てみたしている。
	バッファの大きさは$B$なので、
	$\mbox{arr}(p_{1}) = \mbox{arr}(q'_{1})$
	が成立する様なpacket $q'_{1}$の数は高々$B$個である。
	\fi
	\ifnum \count11 > 0
	\com{（■英語）}
	Let $p_{1}$ be the 1-packet corresponding to $p$, 
	$q'$ be an $x$-packet accepted by $OPT$, 
	and $q'_{1}$ be the 1-packet corresponding to $q'$.
	As we assume that $OPT$ never accepts a packet of an incompleted frame,
	$q'_{1}$ is accepted by $OPT$. 
	Since the input is order-respecting, 
	$\mbox{arr}(p) \geq \mbox{arr}(q')$ 
	if $\mbox{arr}(p_{1}) > \mbox{arr}(q'_{1})$, 
	that is, 
	such $q'$ does not satisfy the second condition of $q$ in the statement of this lemma. 
	Since the block numbers of 1-packets are monotonically non-decreasing in an arrival order, 
	$g(p_{1}) \leq g(q'_{1})$ 
	if $\mbox{arr}(p_{1}) < \mbox{arr}(q'_{1})$, 
	namely, 
	such $q'$ does not satisfy the third condition of $q$. 
	Thus, 
	only $q'$ such that 
	$\mbox{arr}(p_{1}) = \mbox{arr}(q'_{1})$ can satisfy all the conditions of $q$. 
	Since the buffer size is $B$, 
	the number of such $q'_{1}$ accepted by $OPT$ is at most $B$, 
	which completes the proof.
	\fi
\end{proof}
%
%

%
\ifnum \count10 > 0
次の補題では、
2つの$x(\geq 2)$-packet $p$と$p'$が同時に$MF$のバッファ内に存在する場合を考える。
このとき、
どんなアルゴリズムでも（$OPT$でさえも）
$g(p)$から$g(p')-1$の通し番号をもつ$x$-packetをあまり受理できないことを示す。
具体的には、
その数は$5B+A-4$である。
すなわち、
$OPT$は、
$g(p)$から$g(p')-1$の通し番号をもつframeは$O(B)$個しかcompleteできない
ということを意味する。
\fi
\ifnum \count11 > 0
\com{（■英語）}
In the next lemma, 
we consider the case where there exist two $x(\geq 2)$-packets $p$ and $p'$ in $MF$'s buffer at the same time. 
Then any algorithm (even $OPT$) cannot accept many $x$-packets whose block numbers lie between $g(p)$ and $g(p')-1$. 
Specifically, its number is at most $5B+A-4$. 
Hence, 
that means the number of $OPT$'s completed frames with block numbers lying between $g(p)$ and $g(p')-1$ is $O(B)$. 
\fi
%

%
\begin{LMA}\label{LMA:k.21}
	\ifnum \count10 > 0
	任意の$x \in [2, k]$、
	同時刻にMFのバッファ内に存在する$x$-packets $p, p' (\ne p)$
	(ただし、$g(p) + 2 \leq g(p')$が成立する。)
	に対して、
	$g(p) \geq 2$ならば、
	$OPT$がacceptし、
	その通し番号が$g(p)$thから$(g(p')-1)$である様な$x$-packetの数は、高々$5B + A - 4$である。
	また、
	$g(p) = 1$ならば、
	$OPT$がacceptし、
	その通し番号が$1$から$(g(p')-1)$である様な$x$-packetの数は、高々$4B + A - 3$である。
	\fi
	\ifnum \count11 > 0
	\com{（■英語）}
	For any $x \in [2, k]$, and any two $x$-packets $p$ and $p'$ such that $g(p') - g(p) \geq 2$, 
	suppose that both $p$ and $p'$ are stored in $MF$'s buffer at the same time. 
	Then, 
	if $g(p) \geq 2$, 
	the number of $x$-packets $\tilde{p}$ such that $g(\tilde{p}) \in [g(p), g(p')-1]$, and $\tilde{p}$ is accepted by $OPT$ is at most $5B + A - 4$. 
	Moreover, 
	if $g(p) = 1$, 
	the number of $x$-packets $\hat{p}$ such that $g(\hat{p}) \in [1, g(p')-1]$, and $\hat{p}$ is accepted by $OPT$ is at most $4B + A - 3$. 
	\fi
\end{LMA}
\ifnum \count14 < 1
\begin{proof}
	\ifnum \count10 > 0
	（■）
	最初に
	$g(p) \geq 2$が成立する場合を考える。
	まず、
	ある通し番号$u (\in [g(p), g(p')-1])$に対して、
	時刻$\mbox{arr}(p)$より前にarriveする様な$OPT$がacceptするpacketで、
	その通し番号が$u$である$x$-packetを考える。
	$p_{1}$を$p$の1-packetとし、
	$p_{1}$を$MF$がacceptする$g(p)$の$j(\in [1, A])$番目の1-packetであると仮定する。
	また、
	$p'_{1}$を$MF$がacceptする$g(p)$の$1$番目の1-packetであると仮定する。
	このとき、
	補題~\ref{LMA:ap.1}(b)より、
	$MF$がacceptする$g(p)$の$j(\in [1, A])$番目の1-packetを
	$GR_{1}$もまたacceptする。
	よって、
	$t_{1} = t_{p'_{1}}-$、
	$t_{2} = t_{p_{1}}-$
	すなわち、
	$w = j-1$
	として、
	補題~\ref{LMA:k.20}を適用すると、
	$g(p) \geq 2$の場合、
	$p_{1}$のarrival eventより前に$OPT$がacceptする様な、通し番号が$g(p_{1})(=g(p))$である1-packetの数は高々$j-1+B-1 \leq A-1+B-1$個である。
	また、
	$p_{1}$と同時刻にarriveし、
	$OPT$がacceptする1-packetの集合を$P$とする。
	このとき、$P$の要素数は、バッファの大きさより、高々$B$個である。
	その$B$個の1-packet以外の、通し番号が$g(p_{1})$である任意の$1$-packetを$\hat{p}_{1}$とし、
	その$x$-packetを$\hat{p}$とする。
	$\mbox{arr}(\hat{p}_{1}) > \mbox{arr}(p_{1})$より、
	到着するパケットはorder-respectingなので、
	$\mbox{arr}(\hat{p}) \geq \mbox{arr}(p)$が成立する。
	よって、
	通し番号が$g(p_{1}) (= g(p))$であり、$OPT$のacceptする$x$-packetで、
	時刻$\mbox{arr}(p)$より前にarriveする様なpacketの数は、
	高々$2B + A - 2$個である。
	補題~\ref{LMA:k.00}より、
	任意の通し番号$u' (\in [g(p)+1, g(p')-1])$に対して、
	$g(q') = u'$が成立する様な任意の$x$-packet $q'$に対して、
	$\mbox{arr}(p) \leq \mbox{arr}(q')$が成立する。
	以上より、
	時刻$\mbox{arr}(p)$より前にarriveする様な$OPT$がacceptし、通し番号が$u (\in [g(p), g(p')-1])$である様な$x$-packetの数は、
	高々$2B + A - 2$個である。
	次に、
	時刻$\mbox{arr}(p)$以後にarriveする様な、
	$OPT$がacceptするpacketであり、通し番号が$u (\in [g(p), g(p')-1])$である$x$-packet $q''$を考える。
	$t = \mbox{arr}(p)$とおき、
	$d$の直前のdelivery eventのおこる整数時間を$t'$とおく。
	このとき、
	時刻$[t, t']$の間にtransmitされるpacketの数は、
	高々$t' - t + 1$である。
	また、
	$p$は$d$までは、$MF$のバッファ内に存在するので、
	$t' - t + 1 \leq B - 1$
	が成立する。
	時刻$[t, t']$の間に到着して、
	$OPT$にacceptされる$x$-packetの数は、
	高々$B + t' - t$である。
	以上の式より、
	$B + t' - t \leq 2B - 2$
	が成立する。
	よって、
	$OPT$がacceptするpacketであり、
	通し番号が$u (\in [g(p), g(p')-1])$の$x$-packetで、
	時刻$[\mbox{arr}(p), \mbox{arr}(p')]$の間に到着するpacketの数は、
	高々$2B - 2$個である。
	補題\ref{LMA:k.02}より、
	$OPT$がacceptするpacketであり、
	その通し番号が$g(p)$th, ..., もしくは$(g(p')-1)$の$x$-packetで、
	$\mbox{arr}(p') < \mbox{arr}(q')$が成立する様な$x$-packet $q'$の数は、
	高々$B$個である。
	よって、
	$OPT$がacceptするpacketであり、
	その通し番号が$g(p)$から$(g(p')-1)$である$x$-packetの数は、高々$2B + A - 2 + 2B - 2 + B = 5B + A - 4$である。
	次に
	$g(p) = 1$が成立する場合を考える。
	(a)と同様に考えて、
	補題~\ref{LMA:ap.1}と
	補題~\ref{LMA:k.20}より、
	$p_{1}$のarrival eventより前に$OPT$がacceptする様な、通し番号が$1$の1-packetは高々$j-1 \leq A-1$個である。
	よって、上記と同様の議論より、
	$OPT$がacceptするpacketであり、
	その通し番号が$1$から$(g(p')-1)$である様な$x$-packetの数は、高々$B + A - 1 + 2B - 2 + B = 4B + A - 3$である。
	\fi
	\ifnum \count11 > 0
	\com{（■英語）}
	First, 
	we consider the case of $g(p) \geq 2$. 
	Let $q$ be an $x$-packet satisfying the conditions of the lemma, 
	i.e., 
	an $x$-packet $q$ such that $g(q) \in [g(p), g(p')-1]$ 
	and $q$ is accepted by $OPT$. 
	We count the number of such $q$ for each of the cases 
	(i) $\mbox{arr}(q) < \mbox{arr}(p)$, 
	(ii) $\mbox{arr}(p) \leq \mbox{arr}(q) \leq \mbox{arr}(p')$, 
	and (iii) $\mbox{arr}(p') < \mbox{arr}(q)$.
	(i) 
	First, 
	note that there is no $q$ such that $g(q) \in [g(p)+1, g(p')-1]$ by Lemma~\ref{LMA:k.00}, 
	since $\mbox{arr}(q) < \mbox{arr}(p)$. 
	Hence, 
	we suppose that $g(q) = g(p)$. 
	Let $p_{1}$ and $q_{1}$ be the 1-packets corresponding to $p$ and $q$, respectively, 
	and 
	suppose that $p_{1}$ ($p'_{1}$) is the $j$th (first) 1-packet accepted by $MF$ with block number $g(p)$. 
	To count the number of $q$ satisfying the condition, 
	we count the number of corresponding $q_{1}$. 
	Note that $g(q_{1}) = g(p_{1})$ since $g(q) = g(p)$. 
	By Lemma~\ref{LMA:ap.1}(b), 
	the $j(\in [1, A])$th 1-packet accepted by $MF$ is also accepted by $GR_{1}$. 
	The number of $q_{1}$ such that
	$\mbox{arr}(q_{1}) < \mbox{arr}(p_{1})$ is at most $j-1+B-1$ 
	by applying Lemma~\ref{LMA:k.20} with $t_{1} = t_{p'_{1}}-$ and $t_{2} = t_{p_{1}}-$, i.e., $w=j-1$, 
	and this is at most $A+B-2$ since $j\leq A$ by Lemma~\ref{LMA:ap.1}(b). 
	The number of $q_{1}$ such that $\mbox{arr}(q_{1}) = \mbox{arr}(p_{1})$ is at most $B$, 
	since the buffer size is $B$. 
	Finally, 
	the number of $q_{1}$ such that $\mbox{arr}(q_{1}) > \mbox{arr}(p_{1})$ 
	is zero by the order-respecting assumption 
	because $\mbox{arr}(q) < \mbox{arr}(p)$. 
	Hence, 
	the number of $q$ in Case (i) is at most $(A+B-2)+B=2B+A-2$.
	(ii) 
	Let $t$ be any non-event time when both $p$ and $p'$ are stored in $MF$'s buffer. 
	Let $w = \mbox{arr}(p)$ and 
	suppose that the delivery subphase just before $t$ is at the $w'$th phase. 
	Then, 
	the number of delivery subphases during $[w, w']$ is $w'-w+1$. 
	Since $p$ is still stored in $MF$'s buffer at $t$, 
	$w'-w+1 \leq B-1$ 
	(as otherwise, $MF$ must have transmitted $p$ before $t$). 
	The number of $x$-packets which arrive during $[w, w']$ and are accepted by $OPT$ is at most $B+w'-w \leq 2B-2$. 
	Thus, 
	the number of $x$-packets $q$ in this case is at most $2B-2$.
	(iii) 
	By Lemma~\ref{LMA:k.02}, 
	the number of $x$-packets $q$ in this case is at most $B$.
	Putting (i), (ii), and (iii) together, 
	the number of $x$-packets $q$ is at most $(2B+A-2) + (2B-2) + B = 5B + A - 4$.
	For $g(p)=1$, 
	the argument is the same as the case of $g(p)\geq 2$,
	except that at an application of Lemma~\ref{LMA:k.20} in Case (i), 
	we let $t$ be the time before the beginning of the input. 
	Then, 
	the number of $q_{1}$ such that $q_{1}$ is accepted by $OPT$, 
	$g(q_{1}) = g(p_{1})$, and $\mbox{arr}(q_{1}) < \mbox{arr}(p_{1})$ is at most $A - 1$, 
	instead of $A+B-2$ in the case of $g(p)\geq 2$. 
	Then the number of $x$-packets $q$ in question is at most $(B+A-1)+(2B-2)+B=4B+A-3$.
	\fi
\end{proof}
\fi
\ifnum \count10 > 0
補題~\ref{LMA:k.21}を用いるための準備をする。
具体的には、
連続する2つのgood通し番号のpacketは
ある時間において、
$MF$のバッファ内に同時に存在することを示す。
\fi
\ifnum \count11 > 0
\com{（■英語）}
We prepare for using Lemma~\ref{LMA:k.21}. 
Specifically, 
we show that some two packets whose block numbers are $a_{j}$ and $a_{j+1}$, respectively, exist in $MF$'s buffer at some time simultaneously. 
\fi
%

%
\begin{LMA}\label{LMA:k.22}
	\ifnum \count10 > 0
	（■）
	任意のnon-event time $d$に対して、
	$a_{j}(d) + 2 \leq a_{j+1}(d)$
	が成立する様な、
	任意の$j \in [1, m(d)-1]$に対して、
	あるnon-event time $d' (< d)$、ある$x \in [2, k]$に対して、
	$d'$に、
	通し番号が$a_{j}(d)$である$x$-packet $q''$
	と
	通し番号が$a_{j+1}(d)$である$x$-packet $q'$
	が
	共に$MF$のバッファ内に存在する。
	\fi
	\ifnum \count11 > 0
	\com{（■英語）}
	Suppose that $a_{j+1} - a_{j} \geq 2$ for an integer $j (\in [1, m-1])$.
	Then there exist two $x$-packets $q$ and $q'$ for some
	integer $x \in [2, k]$ such that $g(q) = a_{j}$,
	$g(q') = a_{j+1}$, 
	and both $q$ and $q'$ are stored in $MF$'s buffer at the same time.
	\fi
\end{LMA}
\ifnum \count14 < 1
\begin{proof}
	\ifnum \count10 > 0
	For a non-event time $t$, 
	we say that a block number $u$ is {\em good} at $t$ 
	if $u = M$ or at least $\lfloor A/2 \rfloor$ frames with the block number $u$ are valid at $t$, 
	and {\em bad} at $t$ otherwise. 
	Note that the set of good block numbers at the end of the input coincides the set $G$ 
	(see Sec.~\ref{overview_k} for the definition of $G$).
	Since $a_{j} + 2 \leq a_{j+1}$, 
	there must be at least one block number between $a_{j}$ and $a_{j+1}$. 
	Those block numbers were initially good but turned bad at some event by the execution of $MF$, 
	since $a_{j}$ and $a_{j+1}$ are consecutive good block numbers at the end of the input.
	Let $u$ ($a_{j} < u < a_{j+1}$) be the block number that turned bad lastly during 入力が始まってから終わるまでの間. 
	通し番号$u$がturn badになるのは、
	ある$x (\in [2, k])$-packet $p'$が到着したときであるdecision time $t_{p'}$に起こる。
	具体的には、
	$t_{p'}$において$MF$は$p'$を受理し、
	Case 2.2.2を実行して、
	${\ell}(t_{p'}-, p) = \lfloor A/2 \rfloor + 1$が成立する
	通し番号が$u(=g(p))$の$x$-packet $p$をpreemptする。
	更に、
	$MF$がCase 2.2.2.1を実行して、
	$MF$のバッファ内の通し番号が$u$のpacketを全て破棄したときに起こる。
	以下では、
	$q''$と$q'$が
	$t_{p'}+$において、
	$MF$のバッファ内に存在することを示す。
	ここで、
	$t_{p'}$の前後にバッファ内に存在するパケットの通し番号について確認する。
	$MF$のCase 2.2.2の定義より、
	$t_{p'}-$における$MF$のバッファ内の$x$-packetの数は$A$であり、
	$t_{p'}-$における
	$MF$のバッファ内の、通し番号が$g(p)$である$x$-packetの数は、
	$p$を除いて$\lfloor A/2 \rfloor - 1$個である。
	更に、
	補題~\ref{LMA:k.01}より、
	$MF$のバッファ内の$x$-パケットの通し番号は昇順である。
	よって、
	$g(p) > g(p'')$
	である。
	（式(a)）
	ただし、
	${\ell}(t_{p'}-, p'') = {\ell}(t_{p'}+, p'') = 1$
	が成立する$x$-packetを$p''$とする。
	また、
	$p'$をacceptしてCase 2.2.2.1を実行してバッファ内の通し番号が$g(p)$であるpacketを全てpreemptするので、
	$g(p') \ne g(p)$が成立する。
	よって、
	補題~\ref{LMA:k.01}より、
	$g(p) < g(p')$
	が成立する。
	（式(b)）
	Now we discuss the existence of $q''$. 
	まず$a_{j} < g(p'')$と仮定してみる。
	このとき、
	(式(a))と合わせて、
	$g(p) > g(p'') > a_{j}$
	が成立する。
	これは、
	$a_{j}$と$u(=g(p))$の間にgood通し番号が存在しない、
	という$u$の定義に反する。
	よって、矛盾である。
	すなわち、
	$a_{j} \geq g(p'')$
	が成立している。
	同様に、
	$a_{j+1} > g(p')$と仮定すると、
	(式(b))と合わせて、
	$g(p) < g(p') < a_{j+1}$
	が成立し、矛盾を導ける。
	結果、
	$a_{j+1} \leq g(p')$ (式(c))
	が成立している。
	次に、
	$a_{j} = g(p'')$が成立する場合を考えよう。
	このとき、
	$q'' = p''$と定義する。
	確かに、
	$q''$は$t_{p'}+$においてバッファ内に存在しており、
	題意を満たす。
	最後に、
	$a_{j} > g(p'')$が成立する場合を考える。
	$q''$の定義より
	$a_{j} = g(q'')$が成立しており、
	(式c)と合わせて、
	$g(p') \geq a_{j+1} > a_{j} = g(q'') > g(p'')$
	が成立する。
	よって、
	補題~\ref{LMA:k.00}より、
	$\mbox{arr}(p') \geq \mbox{arr}(q'') \geq \mbox{arr}(p'')$
	が成立する。
	$p'$と$p''$は$t_{p'}+$において$MF$のバッファ内に存在し、
	$t_{p'}+$において$q''$はvalidであるので、
	$q''$は$t_{p'}+$において$MF$のバッファ内に存在する。
	同様に$q'$についても考えて、
	$a_{j+1} = g(p')$が成立する場合、
	$q' = p'$と定義する。
	また、
	$a_{j+1} < g(p')$が成立する場合、
	$g(p') > g(q') = a_{j+1} > a_{j} \geq g(p'')$
	が成立し、
	補題を示すことが出来る。
	\fi
	\ifnum \count11 > 0
	\com{（■英語）}
	For a non-event time $t$, 
	we say that a block number $u$ is {\em good} at $t$ 
	if $u = M$ or at least $\lfloor A/2 \rfloor$ frames with the block number $u$ are valid at $t$, 
	and {\em bad} at $t$ otherwise. 
	Note that the set of good block numbers at the end of the input coincides the set $G$ 
	(see Sec.~\ref{overview_k} for the definition of $G$). 
	Since $a_{j+1} - a_{j} \geq 2$, 
	there must be at least one block number between $a_{j}$ and $a_{j+1}$. 
	Those block numbers were initially good but turned bad at some event, 
	since $a_{j}$ and $a_{j+1}$ are good block numbers that are consecutive at the end of the input.
	Let $u$ ($a_{j} < u < a_{j+1}$) be the block number that turned bad lastly. 
	The event time when block number $u$ turns bad is the decision time $t_{p'}$ when some $x (\in [2, k])$-packet $p'$ arrives. 
	Specifically, 
	$MF$ accepts $p'$ at $t_{p'}$, and 
	preempts an $x$-packet $p''$ with block number $u(= g(p''))$ at Case 2.2.2 
	such that ${\ell}(t_{p'}-, p'') = \lfloor A/2 \rfloor + 1$. 
	Moreover, 
	$MF$ preempts all the packets with block number $u$ in $MF$'s buffer by executing Case 2.2.2.1. 
	Now we discuss the block numbers of packets in $MF$'s buffer before or after $t_{p'}$. 
	By the definition of Case 2.2.2 in $MF$, 
	the number of $x$-packets in $MF$'s buffer at $t_{p'}-$ is $A$, and 
	the number of $x$-packets whose block number is $g(p'')$ in $MF$'s buffer at $t_{p'}-$ except $p''$ is $\lfloor A/2 \rfloor - 1$. 
	In addition, 
	all the $x$-packets in $MF$'s buffer are queued in ascending order by their block numbers by Lemma~\ref{LMA:k.01}. 
	Hence,
	(a) $g(p'') > g(p)$ holds, 
	where $p$ is the $x$-packet such that ${\ell}(t_{p'}-, p) = {\ell}(t_{p'}+, p) = 1$. 
	Also, 
	$MF$ accepts $p'$, and preempts all the packets with block number $g(p'')$ at Case 2.2.2.1. 
	Thus, 
	$g(p') \ne g(p'')$, 
	which means that (b) $g(p'') < g(p')$ holds according to Lemma~\ref{LMA:k.01}. 
	Now if $a_{j} < g(p)$, 
	then $g(p'') > g(p) > a_{j}$ by (a). 
	This contradicts the definition of $u$, namely, 
	the definition that there does not exist any good block number between $a_{j}$ and $u(=g(p''))$. 
	Hence $a_{j} \geq g(p)$. 
	In the same way, 
	if $a_{j+1} > g(p')$, 
	then $g(p'') < g(p') < a_{j+1}$ by (b). 
	We have the contradiction as well, which means that 
	(c) $a_{j+1} \leq g(p')$. 
	In the following, 
	we prove that $q$ and $q'$ mentioned in this lemma exist in the buffer at time $t_{p'}+$.
	We first show the existence of $q$. 
	Let us consider the case of $a_{j} = g(p)$. 
	In this case, 
	$p$ is clearly stored in $MF$'s buffer at $t_{p'}+$, and 
	$p$ satisfies the condition of $q$. 
	Next, we consider the case of $a_{j} > g(p)$. 
	Since $a_{j}$ is a good block number by definition, 
	there must be a packet $p'''$ such that $a_{j} = g(p''')$ and $p'''$ is valid at $t_{p'}+$.
	Then, 
	$g(p') \geq a_{j+1} > a_{j} = g(p''') > g(p)$ by (c) and 
	hence 
	$\mbox{arr}(p') \geq \mbox{arr}(p''') \geq \mbox{arr}(p)$ by Lemma~\ref{LMA:k.00}. 
	Note that 
	$MF$ stores both $p'$ and $p$ in its buffer at $t_{p'}+$, 
	and $p'''$ is valid at $t_{p'}+$ by the above definition. 
	Therefore, 
	$p'''$ is stored in $MF$'s buffer at $t_{p'}+$, and 
	thus this $p'''$ satisfies the condition of $q$.
	The case of $q'$ can be proven in the same way as $q$. 
	Namely, if $a_{j+1} = g(p')$, 
	then let $q' = p'$. 
	Also, 
	if $a_{j+1} < g(p')$, 
	then
	there must be $q'$ satisfying $g(p') > g(q') = a_{j+1} > a_{j} \geq g(p)$. 
	This completes the proof. 
	\fi
\end{proof}
\fi
\ifnum \count10 > 0
Now we are ready to give the last key lemma. 
\fi
\ifnum \count11 > 0
\com{（■英語）}
Now we are ready to give the last key lemma. 
\fi
%

%
\begin{LMA}\label{LMA:k.2}
	\ifnum \count10 > 0
	(■)
	通し番号が$a_{1}$から$a_{2}-1$である様な$OPT$のcompleteなフレームの数は、
	高々$4B + A - 3$である。
	また、
	任意の$i \in [2, m-1]$に対して、
	通し番号が$a_{i}$から$a_{i+1}-1$である様な$OPT$のcompleteなフレームの数は、
	高々$5B + A - 4$である。
	更に、
	通し番号が$a_{m}$である$OPT$のcompleteなフレームの数は、
	高々$4B - 1$である。
	\fi
	\ifnum \count11 > 0
	\com{（■英語）}
	(a) 
	The number of frames $f$ completed by $OPT$ 
	such that $g(f) = a_{m}$ is at most $4B-1$. 
	(b) 
	The number of frames $f$ completed by $OPT$ such that $g(f) \in [a_{1}, a_{2}-1]$ is at most $4B+A-3$. 
	(c) 
	For any $i \in [2, m-1]$, 
	the number of frames $f$ completed by $OPT$ such that $g(f) \in [a_{i}, a_{i+1}-1]$ is at most $5B+A-4$.
	\fi
\end{LMA}
%
%
\begin{proof}
	\ifnum \count10 > 0
	Lemma~\ref{LMA:k.22}より、
	$a_{j'}(\tau) + 2 \leq a_{j'+1}(\tau)$
	が成立する様な任意の$j' \in [1, m(\tau)-1]$、
	あるnon-event time $d'(< \tau)$、ある$x (\in [2, k])$に対して、
	$d'$に、
	通し番号が$a_{j'}(\tau)$の$x$-packet $p$
	と
	通し番号が$a_{j'+1}(\tau)$の$x$-packet $p'$
	が
	共に$MF$のバッファ内に存在する。
	Lemma~\ref{LMA:k.21}より、
	任意の$y \in [2, k]$に対して、
	$a_{1}(\tau) + 2 \leq a_{2}(\tau)$が成立する様な
	通し番号が$a_{1}(\tau)$から$a_{2}(\tau)-1$である様な$OPT$の$y$-packetの数は、高々$4B + A - 3$である。
	また、
	$a_{j}(\tau) + 2 \leq a_{j+1}(\tau)$が成立する様な
	任意の$j \in [2, m(\tau)-1]$に対して、
	通し番号が$a_{j}(\tau)$から$a_{j+1}(\tau)-1$である様な$OPT$の$y$-packetの数は、
	高々$5B + A - 4$である。
	また、
	任意の通し番号$u$に対して、
	$GR_{1}$がacceptする通し番号$u$の最初のパケットを$q$、
	$GR_{1}$がacceptする通し番号$u$の最後、すなわち、$3B$番目のパケットを$q'$とする。
	このとき、
	$t_{1} = t_{q}-$、
	$t_{2} = t_{q'}+$、
	すなわち、
	$w = 3B$として
	補題~\ref{LMA:k.20}を適用すると、
	$OPT$の任意の通し番号の1-packetの数は、
	高々$3B + B - 1 = 4B - 1$である。
	\com{（■（$\gamma+B-1$）}
	よって、
	$OPT$の各通し番号の$y$-packetの数は、高々$4B - 1$である。
	すなわち、
	$j = m$
	と
	$a_{j}(\tau) + 1 = a_{j+1}(\tau)$が成立する様な
	任意の$j \in [2, m(\tau)-1]$に対して、
	通し番号が$a_{j}(\tau)$である様な$OPT$の$y$-packetの数は、
	高々$4B - 1$である。
	よって、
	任意の$j \in [2, m(\tau)-1]$に対して、
	通し番号が$a_{j}(\tau)$から$a_{j+1}(\tau)-1$である様な$OPT$の$y$-packetの数は、
	高々$5B + A - 4$である。
	仮定より、
	$OPT$はcompleteなフレームに属するpacketのみacceptする。
	よって、
	題意が満たされる。
	\fi
	\ifnum \count11 > 0
	\com{（■英語）}
	Fix the block number $u (\neq M)$. 
	We count the number of 1-packets $p$ accepted by $OPT$ such that $g(p)=u$. 
	Note that the number of 1-packets with block number $u$ accepted by $GR_{1}$ is at most $3B$. 
	Let $q$ ($q'$) be the first (last, i.e., $3B$th) 1-packet accepted by $GR_{1}$ with block number $u$. 
	Also, let $q''$ be the first 1-packet accepted by $GR_{1}$ after $t_{q'}+$. 
	Then $q''$ has the block number $u+1$ by definition, 
	and hence any packet with block number $u$ arrive during time $[t_{q}-, t_{q''}-]$. 
	By applying Lemma~\ref{LMA:k.20} with $t_{1} = t_{q}-$ and $t_{2} = t_{q''}-$, i.e., $w=3B$, 
	the number of 1-packets $p$ accepted by $OPT$ such that $g(p)=u$ is at most $3B+B-1=4B-1$. 
	When $u=M$, 
	the same upper bound can be obtained by a slight modification of the above argument. 
	We use this fact several times in the following.
	(a) 
	By the above discussion, the number of 1-packets $p$ accepted by
	$OPT$ such that $g(p)=a_{m}$ is at most $4B-1$.  Since $OPT$ never
	preempts a packet and $OPT$ accepts only packets of completed frames by
	assumption, the number of frames $f$ completed by $OPT$ such that
	$g(f)=a_{m}$ is at most $4B-1$.
	(b) 
	In the case of $a_{1}=a_{2}-1$, by the same argument as (a) we can
	conclude that the number of completed frames is at most $4B-1 \leq 4B+A-3$. 
	If $a_{1}+2 \leq a_{2}$, 
	we know by Lemma~\ref{LMA:k.22} that
	two $x$-packets $\hat{p}$ and $\tilde{p}$ such that 
	$g(\hat{p})=a_{1}$ and $g(\tilde{p})=a_{2}$ are stored in $MF$'s buffer at the same time. 
	Then
	by Lemma~\ref{LMA:k.21}, 
	the number of $x$-packets $p$ accepted by $OPT$ such that $g(p) \in [a_{1}, a_{2}-1]$ is at most $4B+A-3$ (recall that
	$a_{1}=1$). 
	By the same argument as above, 
	we can conclude that the number of frames completed by $OPT$ such that 
	$g(f) \in [a_{1}, a_{2}-1]$ is also at most this number.
	(c) 
	The argument is almost the same as (b) and hence is omitted.
	\fi
\end{proof}
%
%

\section{Lower Bound for Deterministic Algorithms} \label{sec:LB}
\ifnum \count10 > 0
本節では、
決定性アルゴリズムに対する下限を示す。
\fi
\ifnum \count11 > 0
In this section, 
we show a lower bound for any deterministic algorithm. 
\fi
%
%
\ifnum \count14 > 0
%
The proof of the following theorem is included in Appendix~\ref{sec:ap.3}. 
\fi
\begin{THM}\label{thm:2}
	\ifnum \count10 > 0
	$k \geq 2$とする。、
	$B \geq k-1$ならば、
	任意の決定性アルゴリズムに対する競合比は少なくとも$\frac{2B}{\lfloor {B/(k-1)} \rfloor} + 1$である。
	また、
	$B \leq k-2$ならば、
	任意の決定性アルゴリズムに対する競合比は発散する。
	\fi
	\ifnum \count11 > 0
	\com{（■英語）}
	Suppose that $k \geq 2$. 
	The competitive ratio of any deterministic algorithm is 
	at least $\frac{2B}{\lfloor {B/(k-1)} \rfloor} + 1$ 
	if $B \geq k-1$, and unbounded if $B \leq k-2$. 
	\fi
\end{THM}
\ifnum \count14 < 1
\begin{proof}
	\ifnum \count10 > 0
	オンラインアルゴリズム$ALG$を固定して考える。
	次の様な入力$\sigma$を考える。
	phase $0$に$2B$個の1-パケットが到着する。
	このとき、$ALG$は$x (\leq B)$個のパケットを受理する。
	一方で、$OPT$は$ALG$が受理しないパケットを$B$個受理する。
	$ALG$が受理する$x$個のパケットの集合を$C$と呼び、
	$OPT$が受理する$B$個のパケットの集合を$D$と呼ぶ。
	（図~\ref{fig:LBCR2km1}参照。）
	$B$回の送信サブフェイズの後、
	$B + \lfloor \frac{B}{k-1} \rfloor$個の1-パケットが到着する。
	このとき、$ALG$は$y (\leq B)$個のパケットを受理する。
	一方で、$OPT$は$ALG$が受理しないパケットを$\lfloor \frac{B}{k-1} \rfloor$個だけ受理する。
	$ALG$が受理する$y$個のパケットの集合を$E$と呼び、
	$OPT$が受理する$\lfloor \frac{B}{k-1} \rfloor$個のパケットの集合を$F$と呼ぶ。
	$B$回の送信サブフェイズの後、
	phase $2B$に$2B$個の1-パケットが到着する。
	このとき、$ALG$は$z (\leq B)$個のパケットを受理する。
	一方で、$OPT$は$ALG$が受理しないパケットを$B$個受理する。
	$ALG$が受理する$z$個のパケットの集合を$G$と呼び、
	$OPT$が受理する$B$個のパケットの集合を$H$と呼ぶ。
	以上で、入力中の全ての1-packetは到着した。
	時刻$2B$より後は、各$j( \geq 2)$-packetが到着する。
	各$j = 2, ..., k$に対して、
	phase $3B + (j - 2)B = (j+1)B$に、
	パケット集合$D$の$B$個の$j$-パケットが到着し、
	$OPT$はそれらを全て受理し送信する。
	一方、
	phase $(k + 2)B$に、
	パケット集合$C, E, F, G$の2-パケットから$k$-パケット全てが一度に到着する。
	このとき、
	$OPT$は、$F$に対応する$\lfloor \frac{B}{k-1} \rfloor (k-1) \leq B$個のpacketをacceptする。
	一方、
	$ALG$がaccpetできるそれらのpacketは高々$B$個であるので、
	$ALG$がcompleteできるフレームの数は高々$\lfloor \frac{B}{k-1} \rfloor$個である。
	その到着フェイズの後、
	パケット集合$H$の2-パケットから$k$-パケットが到着し、
	$OPT$はそれらを全て受理し送信する。
	以上より、
	$V_{ALG}(\sigma) \leq \lfloor \frac{B}{k-1} \rfloor$
	と
	$V_{OPT}(\sigma) = 2B + \lfloor \frac{B}{k-1} \rfloor$
	が成立する。
	よって、
	$B \geq k-1$ならば、
	$\frac{V_{OPT}(\sigma)}{V_{ALG}(\sigma)} 
		\geq \frac{2B + \lfloor \frac{B}{k-1} \rfloor}{\lfloor \frac{B}{k-1} \rfloor} = \frac{2B}{\lfloor \frac{B}{k-1} \rfloor} + 1$が成立する。
	また、
	$B \leq k-2$ならば、
	競合比は発散する。
	\fi
	\ifnum \count11 > 0
	\com{（■英語）}
	Fix an online algorithm $ALG$. 
	Let us consider the following input $\sigma$. 
	(See Figure~\ref{fig:LBCR2km1}.) 
	At the 0th phase, 
	$2B$ 1-packets arrive. 
	$ALG$ accepts at most $B$ 1-packets, 
	and $OPT$ accepts $B$ 1-packets that are not accepted by $ALG$. 
	Let $C$ ($D$, respectively) be the set of the 1-packets accepted by $ALG$ ($OPT$, respectively). 
	At the $i$th phase ($i \in [1,B-1]$), 
	no packets arrive. 
	Hence, just after the $(B-1)$st phase, 
	both $ALG$'s and $OPT$'s queues are empty 
	(since $B$ delivery subphases occur).
	At the $B$th phase, 
	$B + \lfloor \frac{B}{k-1} \rfloor$ 1-packets arrive in the same manner as the first $2B$ 1-packets. 
	$ALG$ can accept at most $B$ 1-packets, and 
	$OPT$ accepts $\lfloor \frac{B}{k-1} \rfloor$ 1-packets that are not accepted by $ALG$. 
	Let $E$ ($F$, respectively) be the set of the packets accepted by $ALG$ ($OPT$, respectively). 
	At the $i$th phase ($i \in [B+1,2B-1]$), 
	no packets arrive, 
	and both $ALG$'s and $OPT$'s queues are empty
	just after the $(2B-1)$st phase. 
	Once again
	at the $2B$th phase, $2B$ 1-packets arrive. 
	$ALG$ accepts at most $B$ 1-packets, 
	and $OPT$ accepts $B$ 1-packets that are not accepted by $ALG$. 
	Let $G$ ($H$, respectively) be the set of the 1-packets accepted by $ALG$ ($OPT$, respectively). 
	This is the end of the arrivals and deliveries of 1-packets.
	At the $i$th phase ($i \in [2B+1,3B-1]$), 
	no packets arrive, 
	and hence just before the $3B$th phase, 
	both $ALG$'s and $OPT$'s queues are empty. 
	For each $j = 2, ..., k$, 
	the $B$ $j$-packets corresponding to 1-packets in $D$ arrive at the $(j+1)B$th phase. 
	$OPT$ accepts and transmits them.
	(There is no incentive for $ALG$ to accept them.)
	Next, all the packets corresponding to all the 1-packets in $C \cup E \cup F \cup G$ arrive at the $(k+2)B$th phase. 
	Since $ALG$ needs to accept all the $k-1$ packets of the same frame to complete it, 
	the number of frames $ALG$ can complete is at most $\lfloor \frac{B}{k-1} \rfloor$. 
	$OPT$ accepts all the $\lfloor \frac{B}{k-1} \rfloor (k-1)$ packets corresponding to all the 1-packets in $F$. 
	Note that this is possible because $\lfloor \frac{B}{k-1} \rfloor (k-1) \leq B$. 
	Hence, $OPT$ completes all the $\lfloor \frac{B}{k-1} \rfloor$ frames of $F$. 
	After which 
	all the packets corresponding to 1-packets in $H$ arrive one after the other, and 
	$OPT$ can accept and transmit them. 
	Note that the input sequence is order-respecting.
	By the above argument, 
	we have $V_{ALG}(\sigma) \leq \lfloor \frac{B}{k-1} \rfloor$
	and 
	$V_{OPT}(\sigma) = 2B + \lfloor \frac{B}{k-1} \rfloor$. 
	Therefore, 
	if $B \geq k-1$, 
	$\frac{V_{OPT}(\sigma)}{V_{ALG}(\sigma)} \geq \frac{2B}{\lfloor \frac{B}{k-1} \rfloor} + 1$. 
	If $B \leq k-2$, 
	the competitive ratio of $ALG$ is unbounded. 
	\fi
\end{proof}
\fi
\ifnum \count14 < 1
\ifnum \count12 > 0
\begin{figure*}
	 \begin{center}
	  \includegraphics[width=150mm]{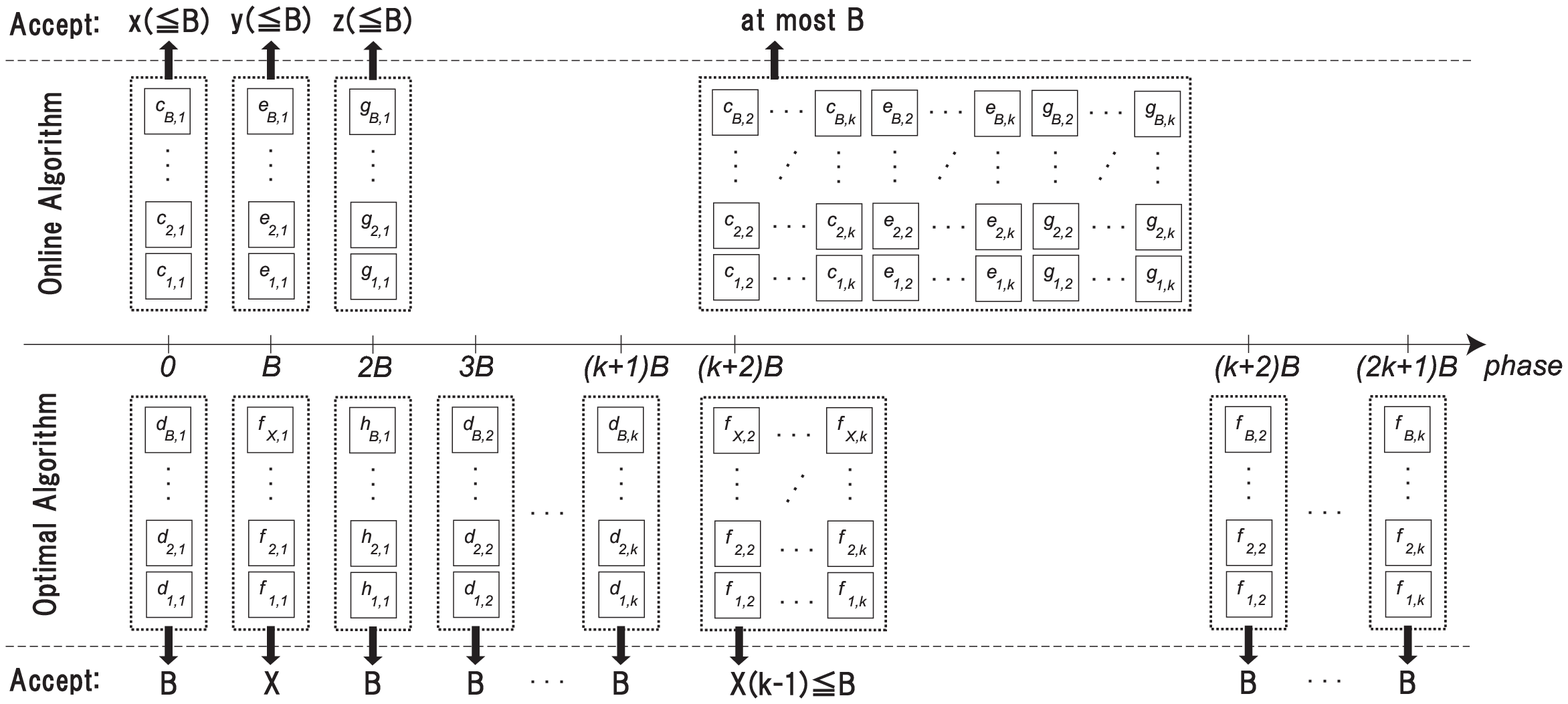}
	 \end{center}
	 \caption{Lower Bound Instance. 
	 	Each square denotes an arriving packet accepted by an online algorithm or $OPT$. 
		In the figure $X = \lfloor \frac{B}{k-1} \rfloor$. 
	 	}
	\label{fig:LBCR2km1}
\end{figure*}
\fi
\fi
%

\section{Lower Bound for Randomized Algorithms} \label{sec:rand}
%

%
\ifnum \count14 > 0
%
The proof of the following theorem is included in Appendix~\ref{sec:ap.3}.
\fi
\begin{THM}\label{thm:4}
	\ifnum \count10 > 0
	$k \geq 3$に対して、
	任意の確率アルゴリズムの競合比は少なくとも$k-1$である。
	\fi
	\ifnum \count11 > 0
	\com{（■英語）}
	When $k \geq 3$, 
	the competitive ratio of any randomized algorithm is at least $k - 1 - \epsilon$ for any constant $\epsilon$ against an oblivious adversary.
	\fi
\end{THM}
\ifnum \count14 < 1
\begin{proof}
	\ifnum \count10 > 0
	まず、本証明を概説する。
	オンライン確率アルゴリズムには$k-1$通りの選択肢をもつ入力が与えられる。
	$k-1$通りの選択肢の中には、1つだけ多くのframeをcomplete出来る選択肢がある（called good）。
	のこりの$k-2$の選択肢は、どんなに頑張ってもそれと比べてごくわずかのframeしかcomplete出来ない。
	オンラインアルゴリズムには選択を終えた後に、その結果が提示される。
	オンラインアルゴリズムは、もちろん、good選択肢を知らないが、
	$OPT$はその選択肢をしっている。
	例えば、1番目の選択肢がgoodだとしよう。
	オンラインはたまたま1番目の選択肢を重視しているかもしれない。
	そこで、
	更に、
	$k-1$個の全く同じ選択肢をもち、当たりの選択肢だけが異なる$k-1$通りの入力を考えよう。
	オンラインアルゴリズムが選択を行う前の段階では、
	$k-1$個の入力は全て同じ入力にみえる様になっている。
	すなわち、
	1番目の選択肢を重視するアルゴリズムは、
	残りの$k-2$個の入力において大損する可能性があるのである。
	用語を定義する。
	任意のevent time $t$、任意の整数$x \in [1, k]$に対して、
	に$x$-パケットが$B$個到着し、
	そのあと、$t$から$t + B-1$まで$B$回送信フェイズが実行される。
	この部分入力列を{\em $\boldsymbol{x}$-サブラウンド}と呼ぶ。
	十分大きな整数$y$に対して、
	$x$-サブラウンドが$y$回繰り返された場合、
	その部分入力列を{\em $\boldsymbol{x}$-ラウンド}と呼ぶ。
	（図~\ref{fig:LBR1}参照。）
	$1$-ラウンドから$k$-ラウンドまで連続した部分入力列を、
	{\em 良いラウンド}と呼ぶ。
	$1$-ラウンドから$k-1$-ラウンドまで連続した後に、
	それらの$By$個の$k$-パケットが全て一度に到着する入力を、
	{\em 悪いラウンド}と呼ぶ。
	それらを踏まえて、
	次の様$\sigma$な入力を考える。
	$\sigma$は、$k-1$個のラウンド：
	$1$個の良いラウンドと$k-2$個の悪いラウンドから構成されている。
	具体的には、$i \in [1, k-2]$に対して、
	$i+1$番目のラウンドは、$i$番目のラウンドの$1$-サブラウンドが終了した直後に開始される。
	よって、
	ある時刻に
	$1$番目のラウンドの$k-1$-サブラウンド、
	$2$番目のラウンドの$k-2$-サブラウンド、
	$3$番目のラウンドの$k-3$-サブラウンド、
	$\cdots$、
	$k-1$番目のラウンドの$1$-サブラウンドが同時に開始される。
	（図~\ref{fig:LBR2}参照。）
	このとき、
	オンラインアルゴリズム$ALG$はどのラウンドが良いフレームなのか分からないので、
	良いラウンドの$yB$個のフレームのうち、
	高々$yB/(k-1)$個しか完成させることが出来ない。
	一方で、$OPT$は良いラウンドの全てのフレームを構成するパケットを受理することが出来る。
	よって、
	$\frac{V_{OPT}(\sigma)}{{\mathbb E}[V_{ALG}(\sigma)]} \geq \frac{yB}{yB / (k-1)} = k - 1$が成立する。
	\fi
	\ifnum \count11 > 0
	\com{（■英語）}
	Fix an arbitrary randomized online algorithm $ALG$.  Let $y$ be a large
	integer that will be fixed later.  Our adversarial input $\sigma$
	consists of $(k-1)yB$ frames.  These frames are divided into $k-1$
	groups each with $yB$ frames. 
	Also, frames of each group are divided into $y$ subgroups each with $B$ frames.  For each $i (\in [1,k-1])$ and
	$j (\in [1,y])$, let $F(i,j)$ be the set of frames in the $j$th subgroup
	of the $i$th group and let $F(i)=\cup_{j} F(i,j)$.  For each $x (\in
	[1,k])$, let $P(i,j,x)$ be the set of $x$-packets of the frames in
	$F(i,j)$ and let $P(i,x)=\cup_{j} P(i,j,x)$.
	
	We first give a very rough idea of how to construct the adversary.
	Among the $k-1$ groups defined above, one of them is a good group. 
	In the first half of the input (from phase $0$ to phase $(k-1)yB-1$), 
	the adversary gives packets to the online algorithm in such a way that the algorithm cannot distinguish the good group. 
	Also, since the buffer size is bounded,
	the algorithm must give up many frames during the first half; 
	only $yB$ frames can survive at the end of the first half. 
	In the second half of the input, 
	remaining packets are given in such a way that $k$-packets from the bad groups arrive at a burst, 
	while $k$-packets from the good group arrive one by one. 
	Hence, 
	if the algorithm is lucky enough to keep many packets of the good group (say, Group 1) at the end of the first half, then it can complete many frames eventually. 
	However, 
	such an algorithm behaves very poorly for an input in which Group 1 is bad. 
	Therefore, the best strategy of an online algorithm (even randomized one) is to keep equal number of frames from each group during the first half. 
	
	Before showing our adversarial input, we define a subsequence of an input. 
	For any $t$, suppose that $B$
	packets of $P(i,j,x)$ arrive at the $t$th phase and no packets arrive
	during $t+1$ through $(t+B-1)$st phases. 
	Let us call this subsequence a {\em subround of $P(i,j,x)$ starting at the $t$th phase}. 
	Notice that if we focus on
	a single subround, an algorithm can accept and transmit all the packets
	of $P(i,j,x)$ by the end of the subround.  A {\em round of $P(i,x)$
	starting at the $t$th phase} is a concatenation of $y$ subrounds of $P(i,j,x)$ ($j\in [1, y]$), 
	where each subround of $P(i,j,x)$ starts at the $(t+(j-1)B)$th phase.
	(See the left figure in Fig.~\ref{fig:LBR1}.)
	
	Our input consists of rounds of $P(i,x)$ starting at the $(i+x-2)yB$th phase, 
	for $i\in [1, k-1]$ and $x\in [1, k-1]$. 
	(See Fig.~\ref{fig:LBR2}.)
	Note that any two rounds $P(i,x)$ and $P(i',x')$ start simultaneously if $i + x = i' + x'$. 
	Currently, 
	the specification of the arrival of packets in $P(i,x)$ for $x=k$ is missing. 
	This is the key for the construction of our adversary and will be explained shortly.
	
	Consider $k-1$ rounds (of $P(1,k-1), P(2,k-2), \cdots, P(k-1,1)$) starting at the $(k-2)yB$th phase, 
	which occur simultaneously.  Note that for each
	$j$, at the $j$th subround of these $k-1$ rounds, $ALG$ can accept at
	most $B$ packets (out of $(k-1)B$ ones) because of the size constraint
	of the buffer.  For each $j \in [1,y]$, let $A_{i,j}$ denote the expected
	number of packets that $ALG$ accepts from $P(i,j,k-i)$.  By the above
	argument, we have that $\Sigma_{i} A_{i,j} \leq B$ and hence $\Sigma_{i}
	\Sigma_{j} A_{i,j} \leq yB$. 
	Let $A_{i}=\Sigma_{j}A_{i,j}$ and let $A_{z}$ be the minimum among $A_{1}, A_{2}, \cdots, A_{k-1}$ 
	(ties are broken arbitrarily). 
	Note that $A_{z} \leq \frac{yB}{k-1}$ since $\Sigma_{i} A_{i} = \Sigma_{i} \Sigma_{j} A_{i,j} \leq yB$.
	Also, note that since $A_{i}$ is an expectation, $z$ is determined only by the description of $ALG$ (and not by the actual behavior of $A$).
	
	We now explain the arrival of packets in $P(i,k)$ ($i\in [1, k-1]$).
	(See the right figure in Fig.~\ref{fig:LBR1}.)
	For $i\neq z$, all the $yB$ packets in $P(i,k)$ arrive simultaneously
	at the $(i+k-2)yB$th phase.  As for $i=z$, packets are given as a
	usual round, i.e., we have a round of $P(z,k)$ starting at
	$(z+k-2)yB$. 
	It is not hard to verify that this input is order-respecting.
	Also, 
	it can be easily verified that our adversary is oblivious because the construction of the input does not depend on the actual behavior of $ALG$. 
	Specifically, 
	$z$ depends on only the values of $A_{i,j}$ $(i \in [1, k-1], j \in [1, y])$, and 
	$\sigma$ can be constructed not with time but in advance. 
	
	First, note that $OPT$ can accept and transmit all the packets in
	$P(z,x)$ for any $x$.  Therefore, $OPT$ can complete all the $yB$
	frames in $F(z)$ and hence $V_{OPT}(\sigma) \geq yB$.  On the other
	hand, since all the packets in $P(i,k)$ ($i\neq z$) arrive
	simultaneously, $ALG$ can accept at most $B$ packets of them and hence can
	complete at most $B$ frames of $F(i)$ for each $i$.  As for $F(z)$,
	$ALG$ can complete at most $A_{z} \leq \frac{yB}{k-1}$ frames of them and
	hence ${\mathbb E}[V_{ALG}(\sigma)] \leq \frac{yB}{k-1}+(k-2)B$.  If we
	take $y \geq \frac{(k-1)^{2}(k-2)}{\epsilon}-(k-1)(k-2)$, we have that
	\begin{equation*}
		\frac{V_{OPT}(\sigma)}{{\mathbb E}[V_{ALG}(\sigma)]} \geq
		 \frac{yB}{(yB)/(k-1)+(k-2)B}
		 = k-1-\frac{(k-1)^{2}(k-2)}{y+(k-1)(k-2)}
		 \geq k-1-\epsilon. 
	\end{equation*}
	\fi
\end{proof}
%
\fi

%
\ifnum \count14 < 1
\ifnum \count12 > 0
\begin{figure*}
	 \begin{center}
	  \includegraphics[width=150mm]{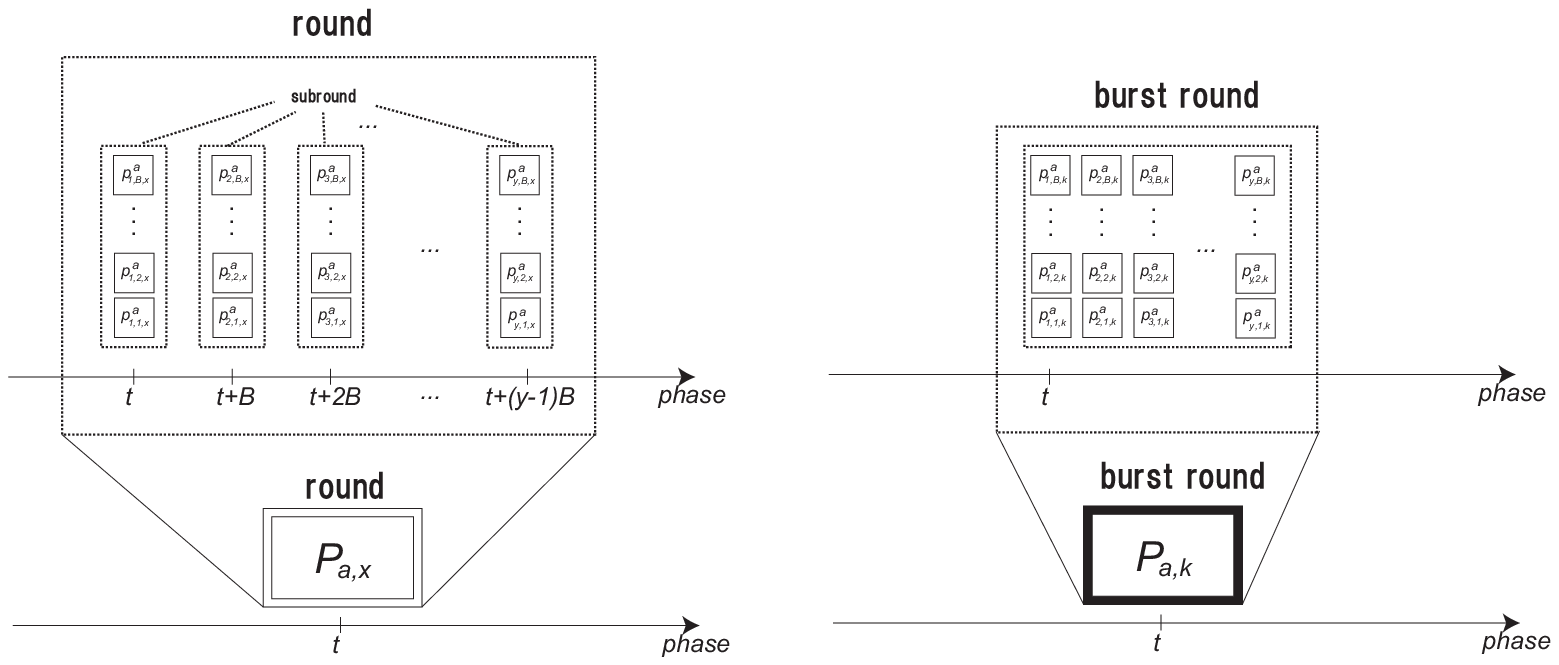}
	 \end{center}
	 \caption{
		$P(a,x)$ is written as $P_{a,x}$ in this figure. 
		The left figure shows a round of $P(a,x)$ 
		except for the case where $a \ne z$ and $x = k$. 
		On the other hand, 
		the right figure shows a round of $P(a,k)$ 
		for each $a (\in [1, k-1])$ such that $a \ne z$. 
	 }
	\label{fig:LBR1}
\end{figure*}
\begin{figure*}
	 \begin{center}
	  \includegraphics[width=150mm]{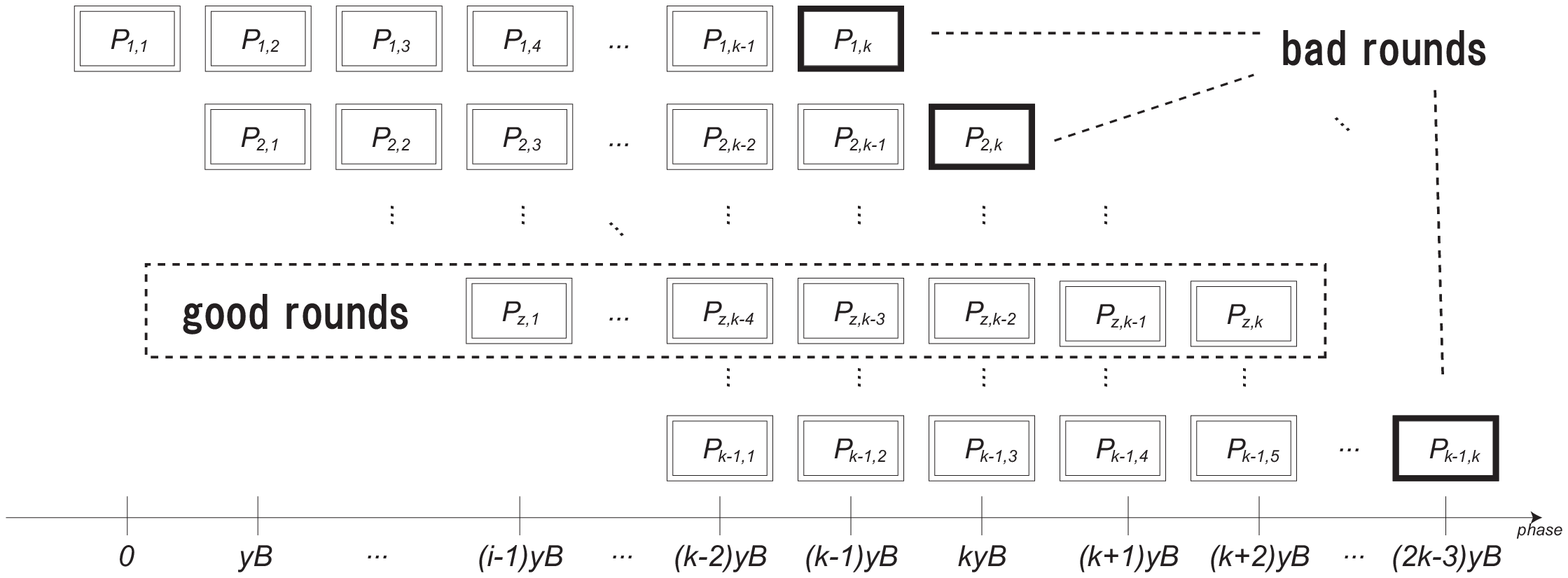}
	 \end{center}
	 \caption{
	 	Lower bound instance for randomized algorithms. 
		}
	\label{fig:LBR2}
\end{figure*}
\fi
%
\fi

\newpage
\appendix

\section{Lower bound for SP} \label{sec:ap.1}

We give an input $\sigma$ for which $SP$'s competitive ratio is as bad
as $\Omega(k^{2})$.  Let $D = 3B$ and $N=3 \times 2^{k-1}$.  $\sigma$
consists of $NB$ frames $f_1,\ldots,f_{NB}$.  For any $i (\in [1,NB])$
and any $j (\in [1,k])$, let $p_{i,j}$ denote the $j$-packet of $f_i$.
Fig.~\ref{tab:spalower} shows a pseudocode of generating $\sigma$.  Note
that in $\sigma$, all the 1-packets arrive first.  After that, all the
2-packets arrive, then all the 3-packets arrive, and so on.  An example
of the arrival sequence of $\sigma$ for $k=5$ is depicted in
Figs.~\ref{fig:spa_mf_graph1} through \ref{fig:spa_mf_graph5}, corresponding
to 1- through 5-packets, respectively.  Each figure consists of two
graphs.  An upper graph shows the arrival phase of each packet, where
the horizontal axis shows the packet indices and vertical axis shows the
phases.  For example, Fig.~\ref{fig:spa_mf_graph1} shows that 1-packets
$p_{i,1}$ ($i\in [1,B]$) arrive at the $0$th phase, indicated as (1),
1-packets $p_{i,1}$ ($i\in [B+1,2B]$) arrive at the $B$th phase,
indicated as (2), and so on.

First, consider $SP$'s behavior.  Without loss of generality, we assume
that $SP$ prioritizes frames with smaller indices, e.g., if two packets
$p_{i,j}$ and $p_{i',j}$ with $i<i'$ arrive at the same time and $SP$ is
able to accept only one packet, then $SP$ accepts $p_{i,j}$ and rejects
$p_{i',j}$.  The lower graphs of Figs.~\ref{fig:spa_mf_graph1} through
\ref{fig:spa_mf_graph5} show the behavior of $SP$ and $OPT$. 
For example, 
Fig.~\ref{fig:spa_mf_graph1} shows that $SP$ accepts 1-packets $p_{i,1}$ ($i\in [1,A]$), 
indicated as (5'), $p_{i,1}$ ($i\in [B+1, B+A]$), 
indicated as (6'), and so on. 
Now, for each $w\in [1,N]$, $SP$ accepts $A$ 1-packets $p_{(w-1)B+1,1}, p_{(w-1)B+2,1},\ldots,p_{(w-1)B+A,1}$, hence
$NA$ 1-packets in total, and rejects the rest.  Next, for each integer
$j (\in [2,k])$ and each integer $y (\in [0,2^{k-j}-1])$, $SP$ accepts
$A$ $j$-packets
$p_{y2^{j-1}D+1,j},\ldots,\allowbreak{}p_{y2^{j-1}D+A,j}$ but rejects
others.  Note that the number of $j$-packets accepted by $SP$ is
$2^{k-j}A$.  In particular, the number of $k$-packets accepted by $SP$
is $A$. Therefore, $V_{SP}(\sigma)=A$.

Next, consider $OPT$'s behavior.  Let $b_1 = 0$ and $b_z =
\sum_{j=1}^{z-1}2^{k-j-1}D$ for each integer $z (\in [2,k-1])$.  $OPT$
completes $f_{b_{z}+D+1},\ldots,f_{b_{z}+D+B}$ for each integer $z (\in
[1,k-1])$. 
Therefore, $V_{OPT}(\sigma)=(k-1)B$ and
$\frac{V_{OPT}(\sigma)}{V_{SP}(\sigma)}=\frac{(k-1)B}{A}=\Omega(k^{2})$ since $A = \lfloor B/k \rfloor$.

On the other hand, 
$MF$ completes $f_{b_{z} +1},\ldots,f_{b_{z} + \lfloor A/2 \rfloor}$
for each integer $z (\in [1,k-1])$ and
$f_{b_{k}+ \lfloor A/2 \rfloor +1},\ldots,f_{b_{k} +A}$.
Therefore,
$V_{MF}(\sigma) = (k-1) \lfloor A/2 \rfloor + A - \lfloor A/2 \rfloor \ge k \lfloor A/2 \rfloor$
and $\frac{V_{OPT}(\sigma)}{V_{MF}(\sigma)} \le \frac{(k-1)B}{k \lfloor A/2 \rfloor}
\le \frac{B}{\lfloor A/2 \rfloor}$ hold.

\begin{figure}
\begin{center}
\begin{tabularx}{\textwidth}{|X|}
\hline
$t := 0$. \\
{\bf for} $w = 1,\ldots,N$ {\bf do} \\
\hspace{3mm} 1-packets $p_{(w-1)B+1,1},\ldots,p_{wB,1}$ arrive at the $t$th phase. \\
\hspace{3mm} $t := t + B$.\\
{\bf end for} \\
{\bf for} $j = 2,\ldots,k$ {\bf do} \\
\hspace{3mm} $t := (j-1)NB$. \\
\hspace{3mm} {\bf for} $y = 0,\ldots,2^{k-j}-1$ {\bf do} \\
    \hspace{6mm} $j$-packets $p_{y2^{j-1}D+1,j},\ldots,p_{y2^{j-1}D+2^{j-2}D+D,j}$ arrive at the $t$th phase. \\
    \hspace{6mm} $t := t + B$. \\
    \hspace{6mm} {\bf for} $x = 1,\ldots,2^{j-2}-1$ {\bf do} \\
        \hspace{9mm} $j$-packets $p_{y2^{j-1}D+2^{j-2}D+xD+1,j},\ldots,p_{y2^{j-1}D+2^{j-2}D+xD+D,j}$ arrive at the $t$th phase. \\
        \hspace{9mm} $t := t + B$.\\
    \hspace{6mm} {\bf end for} \\
\hspace{3mm} {\bf end for} \\
{\bf end for} \\
\hline
\end{tabularx}
\caption{Pseudocode of arriving packets in $\sigma$ for the example of $SP$.}
\label{tab:spalower}
\end{center}
\end{figure}

\ifnum \count12 > 0
\begin{figure*}
	 \begin{center}
	  \includegraphics[width=140mm]{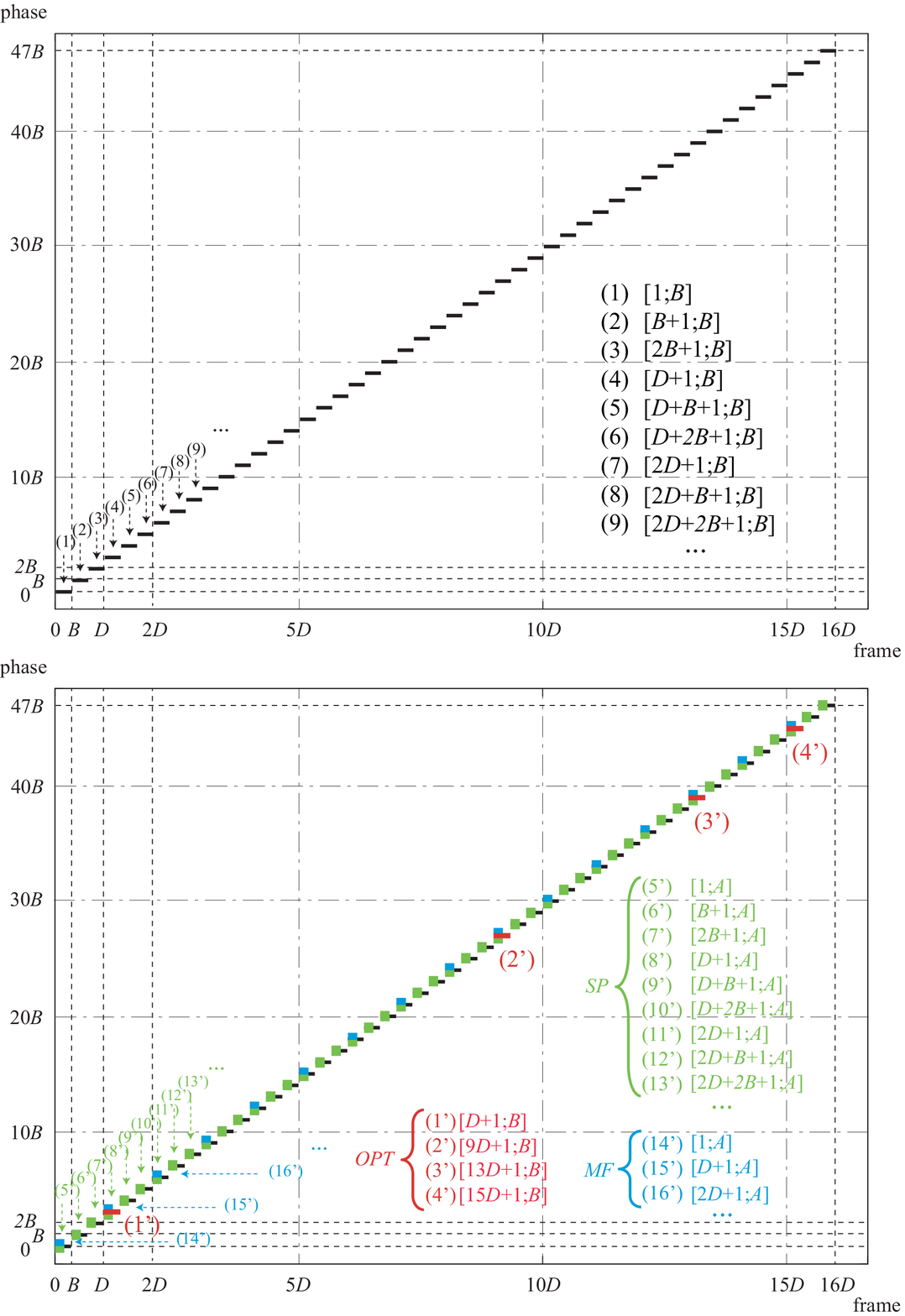}
	 \end{center}
	 \caption{Arriving 1-packets in $\sigma$ for the example of $SP$ for $k=5$. 
	 $[i; x]$ in the figure denotes all the 1-packets in $f_{i},\ldots,f_{i+x-1}$ for any integers $i$ and $x$. }
	\label{fig:spa_mf_graph1}
\end{figure*}
\fi
\ifnum \count12 > 0
\begin{figure*}
	 \begin{center}
	  \includegraphics[width=160mm]{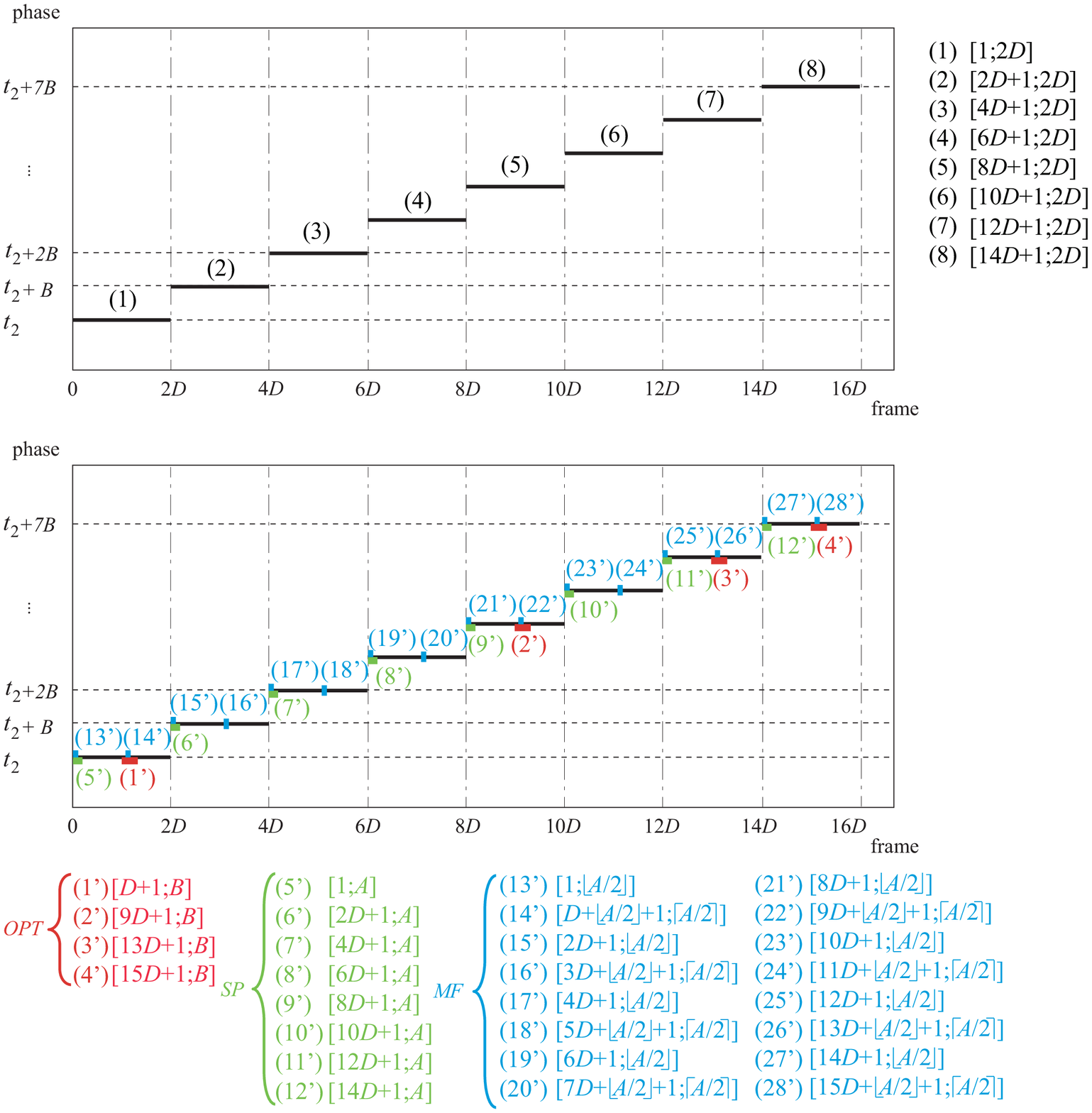}
	 \end{center}
	 \caption{Arriving 2-packets in $\sigma$ for the example of $SP$ for $k=5$.
	 }
	\label{fig:spa_mf_graph2}
\end{figure*}
\fi
\ifnum \count12 > 0
\begin{figure*}
	 \begin{center}
	  \includegraphics[width=160mm]{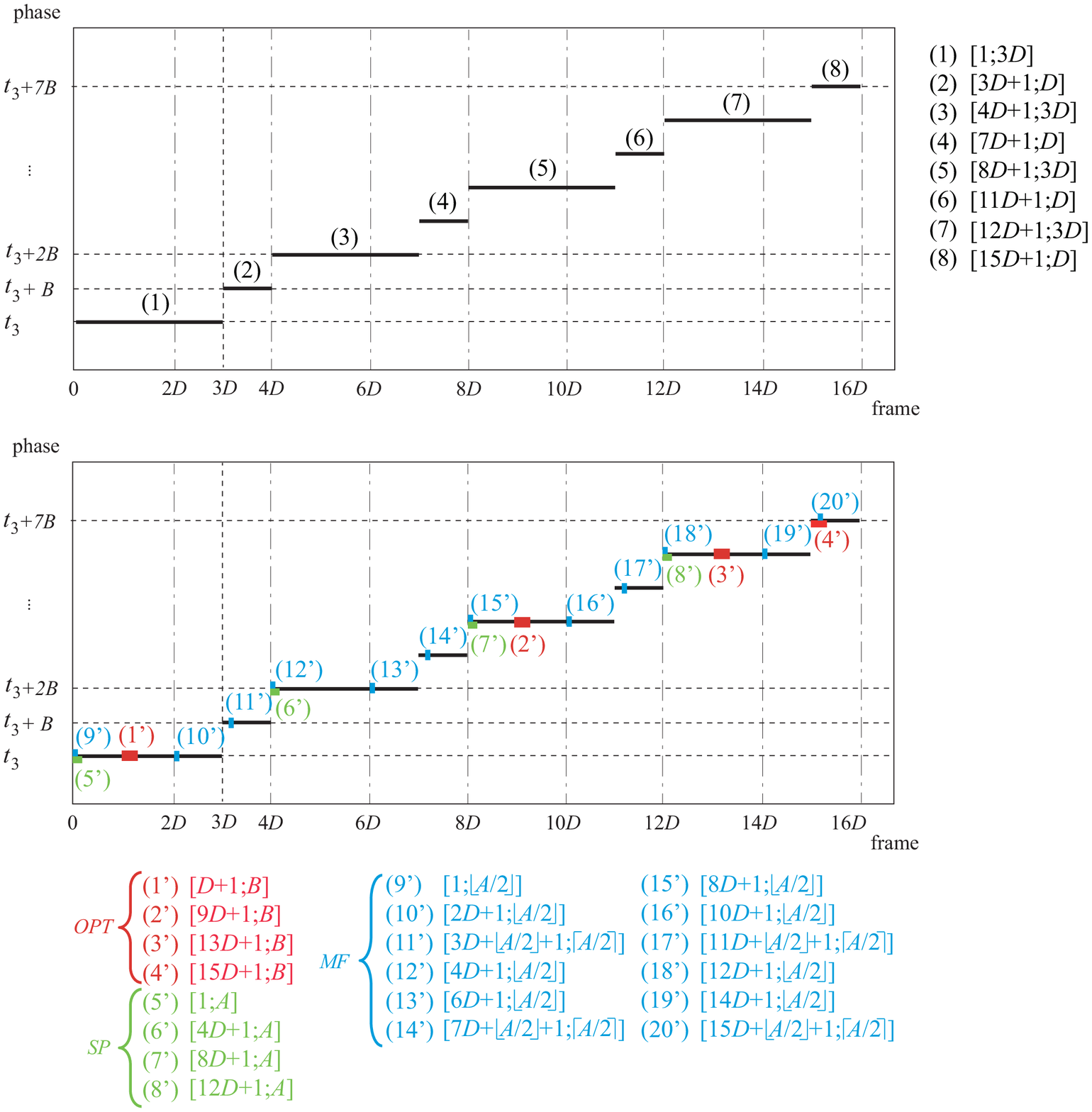}
	 \end{center}
	 \caption{Arriving 3-packets in $\sigma$ for the example of $SP$ for $k=5$.
	 }
	\label{fig:spa_mf_graph3}
\end{figure*}
\fi
\ifnum \count12 > 0
\begin{figure*}
	 \begin{center}
	  \includegraphics[width=160mm]{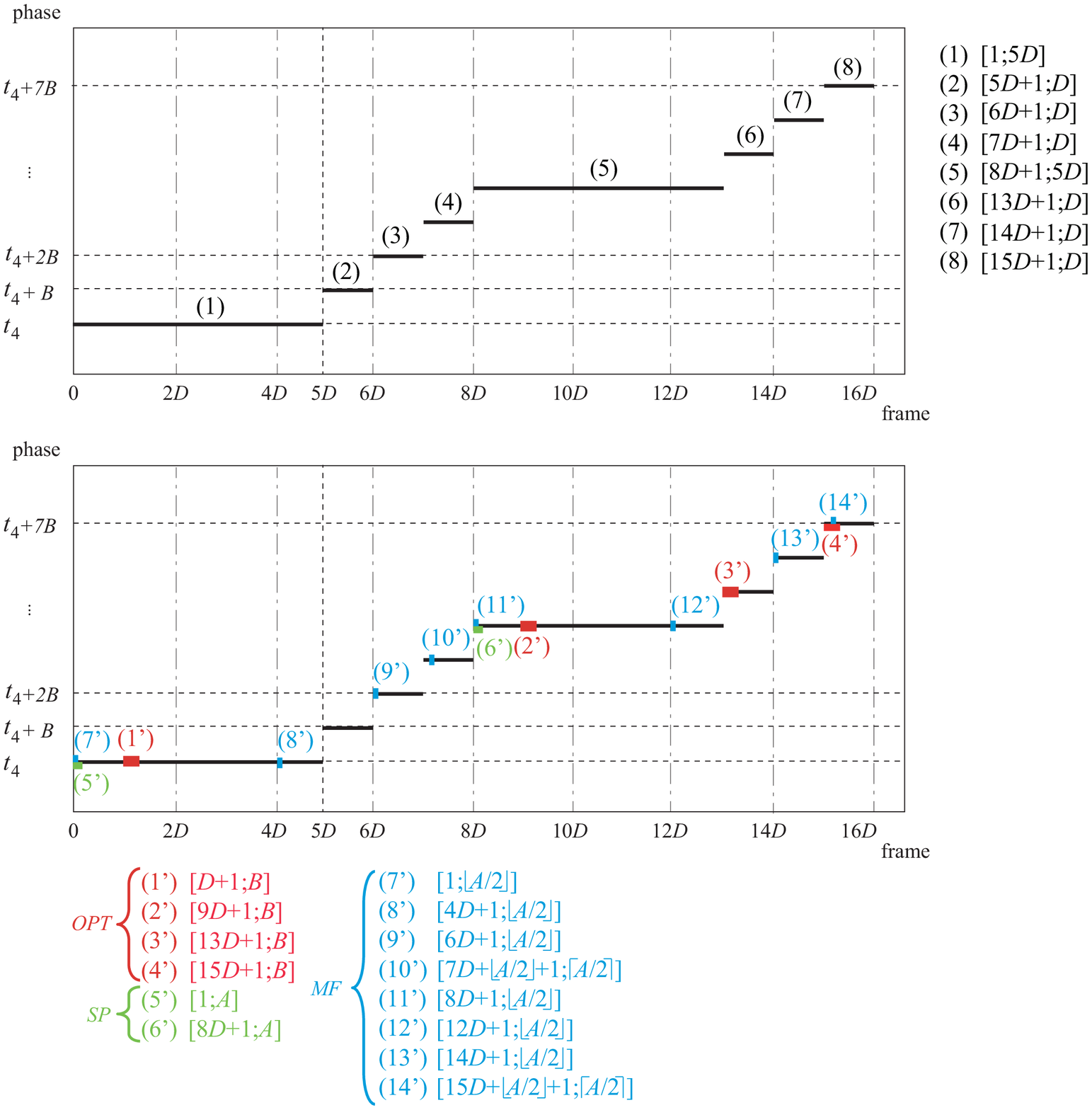}
	 \end{center}
	 \caption{Arriving 4-packets in $\sigma$ for the example of $SP$ for $k=5$.
	 }
	\label{fig:spa_mf_graph4}
\end{figure*}
\fi
\ifnum \count12 > 0
\begin{figure*}
	 \begin{center}
	  \includegraphics[width=160mm]{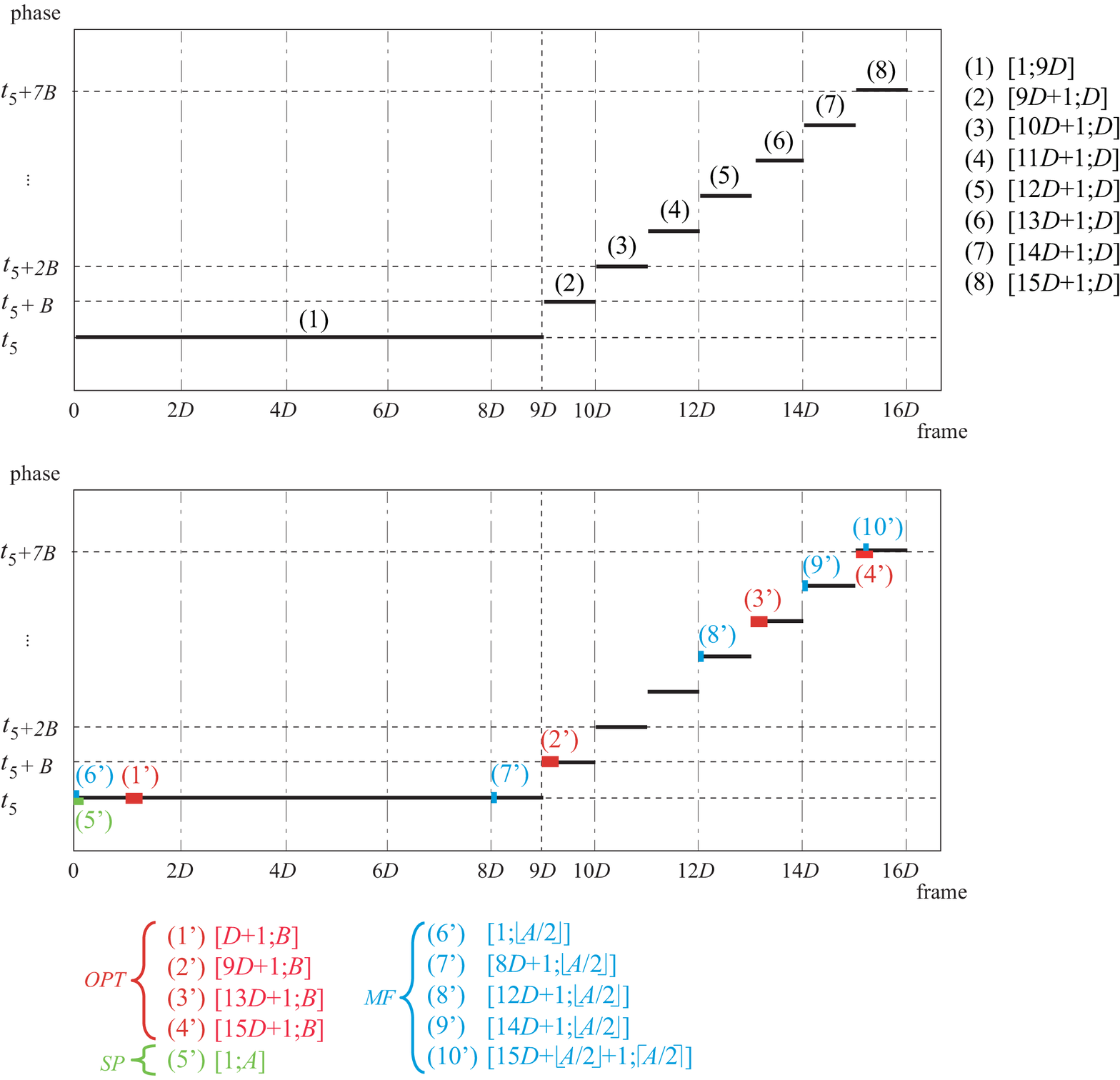}
	 \end{center}
	 \caption{Arriving 5-packets in $\sigma$ for the example of $SP$ for $k=5$.
	 }
	\label{fig:spa_mf_graph5}
\end{figure*}

\section{Execution Example of $MF$} \label{sec:ap.2}
\ifnum \count10 > 0
本節では、
表~\ref{tab:example_mf1}、\ref{tab:example_mf2}で与える入力$\sigma$によって$MF$の動作例を与える。
$k=3, B=12$ とする。
すなわち、$A=12/3=4$である。
$\sigma$ は $120$ 個のフレーム$f_1,\ldots,f_{120}$からなる。
また、任意の$i (\in [1,120])$について、
$f_i$に含まれる1-packet,2-packet,3-packetをそれぞれ$p_i, q_i, r_i$とする。
ここで、$\mbox{arr}(p_1) \le \mbox{arr}(p_2) \le \cdots \le \mbox{arr}(p_{120})$である。
1-packetは表~\ref{tab:example_mf1}で、2,3-packetは表~\ref{tab:example_mf2}で示されるように到着する。
表の各列は、順に、
パケットが到着するinteger time、到着するpacketsの名前、$GR_1$の動作（表~\ref{tab:example_mf1}のみ）、
$MF$の動作、$MF$が実行するCaseの名前、
各パケットのblock numberの値（表~\ref{tab:example_mf1}のみ）
をそれぞれ表す。
たとえば、
時刻$0$に$p_{1}, p_{2}, p_{3}, p_{4}$が到着し、
$MF$は、Case 1.2.1を実行してそれらを受理する。
（図~\ref{fig:MFex}参照。）
これらの1-packetの通し番号は1である。
とりわけ、
時刻$120$に$MF$は2-packet $q_{85}$をacceptし、
バッファ内の$q_{51}$を(middle-dropで)preemptする。
また、
時刻$120$に3-packet $r_{85}$をacceptする場合、
Case 2.2.2を実行して$r_{49}$をpreemptし、
$f_{49}$はvalidでなくなるので2-packet $q_{49}$もpreemptする。
更に、
通し番号が$2$であり、
$MF$が1-packetをacceptするフレームは、
$f_{49}, f_{50}, f_{51}, f_{52}$であるが、
event time $120_{q_{85}}$とevent time $120_{q_{86}}$に、
$q_{51}$と$q_{52}$がpreemptされている。
結果として、
event time $120_{r_{85}}$にvalidなframeの数は$\lfloor A/2 \rfloor = 2$より少なくなるので、
$MF$はCase 2.2.2.1を実行して、
$MF$のバッファ内の通し番号が2である全てのpacketをpreemptする。
\fi
\ifnum \count11 > 0
\com{（■英語）}
In this section, 
we give an execution example of $MF$ for a given input $\sigma$ in Tables~\ref{tab:example_mf1} and \ref{tab:example_mf2}. 
We suppose that $k=3$ and $B=12$, which means $A=12/3=4$. 
$\sigma$ includes $120$ frames $f_1,\ldots,f_{120}$. 
For each $i (\in [1,120])$, 
$p_i, q_i$ and $r_i$ denote the 1-packet, 2-packet and 3-packet in $f_{i}$, respectively. 
We suppose that $\mbox{arr}(p_1) \le \mbox{arr}(p_2) \le \cdots \le \mbox{arr}(p_{120})$. 
All 1-packets (all 2-packets and 3-packets) arrive as shown in Table~\ref{tab:example_mf1} (Table~\ref{tab:example_mf2}). 
Columns starting from the left in the tables present the arrival times of packets, the names of arriving packets, actions by $GR_{1}$ for arriving packets (only in Table~\ref{tab:example_mf1}), actions by $MF$ for arriving packets, the names of cases executed by $MF$ and the block numbers of arriving packets (only in Table~\ref{tab:example_mf1}). 
For example, 
1-packets $p_{1}, p_{2}, p_{3}$ and $p_{4}$ arrive at phase $0$, 
$MF$ executes Case 1.2.1, and accepts these packets. 
(See Figure~\ref{fig:MFex}.)
The block numbers of these 1-packets are 1. 
In particular, 
$MF$ accepts 2-packet $q_{85}$ at phase $120$, and 
preempts $q_{51}$ that is stored in its buffer at $t_{q_{85}}-$. 
(That is, $MF$ discards $q_{51}$ using a ``middle-drop'' policy.)
Moreover, 
when $MF$ accepts 3-packet $r_{85}$ at the $120$th phase, 
$MF$ executes Case 2.2.2, and preempts $r_{49}$. 
Hence, $f_{49}$ becomes invalid for $MF$. 
At this time, $MF$ preempts 2-packet $q_{49}$ as well. 
In addition, 
the frames $f$ such that the 1-packets in $f$ are accepted by $MF$, and $g(f) = 2$ are $f_{49}, f_{50}, f_{51}$ and $f_{52}$. 
At event times $t_{q_{85}}$ and $t_{q_{86}}$, 
$q_{51}$ and $q_{52}$ are preempted, respectively. 
That is, 
the number of valid frames of $MF$ with block number 2 decreases to less than $\lfloor A/2 \rfloor = 2$ at event time $t_{r_{85}}$. 
Thus, 
$MF$ further executes Case 2.2.2.1, and preempts all the packets whose block numbers are 2 in its buffer. 
\fi
\ifnum \count12 > 0
\begin{figure*}
	 \begin{center}
	  \includegraphics[width=150mm]{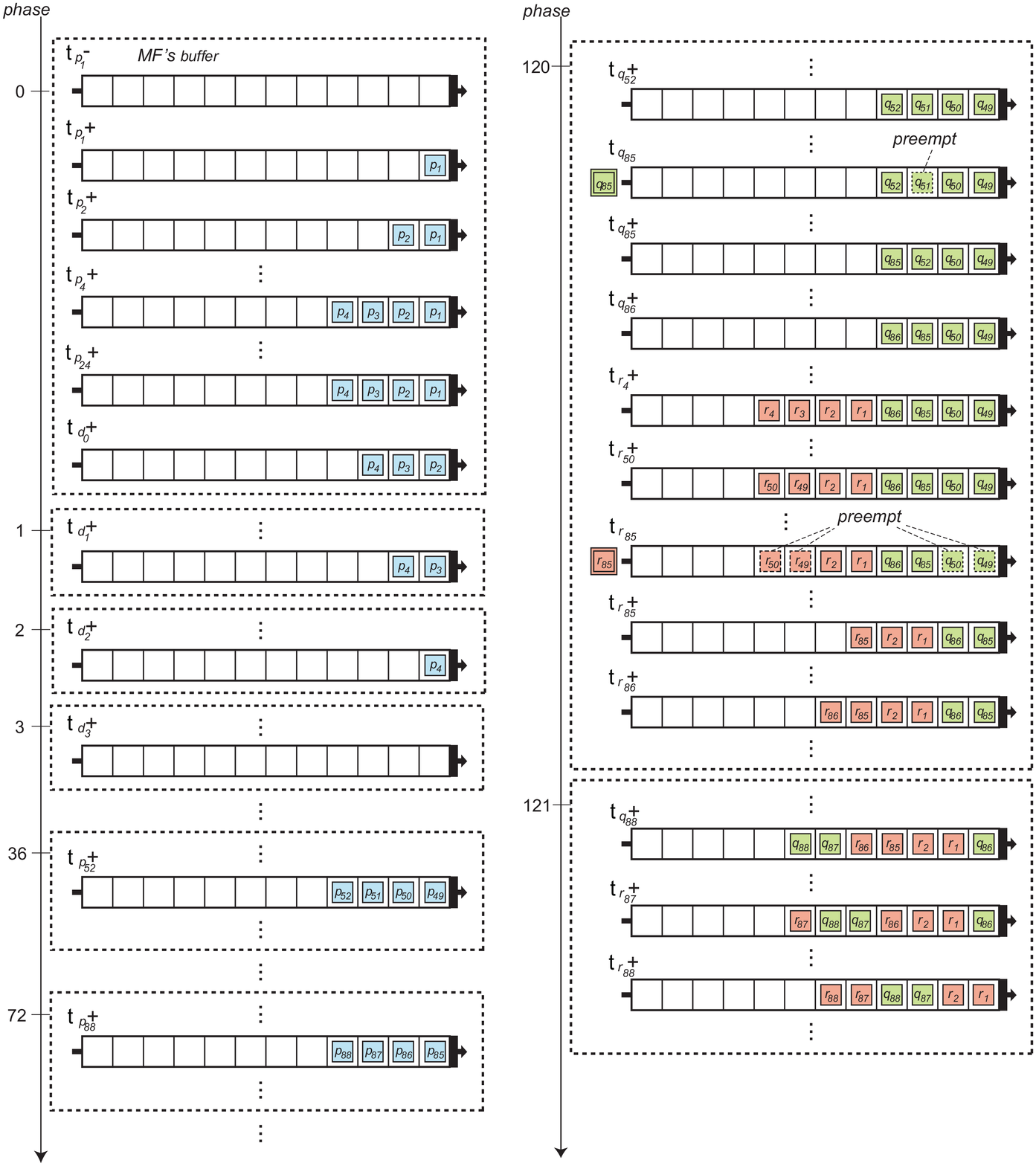}
	 \end{center}
	 \caption{
	 	Execution example of $MF$. 
		$t_{d_{i}}$ denotes the $i$th delivery time. 
	 	}
	\label{fig:MFex}
\end{figure*}
\fi

\begin{table}
\caption{Arriving 1-packets in $\sigma$}
\label{tab:example_mf1}
\begin{center}
\small
\begin{tabular}[t]{|c|c|c|c|c|c|}
\hline
     & Arrival & $GR_1$'s  & $MF$'s &      & \\
Time & Packets & Action    & Action & Case & Block \\
\hline
\multirow{3}*{0}  & $p_{1},\ldots,p_{4}$ & accept & accept & 1.2.1 & 1 \\ \cline{2-6}
                  & $p_{5},\ldots,p_{12}$ & accept & reject & 1.2.2 & 1 \\ \cline{2-6}
                  & $p_{13},\ldots,p_{24}$ & reject & reject & 1.1 & 1  \\
\hline
12                & $p_{25},\ldots,p_{36}$ & accept & reject & 1.2.2 & 1 \\
\hline
\multirow{2}*{24} & $p_{37},\ldots,p_{47}$ & accept & reject & 1.2.2 & 1 \\ \cline{2-6}
                  & $p_{48}$ & accept & reject & 1.2.3 & 1 \\
\hline
\multirow{2}*{36} & $p_{49},\ldots,p_{52}$ & accept & accept & 1.2.1 & 2  \\ \cline{2-6}
                  & $p_{53},\ldots,p_{60}$ & accept & reject & 1.2.2 & 2 \\ \cline{2-6}
\hline
48                & $p_{61},\ldots,p_{72}$ & accept & reject & 1.2.2 & 2 \\
\hline
\multirow{2}*{60} & $p_{73},\ldots,p_{83}$ & accept & reject & 1.2.2 & 2 \\ \cline{2-6}
                  & $p_{84}$ & accept & reject & 1.2.3 & 2 \\
\hline
\multirow{2}*{72} & $p_{85},\ldots,p_{88}$ & accept & accept & 1.2.1 & 3  \\ \cline{2-6}
                  & $p_{89},\ldots,p_{96}$ & accept & reject & 1.2.2 & 3 \\ \cline{2-6}
\hline
84                & $p_{97},\ldots,p_{108}$ & accept & reject & 1.2.2 & 3 \\
\hline
\multirow{2}*{96} & $p_{109},\ldots,p_{119}$ & accept & reject & 1.2.2 & 3 \\ \cline{2-6}
                  & $p_{120}$ & accept & reject & 1.2.3 & 3 \\
\hline
\end{tabular}
\end{center}
\end{table}

\begin{table}
\caption{Arriving 2-packets and 3-packets in $\sigma$}
\label{tab:example_mf2}
\begin{center}
\small
\begin{tabular}[t]{|c|c|c|c|}
\hline
     & Arrival & $MF$'s &      \\
Time & Packets & Action & Case \\
\hline
\multirow{2}*{108} & $q_{1},\ldots,q_{4}$ & accept & 2.2.1 \\ \cline{2-4}
                  & $q_{5},\ldots,q_{48}$ & reject & 2.1 \\
\hline
\multirow{17}*{120} & $q_{49},\ldots,q_{52}$ & accept & 2.2.1 \\ \cline{2-4}
                  & $q_{53},\ldots,q_{84}$ & reject & 2.1 \\ \cline{2-4}
                  & \multirow{2}*{$q_{85}$} & preempt $q_{51}$ & \multirow{2}*{2.2.2} \\
                  &          & accept $q_{85}$ &       \\ \cline{2-4}
                  & \multirow{2}*{$q_{86}$} & preempt $q_{52}$ & \multirow{2}*{2.2.2} \\
                  &          & accept $q_{86}$ &       \\ \cline{2-4}
                  & $r_{1},\ldots,r_{4}$ & accept & 2.2.1 \\ \cline{2-4}
                  & $r_{5},\ldots,r_{48}$ & reject & 2.1 \\ \cline{2-4}
                  & \multirow{2}*{$r_{49}$} & preempt $r_{3}$ & \multirow{2}*{2.2.2} \\
                  &          & accept $r_{49}$ &       \\ \cline{2-4}
                  & \multirow{2}*{$r_{50}$} & preempt $r_{4}$ & \multirow{2}*{2.2.2} \\
                  &          & accept $r_{50}$ &       \\ \cline{2-4}
                  & $r_{51},\ldots,r_{84}$ & reject & 2.1 \\ \cline{2-4}
                  & \multirow{3}*{$r_{85}$} & preempt $r_{49},q_{49}$ & \multirow{2}*{2.2.2} \\
                  &          & accept $r_{85}$ &       \\
                  &                         & preempt $r_{50},q_{50}$ & 2.2.2.1 \\ \cline{2-4}
                  & $r_{86}$ & accept & 2.2.1 \\ \cline{2-4}
\hline
\multirow{7}*{121} & $q_{87},q_{88}$ & accept & 2.2.1 \\ \cline{2-4}
                  & $q_{89},\ldots,q_{120}$ & reject & 2.1 \\ \cline{2-4}
                  & \multirow{2}*{$r_{87}$} & preempt $r_{85}$ & \multirow{2}*{2.2.2} \\
                  &          & accept $r_{87}$ &       \\ \cline{2-4}
                  & \multirow{2}*{$r_{88}$} & preempt $q_{86},r_{86}$ & \multirow{2}*{2.2.2} \\
                  &                         & accept $r_{88}$ &       \\ \cline{2-4}
                  & $r_{89},\ldots,r_{120}$ & reject & 2.1 \\
\hline
\end{tabular}
\end{center}
\end{table}

\ifnum \count14 > 0
\section{Proofs of Lemmas and Theorems} \label{sec:ap.3}
\subsection{Proofs of Lemmas}
\subsection{Proof of Theorem~\ref{thm:2}}
%

	%
	\ifnum \count10 > 0
	オンラインアルゴリズム$ALG$を固定して考える。
	次の様な入力$\sigma$を考える。
	phase $0$に$2B$個の1-パケットが到着する。
	このとき、$ALG$は$x (\leq B)$個のパケットを受理する。
	一方で、$OPT$は$ALG$が受理しないパケットを$B$個受理する。
	$ALG$が受理する$x$個のパケットの集合を$C$と呼び、
	$OPT$が受理する$B$個のパケットの集合を$D$と呼ぶ。
	（図~\ref{fig:LBCR2km1}参照。）
	$B$回の送信サブフェイズの後、
	$B + \lfloor \frac{B}{k-1} \rfloor$個の1-パケットが到着する。
	このとき、$ALG$は$y (\leq B)$個のパケットを受理する。
	一方で、$OPT$は$ALG$が受理しないパケットを$\lfloor \frac{B}{k-1} \rfloor$個だけ受理する。
	$ALG$が受理する$y$個のパケットの集合を$E$と呼び、
	$OPT$が受理する$\lfloor \frac{B}{k-1} \rfloor$個のパケットの集合を$F$と呼ぶ。
	$B$回の送信サブフェイズの後、
	phase $2B$に$2B$個の1-パケットが到着する。
	このとき、$ALG$は$z (\leq B)$個のパケットを受理する。
	一方で、$OPT$は$ALG$が受理しないパケットを$B$個受理する。
	$ALG$が受理する$z$個のパケットの集合を$G$と呼び、
	$OPT$が受理する$B$個のパケットの集合を$H$と呼ぶ。
	以上で、入力中の全ての1-packetは到着した。
	時刻$2B$より後は、各$j( \geq 2)$-packetが到着する。
	各$j = 2, ..., k$に対して、
	phase $3B + (j - 2)B = (j+1)B$に、
	パケット集合$D$の$B$個の$j$-パケットが到着し、
	$OPT$はそれらを全て受理し送信する。
	一方、
	phase $(k + 2)B$に、
	パケット集合$C, E, F, G$の2-パケットから$k$-パケット全てが一度に到着する。
	このとき、
	$OPT$は、$F$に対応する$\lfloor \frac{B}{k-1} \rfloor (k-1) \leq B$個のpacketをacceptする。
	一方、
	$ALG$がaccpetできるそれらのpacketは高々$B$個であるので、
	$ALG$がcompleteできるフレームの数は高々$\lfloor \frac{B}{k-1} \rfloor$個である。
	その到着フェイズの後、
	パケット集合$H$の2-パケットから$k$-パケットが到着し、
	$OPT$はそれらを全て受理し送信する。
	以上より、
	$V_{ALG}(\sigma) \leq \lfloor \frac{B}{k-1} \rfloor$
	と
	$V_{OPT}(\sigma) = 2B + \lfloor \frac{B}{k-1} \rfloor$
	が成立する。
	よって、
	$B \geq k-1$ならば、
	$\frac{V_{OPT}(\sigma)}{V_{ALG}(\sigma)} 
		\geq \frac{2B + \lfloor \frac{B}{k-1} \rfloor}{\lfloor \frac{B}{k-1} \rfloor} = \frac{2B}{\lfloor \frac{B}{k-1} \rfloor} + 1$が成立する。
	また、
	$B \leq k-2$ならば、
	競合比は発散する。
	\fi
	\ifnum \count11 > 0
	\com{（■英語）}
	Fix an online algorithm $ALG$. 
	Let us consider the following input $\sigma$. 
	(See Figure~\ref{fig:LBCR2km1}.) 
	At the 0th phase, 
	$2B$ 1-packets arrive. 
	$ALG$ accepts at most $B$ 1-packets, 
	and $OPT$ accepts $B$ 1-packets that are not accepted by $ALG$. 
	Let $C$ ($D$, respectively) be the set of the 1-packets accepted by $ALG$ ($OPT$, respectively). 
	At the $i$th phase ($i \in [1,B-1]$), 
	no packets arrive. 
	Hence, just after the $(B-1)$st phase, 
	both $ALG$'s and $OPT$'s queues are empty 
	(since $B$ delivery subphases occur).
	At the $B$th phase, 
	$B + \lfloor \frac{B}{k-1} \rfloor$ 1-packets arrive in the same manner as the first $2B$ 1-packets. 
	$ALG$ can accept at most $B$ 1-packets, and 
	$OPT$ accepts $\lfloor \frac{B}{k-1} \rfloor$ 1-packets that are not accepted by $ALG$. 
	Let $E$ ($F$, respectively) be the set of the packets accepted by $ALG$ ($OPT$, respectively). 
	At the $i$th phase ($i \in [B+1,2B-1]$), 
	no packets arrive, 
	and both $ALG$'s and $OPT$'s queues are empty
	just after the $(2B-1)$st phase. 
	Once again
	at the $2B$th phase, $2B$ 1-packets arrive. 
	$ALG$ accepts at most $B$ 1-packets, 
	and $OPT$ accepts $B$ 1-packets that are not accepted by $ALG$. 
	Let $G$ ($H$, respectively) be the set of the 1-packets accepted by $ALG$ ($OPT$, respectively). 
	This is the end of the arrivals and deliveries of 1-packets.
	At the $i$th phase ($i \in [2B+1,3B-1]$), 
	no packets arrive, 
	and hence just before the $3B$th phase, 
	both $ALG$'s and $OPT$'s queues are empty. 
	For each $j = 2, ..., k$, 
	the $B$ $j$-packets corresponding to 1-packets in $D$ arrive at the $(j+1)B$th phase. 
	$OPT$ accepts and transmits them.
	(There is no incentive for $ALG$ to accept them.)
	Next, all the packets corresponding to all the 1-packets in $C \cup E \cup F \cup G$ arrive at the $(k+2)B$th phase. 
	Since $ALG$ needs to accept all the $k-1$ packets of the same frame to complete it, 
	the number of frames $ALG$ can complete is at most $\lfloor \frac{B}{k-1} \rfloor$. 
	$OPT$ accepts all the $\lfloor \frac{B}{k-1} \rfloor (k-1)$ packets corresponding to all the 1-packets in $F$. 
	Note that this is possible because $\lfloor \frac{B}{k-1} \rfloor (k-1) \leq B$. 
	Hence, $OPT$ completes all the $\lfloor \frac{B}{k-1} \rfloor$ frames of $F$. 
	After which 
	all the packets corresponding to 1-packets in $H$ arrive one after the other, and 
	$OPT$ can accept and transmit them. 
	Note that the input sequence is order-respecting.
	By the above argument, 
	we have $V_{ALG}(\sigma) \leq \lfloor \frac{B}{k-1} \rfloor$
	and 
	$V_{OPT}(\sigma) = 2B + \lfloor \frac{B}{k-1} \rfloor$. 
	Therefore, 
	if $B \geq k-1$, 
	$\frac{V_{OPT}(\sigma)}{V_{ALG}(\sigma)} \geq \frac{2B}{\lfloor \frac{B}{k-1} \rfloor} + 1$. 
	If $B \leq k-2$, 
	the competitive ratio of $ALG$ is unbounded. 
	\fi
	%
%
\ifnum \count12 > 0
\begin{figure*}
	 \begin{center}
	  \includegraphics[width=150mm]{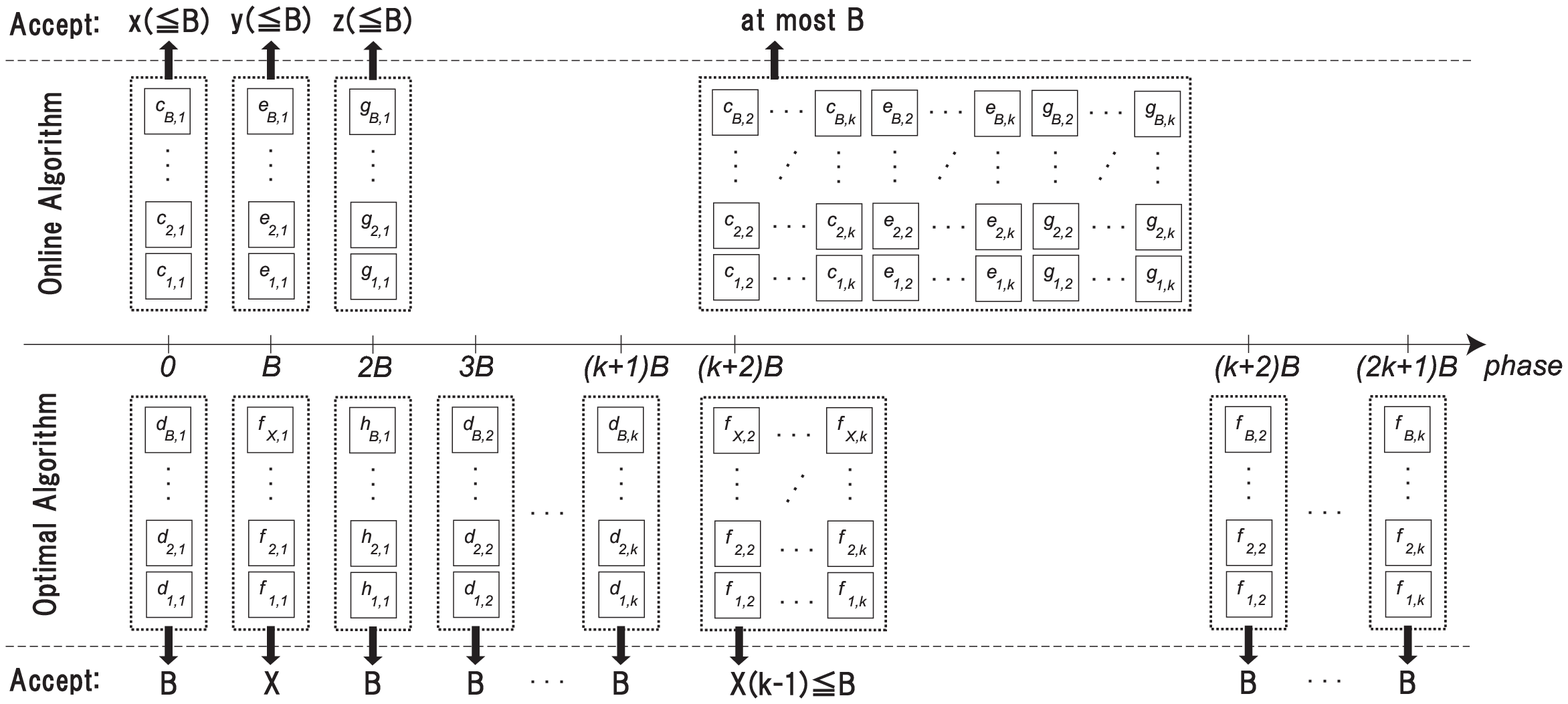}
	 \end{center}
	 \caption{Lower Bound Instance. 
	 	Each square denotes an arriving packet accepted by an online algorithm or $OPT$. 
		In the figure $X = \lfloor \frac{B}{k-1} \rfloor$. 
	 	}
	\label{fig:LBCR2km1}
\end{figure*}
\fi
\subsection{Proof of Theorem~\ref{thm:4}}
%

	%
	\ifnum \count10 > 0
	まず、本証明を概説する。
	オンライン確率アルゴリズムには$k-1$通りの選択肢をもつ入力が与えられる。
	$k-1$通りの選択肢の中には、1つだけ多くのframeをcomplete出来る選択肢がある（called good）。
	のこりの$k-2$の選択肢は、どんなに頑張ってもそれと比べてごくわずかのframeしかcomplete出来ない。
	オンラインアルゴリズムには選択を終えた後に、その結果が提示される。
	オンラインアルゴリズムは、もちろん、good選択肢を知らないが、
	$OPT$はその選択肢をしっている。
	例えば、1番目の選択肢がgoodだとしよう。
	オンラインはたまたま1番目の選択肢を重視しているかもしれない。
	そこで、
	更に、
	$k-1$個の全く同じ選択肢をもち、当たりの選択肢だけが異なる$k-1$通りの入力を考えよう。
	オンラインアルゴリズムが選択を行う前の段階では、
	$k-1$個の入力は全て同じ入力にみえる様になっている。
	すなわち、
	1番目の選択肢を重視するアルゴリズムは、
	残りの$k-2$個の入力において大損する可能性があるのである。
	用語を定義する。
	任意のevent time $t$、任意の整数$x \in [1, k]$に対して、
	に$x$-パケットが$B$個到着し、
	そのあと、$t$から$t + B-1$まで$B$回送信フェイズが実行される。
	この部分入力列を{\em $\boldsymbol{x}$-サブラウンド}と呼ぶ。
	十分大きな整数$y$に対して、
	$x$-サブラウンドが$y$回繰り返された場合、
	その部分入力列を{\em $\boldsymbol{x}$-ラウンド}と呼ぶ。
	（図~\ref{fig:LBR1}参照。）
	$1$-ラウンドから$k$-ラウンドまで連続した部分入力列を、
	{\em 良いラウンド}と呼ぶ。
	$1$-ラウンドから$k-1$-ラウンドまで連続した後に、
	それらの$By$個の$k$-パケットが全て一度に到着する入力を、
	{\em 悪いラウンド}と呼ぶ。
	それらを踏まえて、
	次の様$\sigma$な入力を考える。
	$\sigma$は、$k-1$個のラウンド：
	$1$個の良いラウンドと$k-2$個の悪いラウンドから構成されている。
	具体的には、$i \in [1, k-2]$に対して、
	$i+1$番目のラウンドは、$i$番目のラウンドの$1$-サブラウンドが終了した直後に開始される。
	よって、
	ある時刻に
	$1$番目のラウンドの$k-1$-サブラウンド、
	$2$番目のラウンドの$k-2$-サブラウンド、
	$3$番目のラウンドの$k-3$-サブラウンド、
	$\cdots$、
	$k-1$番目のラウンドの$1$-サブラウンドが同時に開始される。
	（図~\ref{fig:LBR2}参照。）
	このとき、
	オンラインアルゴリズム$ALG$はどのラウンドが良いフレームなのか分からないので、
	良いラウンドの$yB$個のフレームのうち、
	高々$yB/(k-1)$個しか完成させることが出来ない。
	一方で、$OPT$は良いラウンドの全てのフレームを構成するパケットを受理することが出来る。
	よって、
	$\frac{V_{OPT}(\sigma)}{{\mathbb E}[V_{ALG}(\sigma)]} \geq \frac{yB}{yB / (k-1)} = k - 1$が成立する。
	\fi
	\ifnum \count11 > 0
	\com{（■英語）}
	Fix an arbitrary randomized online algorithm $ALG$.  Let $y$ be a large
	integer that will be fixed later.  Our adversarial input $\sigma$
	consists of $(k-1)yB$ frames.  These frames are divided into $k-1$
	groups each with $yB$ frames.  Also, frames of each group is divided
	into $y$ subgroups each with $B$ frames.  For each $i (\in [1,k-1])$ and
	$j (\in [1,y])$, let $F(i,j)$ be the set of frames in the $j$th subgroup
	of the $i$th group and let $F(i)=\cup_{j} F(i,j)$.  For each $x (\in
	[1,k])$, let $P(i,j,x)$ be the set of $x$-packets of the frames in
	$F(i,j)$ and let $P(i,x)=\cup_{j} P(i,j,x)$.
	
	We first give a very rough idea of how to construct the adversary.
	Among the $k-1$ groups defined above, one of them is a good group. 
	In the first half of the input (from phase $0$ to phase $(k-1)yB-1$), 
	the adversary gives packets to the online algorithm in such a way that the algorithm cannot distinguish the good group. 
	Also, since the buffer size is bounded,
	the algorithm must give up many frames during the first half; 
	only $yB$ frames can survive at the end of the first half. 
	In the second half of the input, 
	remaining packets are given in such a way that $k$-packets from the bad groups arrive at a burst, 
	while $k$-packets from the good group arrive one by one. 
	Hence, 
	if the algorithm is lucky enough to keep many packets of the good group (say, Group 1) at the end of the first half, then it can complete many frames eventually. 
	However, 
	such an algorithm behaves very poorly for an input in which Group 1 is bad. 
	Therefore, the best strategy of an online algorithm (even randomized one) is to keep equal number of frames from each group during the first half. 
	
	Before showing our adversarial input, we define a subsequence of an input.
	For any $t$, suppose that $B$
	packets of $P(i,j,x)$ arrive at the $t$th phase and no packets arrive
	during $t+1$ through $(t+B-1)$st phases. 
	Let us call this subsequence a {\em subround of $P(i,j,x)$ starting at the $t$th phase}. 
	Notice that if we focus on
	a single subround, an algorithm can accept and transmit all the packets
	of $P(i,j,x)$ by the end of the subround.  A {\em round of $P(i,x)$
	starting at the $t$th phase} is a concatenation of $y$ subrounds of $P(i,j,x)$ ($j\in [1, y]$), 
	where each subround of $P(i,j,x)$ starts at the $(t+(j-1)B)$th phase.
	(See the left figure in Fig.~\ref{fig:LBR1}.)
	
	Our input consists of rounds of $P(i,x)$ starting at the $(i+x-2)yB$th phase, 
	for $i\in [1, k-1]$ and $x\in [1, k-1]$. 
	(See Fig.~\ref{fig:LBR2}.)
	Note that any two rounds $P(i,x)$ and $P(i',x')$ start simultaneously if $i + x = i' + x'$. 
	Currently, 
	the specification of the arrival of packets in $P(i,x)$ for $x=k$ is missing. 
	This is the key for the construction of our adversary and will be explained shortly.
	
	Consider $k-1$ rounds (of $P(1,k-1), P(2,k-2), \cdots, P(k-1,1)$) starting at the $(k-2)yB$th phase, 
	which occur simultaneously.  Note that for each
	$j$, at the $j$th subround of these $k-1$ rounds, $ALG$ can accept at
	most $B$ packets (out of $(k-1)B$ ones) because of the size constraint
	of the buffer.  For each $j \in [1,y]$, let $A_{i,j}$ denote the expected
	number of packets that $ALG$ accepts from $P(i,j,k-i)$.  By the above
	argument, we have that $\Sigma_{i} A_{i,j} \leq B$ and hence $\Sigma_{i}
	\Sigma_{j} A_{i,j} \leq yB$. 
	Let $A_{i}=\Sigma_{j}A_{i,j}$ and let $A_{z}$ be the minimum among $A_{1}, A_{2}, \cdots, A_{k-1}$ 
	(ties are broken arbitrarily). 
	Note that $A_{z} \leq \frac{yB}{k-1}$ since $\Sigma_{i} A_{i} = \Sigma_{i} \Sigma_{j} A_{i,j} \leq yB$.
	Also, note that since $A_{i}$ is an expectation, $z$ is determined only by the description of $ALG$ (and not by the actual behavior of $A$).
	
	We now explain the arrival of packets in $P(i,k)$ ($i\in [1, k-1]$).
	(See the right figure in Fig.~\ref{fig:LBR1}.)
	For $i\neq z$, all the $yB$ packets in $P(i,k)$ arrive simultaneously
	at the $(i+k-2)yB$th phase.  As for $i=z$, packets are given as a
	usual round, i.e., we have a round of $P(z,k)$ starting at
	$(z+k-2)yB$. 
	It is not hard to verify that this input is order-respecting.
	Also, 
	it can be easily verified that our adversary is oblivious because the construction of the input does not depend on the actual behavior of $ALG$. 
	Specifically, 
	$z$ depends on only the values of $A_{i,j}$ $(i \in [1, k-1], j \in [1, y])$, and 
	$\sigma$ can be constructed not with time but in advance. 
	
	First, note that $OPT$ can accept and transmit all the packets in
	$P(z,x)$ for any $x$.  Therefore, $OPT$ can complete all the $yB$
	frames in $F(z)$ and hence $V_{OPT}(\sigma) \geq yB$.  On the other
	hand, since all the packets in $P(i,k)$ ($i\neq z$) arrive
	simultaneously, $ALG$ can accept at most $B$ packets of them and hence can
	complete at most $B$ frames of $F(i)$ for each $i$.  As for $F(z)$,
	$ALG$ can complete at most $A_{z} \leq \frac{yB}{k-1}$ frames of them and
	hence ${\mathbb E}[V_{ALG}(\sigma)] \leq \frac{yB}{k-1}+(k-2)B$.  If we
	take $y \geq \frac{(k-1)^{2}(k-2)}{\epsilon}-(k-1)(k-2)$, we have that
	\begin{equation*}
		\frac{V_{OPT}(\sigma)}{{\mathbb E}[V_{ALG}(\sigma)]} \geq
		 \frac{yB}{(yB)/(k-1)+(k-2)B}
		 = k-1-\frac{(k-1)^{2}(k-2)}{y+(k-1)(k-2)}
		 \geq k-1-\epsilon. 
	\end{equation*}
	\fi
	%
%
%
\ifnum \count12 > 0
\begin{figure*}
	 \begin{center}
	  \includegraphics[width=150mm]{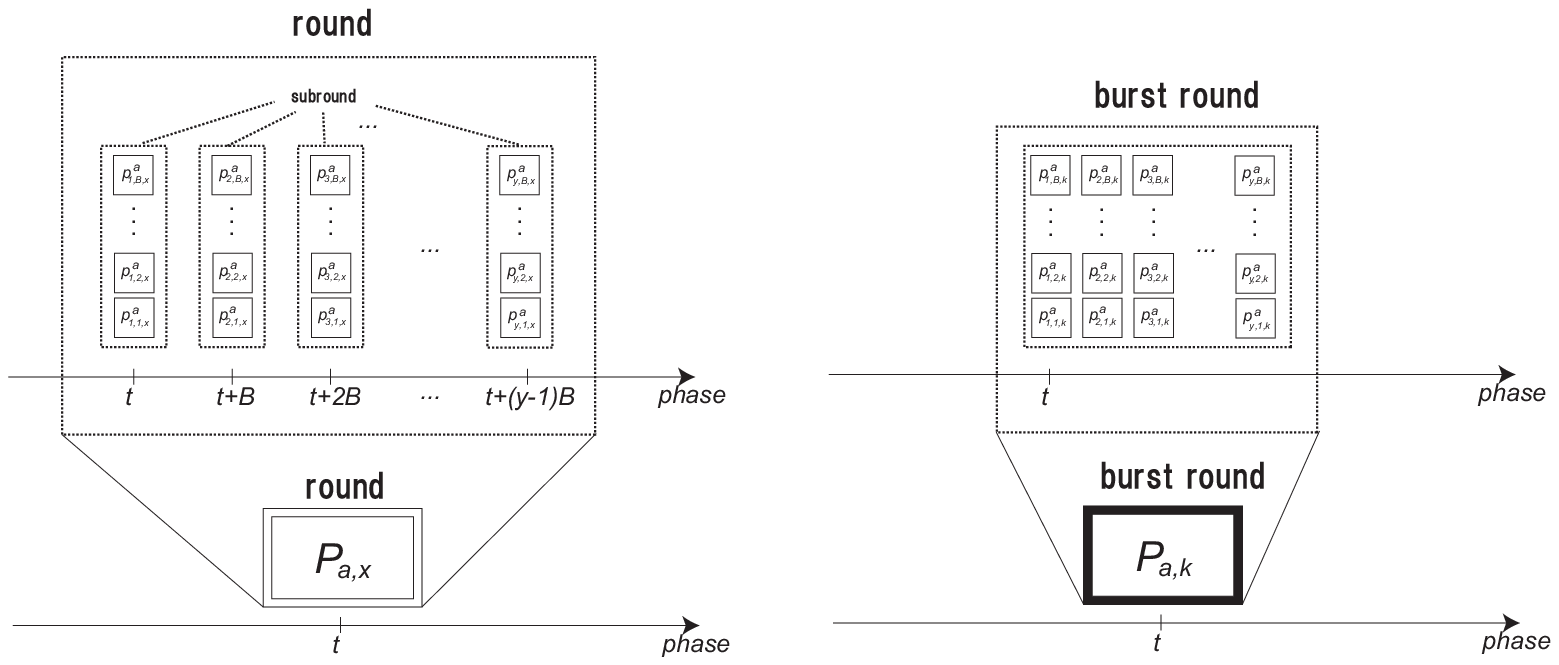}
	 \end{center}
	 \caption{
		$P(a,x)$ is written as $P_{a,x}$ in this figure. 
		The left figure shows a round of $P(a,x)$ 
		except for the case where $a \ne z$ and $x = k$. 
		On the other hand, 
		the right figure shows a round of $P(a,k)$ 
		for each $a (\in [1, k-1])$ such that $a \ne z$. 
	 }
	\label{fig:LBR1}
\end{figure*}
\begin{figure*}
	 \begin{center}
	  \includegraphics[width=150mm]{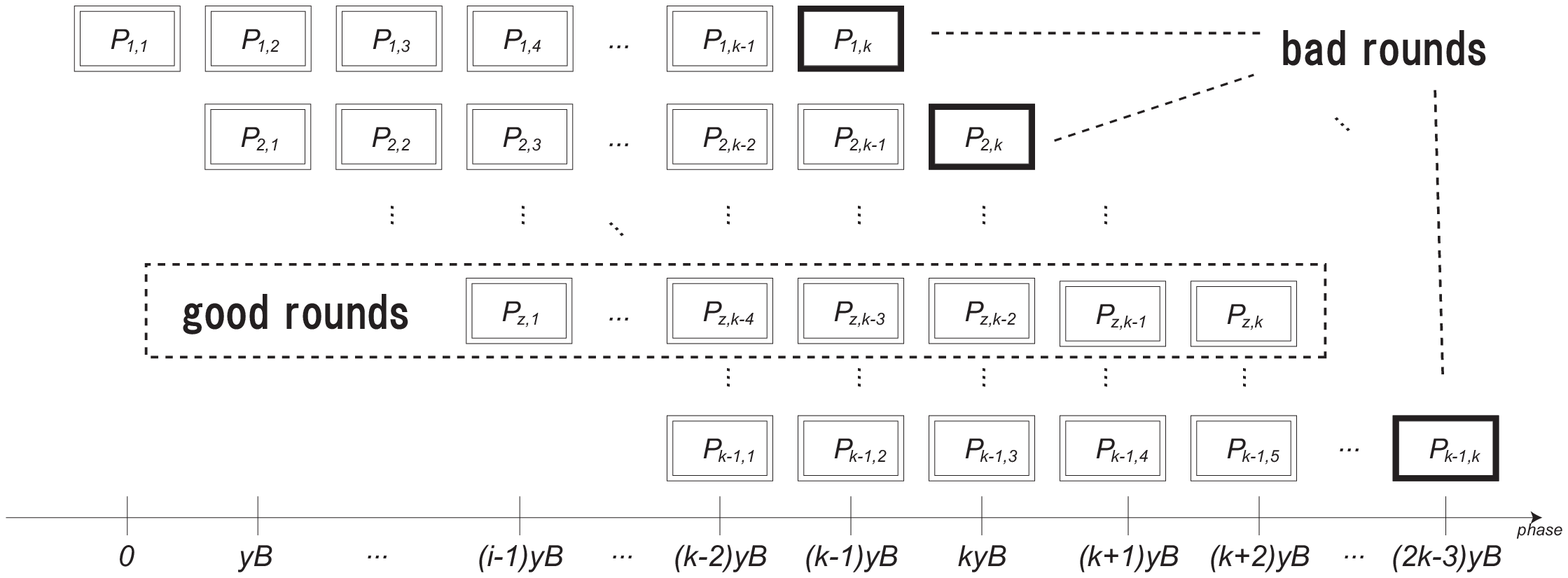}
	 \end{center}
	 \caption{
	 	Lower bound instance for randomized algorithms. 
		}
	\label{fig:LBR2}
\end{figure*}
\fi
\fi

\end{document}